\documentclass[a4paper,oneside,12pt]{article}

\usepackage{jheppub}

\usepackage[usenames,dvipsnames,table]{xcolor}

\usepackage{amsmath}
\usepackage{amsthm}
\usepackage{amsfonts}
\usepackage{braket}
\usepackage{bbold}
\usepackage{mathrsfs}  
\usepackage[mathscr]{eucal}
\usepackage{tikz}
\usepackage{enumerate}
\usepackage[normalem]{ulem}

\def\shadeB{\cellcolor{blue!5}}
\def\shadeR{\cellcolor{red!5}}

\usepackage{subcaption}
\usepackage{float}
\usepackage{afterpage}
\newtheorem{theorem}{Theorem}[section]
\newtheorem{corollary}[theorem]{Corollary}
\newtheorem{lemma}[theorem]{Lemma}
\newtheorem{definition}{Definition}[]

\definecolor{maxcolor}{rgb}{1,0.03,0}
\definecolor{dblue}{rgb}{0.25,0.03,0.8}
\definecolor{vecolor}{rgb}{0.7,0.3,0.9}

\def\c{\mathscr{C}}
\def\f{\mathscr{F}}

\def\om{\boldsymbol{\Omega}}
\DeclareMathOperator{\Tr}{Tr}
\newcommand\eqdef{\mathrel{\overset{\makebox[0pt]{\mbox{\normalfont\tiny\sffamily def}}}{=}}}

\newcommand{\qcf}[1]{Q_{#1}}
\newcommand{\bQ}{{\mathbf Q}}
\newcommand{\extr}[1]{\mathcal{E}_{#1}}
\def\esf{\omega}
\def\regA{\mathcal{A}}
\def\bulk{\mathcal{M}}
\def\univ{\mathcal{O}}
\newcommand{\si}{\mathscr{I}}
\newcommand{\sj}{\mathscr{J}}
\newcommand{\sk}{\mathscr{K}}

\newcommand{\arr}{\text{{\large $\mathbb{A}$}}}

\subheader{}
\title{\boldmath Holographic entropy relations}

\author[a]{Veronika  E. Hubeny,}
\author[a]{Mukund Rangamani,}
\author[b]{\! Massimiliano Rota}
\affiliation[a]{Center for Quantum Mathematics and Physics (QMAP)\\
Department of Physics, University of California, Davis, CA 95616 USA}
\affiliation[b]{Department of Physics, University of California, Santa Barbara, CA 93106, USA}

\emailAdd{veronika@physics.ucdavis.edu}
\emailAdd{mukund@physics.ucdavis.edu}
\emailAdd{mrota@physics.ucsb.edu}

\abstract{
We develop a framework for the derivation of new information theoretic quantities which are natural from a holographic perspective. We demonstrate the utility of our techniques by deriving the tripartite information (the quantity associated to monogamy of mutual information) using a set of  abstract arguments involving bulk extremal surfaces. Our arguments rely on formal manipulations of surfaces and not on local surgery or explicit computation of entropies through the holographic entanglement entropy prescriptions. As an application, we show how to derive a family of similar information quantities for an arbitrary number of parties. The present work establishes the foundation of a broader program that aims at the understanding of the entanglement structures of geometric states for an arbitrary number of parties. We stress that our method is completely democratic with respect to bulk geometries and is equally valid in static and dynamical situations. While rooted in holography, we expect that our construction will provide a useful characterization of multipartite correlations in quantum field theories. 
}

\begin{document} 

\maketitle
\flushbottom

\section{Introduction}
\label{sec:intro}

Quantum information theoretic constructs are playing an increasingly prominent role in theoretical physics. In part, this is thanks to the realization that entanglement can provide a useful diagnostic of interesting features of a quantum system and its dynamics. In the context of holographic dualities, entanglement seems to underlie the mechanism of the duality itself, encouraging the expectation that understanding the entanglement structure will elucidate the emergence of bulk spacetime \cite{VanRaamsdonk:2010pw,Maldacena:2013xja}. 

The most familiar, and in many ways natural, measure of entanglement is the entanglement entropy, defined as the von Neumann entropy of the reduced density matrix of a given subsystem. A particularly natural decomposition is delineated by a spatial region of the background (non-dynamical) spacetime on which the field theory lives.  In what follows we will consider such regions, bounded by smooth entangling surfaces, focusing thus on spatially ordered entanglement in relativistic QFTs.\footnote{  A-priori  the definition of entanglement entropy assumes a bi-partitioning of the Hilbert space. In relativistic quantum field theories one can alternately work directly with the local algebra of observables, thereby circumventing the notion of partitioning of the Hilbert space (which strictly-speaking does not apply). } However, the entanglement entropy associated with these regions has a UV divergence, whose leading part scales with the area of the entangling surface.  This suggests that the most physically meaningful quantities are not the entropies themselves, but rather linear combinations thereof, whose actual values can be finite despite the divergences in the building blocks.  Indeed, this expectation is ratified  within quantum information theory itself, even when dealing with finite quantum systems where such divergences do not arise.  In particular, \emph{information quantities}, which we define to be certain  linear combinations of entropies, have been used in many contexts both in classical and quantum information theory, e.g., to quantify and characterize correlations.\footnote{
Relative entropy is another quantity which is both finite and meaningful in QFTs. It however refers to properties of the state relative to another reference state. We will focus on quantities which capture the intrinsic information theoretic features of a state.
}  Such finite quantities tend to obey interesting bounds, whose saturation typically carries information theoretic significance.
 
The simplest example of such an information quantity is the mutual information between two disjoint subsystems, defined as the difference between the entanglement entropy of the combined system and the sum of the entanglement entropies of the individual subsystems, 
cf., Eq.\eqref{eq:midef} below.  Since this quantity characterizes the amount of correlation (both classical and quantum) between the two subsystems, it cannot be negative.  This powerful statement is known as subadditivity (SA) \cite{Araki:1970ba}, and is satisfied universally, for any quantum system in any state, and for any meaningful partition.  The saturation of this inequality then signifies the lack of correlation between the two subsystems.\footnote{ This typically does not happen in quantum field theories, but can occur in holographic systems if we focus on the leading contribution in the planar limit (large $N$), see \S\ref{sec:overview} \label{fn:misat}.} Similarly, the stronger statement of non-negativity of the conditional mutual information, known as strong subadditivity (SSA) \cite{Lieb:1973cp}, is satisfied by all classical probability distributions and quantum density matrices.  Since one can think of this property as the monotonicity of correlations under inclusion, its saturation implies a Markov property of the subsystems \cite{Hayden_2004}\cite{Casini:2017roe}.
   
However, not all interesting information quantities obey universal bounds: some may satisfy certain inequalities only in some particular circumstances.  There are numerous examples of such restricted relations, such as the non-negativity of the conditional entropy in the classical case, or the Ingleton inequality in the quantum context,  characterizing the set of 4-party stabilizer states \cite{Linden:2013aa}. Nevertheless, the restriction on the validity of such bounds does not diminish their utility.  In fact, such conditional inequalities are in a sense even more interesting than the universal ones, since they are more sensitive to distinctions between different classes of physical systems, and could potentially characterize the essence of this difference.
 
In what follows we will be particularly interested in understanding information quantities in the realm of holographic dualities  exemplified by the gauge/gravity correspondence. In this context, the Ryu-Takayanagi (RT) proposal \cite{Ryu:2006bv}, and its covariant generalization by HRT \cite{Hubeny:2007xt}, underpin the necessary link between field theory entanglement and geometry.\footnote{ For recent reviews of these developments we refer the reader to \cite{VanRaamsdonk:2016exw,Rangamani:2016dms,Harlow:2018fse}.} Of interest to us will be a sub-class of states in such holographic field theories, defined by states whose  dual description is in terms of a  smooth classical bulk geometry. We will henceforth refer to this subset as the set of \emph{geometric states}.\footnote{ 
Attempts to characterize geometric states using concepts from quantum error correction \cite{Almheiri:2014lwa,Dong:2016eik,Harlow:2016vwg} introduce a notion of code subspace of states, which at a heuristic level would coincide with our notion of geometric states, though one essential difference is that the code subspace additionally includes fluctuations of gravitational and matter fields about our geometric states.}
What we wish to  ascertain is which information quantities pertain specifically to such geometric states. 

Indeed, one might hope that the full set of information quantities could potentially usefully characterize this set by providing interesting necessary conditions for a field theory state to have a holographic dual corresponding to a classical geometry. Some examples, such as the monogamy of mutual information (MMI) are already well-known, cf.\ Eq.~\eqref{eq:mmidef} below.  This inequality, relating entanglement entropies for three subsystems, is the statement of non-positivity of tripartite information. It is guaranteed to hold when all correlations are purely quantum and therefore subject to the monogamy property, namely the statement that entanglement between two systems cannot be shared by a third one.  On the other hand, it is easy to construct quantum states which violate this inequality.  The remarkable fact that all geometric states of holographic field theories necessarily satisfy this inequality \cite{Hayden:2011ag} then hints at some residual quantumness of the state (despite the bulk geometry itself being described by classical dynamics),  perhaps even associated with bulk locality in this context \cite{Hubeny:2018bri}, whose precise meaning would be illuminating to understand.

While MMI is well known and easy to prove using the holographic entanglement entropy prescription,\footnote{
See \cite{Hayden:2011ag,Wall:2012uf} which generalize the RT-based proof of SSA \cite{Headrick:2007km}.    Two further (distinct) proofs of MMI based on  the `bit thread' reformulation of holographic entanglement entropy \cite{Freedman:2016zud} recently appeared in \cite{Cui:2018dyq,Hubeny:2018bri}.
}
the form of the corresponding information quantity, namely the tripartite information, has not been derived holographically from first principles. The situation is more dire for the  other inequalities explored in the context of the holographic entropy cone \cite{Bao:2015bfa}. For instance, these authors proved a set of inequalities for five subsystems (and a family of inequalities for a higher number of parties). However, in their present form these inequalities do not render themselves to a simple physical interpretation.  Nor is it fully known whether these inequalities hold for general time-dependent geometries, since the analysis of \cite{Bao:2015bfa} was restricted to time-reflection symmetric states where the RT prescription can be applied.\footnote{ One can argue that this restriction can be lifted in the case of two-dimensional conformal field theories with AdS$_3$ holographic duals \cite{Czech:2018aa}. We thank Xi Dong for discussions on this subject. \label{fn:3dproof}}
So far, we could at best realize that certain specific combinations of entanglement entropies have a definite sign, but hitherto we did not have a good way of deriving further inequalities, or, for the ones which are known, the corresponding information quantities directly.
This motivates a broader program for the search of information quantities; we lay out the general framework for such an exploration and extract some preliminary lessons in this paper.

For definiteness, we will focus on field theories in the planar limit with strong coupling (or large higher spin mass gap) which are expected to be dual to semi-classical Einstein-Hilbert gravity in the bulk.\footnote{ As will become clear in the course of our discussion, much of what we say will continue to apply in the planar limit even when higher curvature corrections are taken into account. In such situations we should use the general prescription given by \cite{Dong:2013qoa,Camps:2013zua} for computing the semi-classical field theory entanglement which involves a geometric functional built from intrinsic and extrinsic curvatures of a codimension-2 bulk surface.  However, the only crucial point for our analysis is the fact that there is a bulk surface which is associated with the field theory entanglement.  } In this limit, the  holographic entanglement entropy prescription of RT/HRT associates the entanglement entropy corresponding to certain spatial region, now thought of as living on the boundary of asymptotically-AdS bulk geometry, to the quarter-area of a certain (extremal) bulk surface homologous to that region.  The association of boundary regions with bulk surfaces will allow us to construct natural information quantities, by identifying classes of boundary region configurations for which these quantities vanish identically.  

Unpacking this statement and identifying a clear algorithm that directs the search for holographic information quantities will be the primary subject of this work.  To this end, we will develop a broad framework for deriving a specific set of information quantities. We will demonstrate the efficacy of our strategy in reproducing known results for a small number of subsystems. These ideas can be easily understood in the case of bipartite systems, and are powerful enough to allow for generalization to arbitrary number of subsystems.   Moreover, since our arguments are quite general, and do not rely on using the RT (as opposed to HRT) prescription, our method will apply to general time-dependent states of the field theory.

It is worth noting that the information quantities we construct using holography can in fact transcend the context of their origin, as is the case for the tripartite information.  Hence one can view our constructions as a new quantum information theoretic tool for obtaining novel information quantities which usefully characterize the entanglement structure of multipartite quantum systems.  It is therefore particularly useful here that our methodology applies equally well for any number of partitions, and is not restricted to static situations. 

In the holographic context our framework is complementary to the holographic entropy cone construction of \cite{Bao:2015bfa}, as we further explain below: rather than focusing on the entropy vectors, we work with entropy relations (corresponding to hyperplanes bounding the cone), which absolves us of having to consider multi-boundary wormholes or making cutoff-dependent statements.  We therefore view the hyperplanes (i.e., the entropy relations) as the more fundamental.  Correspondingly, this should allow us to make closer contact with the physical content of the associated information quantities.

The present work establishes the foundation for a boarder program that will be developed in a sequence of publications \cite{Hubeny:2018ab,Hubeny:2018aa}.  Future investigations will be organized according to three main complementary directions. As mentioned above, the primary goal is to find new information quantities of relevance in holography. We hope to do much more and in fact believe that one can extract the entire collection of \emph{primitive information quantities} (primitive here referring to an irreducible unit as we shall define later), in full generality, for any number of parties. A signature that this might indeed be possible comes from the main result of the present paper. As we will see, under two simple restrictions on the topology of allowed field theory regions and entangling surfaces, one can prove a general result (the {``${\bf I}_{\sf n}$-}Theorem"~\ref{thm:In}) which derives all possible primitive information quantities consistent with this restriction, for an arbitrary number of parties. The next step involves lifting these restrictions and correspondingly extracting more interesting information quantities. The fact that we are able to gain sufficient insight from restricted configurations of regions suggests that as we scan over more complex situations we will be able to uncover more structure.

The second direction aims at establishing a clearer connection to \textit{entropy inequalities} and the general structure of the holographic entropy cone. In the present work, we will show that in the particular case of three parties, the primitive information quantities emerging from our framework yield precisely the 3-party holographic entropy inequalities. This however is not the case for four or more parties, namely there are primitive information quantities which in general do not have a definite sign, even holographically. The plan is then to first construct a `sieve' that can be used to efficiently extract, for any number of parties, a set of `good' candidate inequalities from the set of all primitive information quantities. Then one would want to prove that these candidates are indeed new inequalities which hold for arbitrary dynamical spacetimes.

Finally, the underlying motivation of these efforts is to understand the implications of  holographic entropy inequalities for the entanglement structure of geometric states. As we will explain in due course, it is conceivable that not only the inequalities, but the full structure of the \textit{hyperplane arrangement} of the primitive information quantities, might play an important role. To this end having a framework that allows for efficient exploration of  this object is a necessary first step. We will already see a few glimpses of patterns towards the end of our discussion (see also earlier comments in \cite{Headrick:2013zda} and more recent work \cite{Cui:2018dyq}), but we hope to make  clear that there is more information to be mined here.

The plan of the paper is as follows. In \S\ref{sec:overview} we provide a first informal introduction to our framework, using intuition from the simple cases of two and three parties. In \S\ref{sec:formalization} we proceed with the formalization of the framework and present an overview of the logic that one can follow to derive primitive information quantities, at least in principle, for an arbitrary number of parties. The simple case of three parties is covered in detail in \S\ref{sec:three_parties}. In particular, we will see that the tripartite information falls out very naturally from this procedure, which one can then view as a holographic construction of the tripartite information, and consequently (using additional arguments to prove sign-definiteness) of MMI. The most far reaching result of the present work is the ${\bf I}_{\sf n}$-Theorem \ref{thm:In},  presented in \S\ref{sec:multipartite_information}. As mentioned above, we view it as the first step towards the systematic derivation of all primitive information quantities for an arbitrary number of parties. A more detailed presentation of the plan for future investigations, in relation to the findings presented here, and other interesting open questions are described in \S\ref{sec:discuss}. 

\section{Overview of the framework}
\label{sec:overview}

We begin with a non-technical overview of the framework which will be developed in the rest of the paper. In \S\ref{subsec:overview2} we consider the simplest case of bipartite systems and use it to review the notions of \textit{entropy space}, \textit{entropy vectors} and \textit{entropy cones}. The focus will be on the distinction between quantum mechanics of finite dimensional Hilbert spaces, where entropies are finite, and quantum field theory, where entropies are generically infinite. We will show how this crucial difference suggests that in quantum field theory it is preferable to attribute a fundamental role to \textit{entropy relations}, rather than to \textit{entropy values}. Furthermore, we will explain how for holographic states, the RT/HRT prescription naturally identifies a particular class of such relations. In \S\ref{subsec:overview3} we will introduce the generalization to an arbitrary number of parties and use the intuition from the case of tripartite systems to motivate the definition of the \textit{primitive information quantities} that we want to derive.

\subsection{Entropy constructs for bipartite systems}
\label{subsec:overview2}

To understand the form of information quantities we are after, it is useful to begin our discussion in the familiar context of bipartite systems. Even though our primary interest will be in holographic field theories, it will be helpful to understand the constructs both in simple finite dimensional quantum systems and in a general quantum field theory, which we will therefore do before turning to the aspects that are more naturally suggested by holography.

\subsubsection{Case 1: Finite quantum systems}
\label{subsec:overview2a}

Consider a bipartite Hilbert space $\mathcal{H}_{\mathcal{A}}\otimes\mathcal{H}_{\mathcal{B}}$ and a density matrix $\rho_{\mathcal{AB}}$ acting on it. Starting from $\rho_{\mathcal{AB}}$ we can construct the reduced density matrices $\rho_{\mathcal{A}}$ and $\rho_{\mathcal{B}}$ by tracing out the subsystems $\mathcal{B}$ and $\mathcal{A}$ respectively. For each of these three density matrices, the original one and the two marginals, we can then compute the von Neumann entropy $S(\rho)$. We collect these entropies into a vector ${\bf S}(\rho_{\mathcal{AB}})=\{S_\mathcal{A},S_{\mathcal{B}},S_{\mathcal{AB}}\}\in\mathbb{R}_+^3$ which we will call an \textit{entropy vector}. The space $\mathbb{R}_+^3$ where these vectors live will be referred to as \textit{entropy space}. The collection of all possible entropy vectors, for all possible density matrices and Hilbert spaces, has a complicated structure, but its topological closure is a convex cone, known as the \textit{quantum entropy cone} \cite{Pippenger:2003aa}. 

Furthermore, in the case of bipartite systems, this cone has a remarkably simple structure. It is a polyhedral cone corresponding to the intersection of the half-spaces specified by three inequalities \cite{Pippenger:2003aa}, namely \textit{subadditivity} (SA) and two permutations of the \textit{Araki-Lieb inequality} (AL),
\begin{align}
\text{SA:}\qquad &S_\mathcal{A}+S_\mathcal{B}\geq S_{\mathcal{AB}}
\label{eq:sa}\\
\text{AL:}\qquad &S_\mathcal{A}+S_\mathcal{AB}\geq S_{\mathcal{B}}\label{eq:al}\\
							&S_\mathcal{B}+S_\mathcal{AB}\geq S_{\mathcal{A}}\nonumber
\label{eq:bipartite_inequalities}
\end{align}
If we think of $\{S_\mathcal{A},S_{\mathcal{B}},S_{\mathcal{AB}}\}$ as variables, the equations associated to the saturation of these inequalities can be interpreted as planes in entropy space. This geometric representation will be very convenient in the following. We remind the reader that although formally different, and therefore associated to different planes in entropy space, SA and AL are in fact physically equivalent, since each inequality implies the other. To see that this is the case one can start from SA, introduce the purification $\mathcal{O}$ of the system $\mathcal{AB}$, replace $S_\mathcal{B}$ and $S_\mathcal{AB}$ with the entropies of the complementary subsystems, and relabel $\mathcal{O}\rightarrow\mathcal{B}$. This kind of relation between different inequalities will be ubiquitous also in the multipartite generalization and we will say that one inequality is mapped to the other under the \textit{purification symmetry}. 

Any polyhedral cone has an equivalent description in terms of a finite number of generators called \textit{extremal rays}.\footnote{ By definition, these are one-dimensional subspaces of the entropy space -- they are simply rays emanating from the origin which generate the polyhedral cone.  In particular, any vector within the cone can be obtained as a conical combination of the extremal rays.} For the bipartite quantum entropy cone, the extremal rays are  generated by the following vectors:
\begin{equation}
{\bf S}^{\text{ext}}_1=\{1,1,0\},\; \quad {\bf S}^{\text{ext}}_2=\{1,0,1\},\; \quad {\bf S}^{\text{ext}}_3=\{0,1,1\} \,.
\end{equation}
The entropies of the first vector are trivially realized by any pure state $\ket{\psi}_{\mathcal{AB}}$. More generally, we can consider a state $\ket{\psi}_{\mathcal{AB}}\otimes\ket{\phi}_\mathcal{O}$ and realize the other two vectors by simply relabeling the subsystems, respectively as $\ket{\psi}_{\mathcal{AO}}\otimes\ket{\phi}_\mathcal{B}$ and $\ket{\psi}_{\mathcal{OB}}\otimes\ket{\phi}_\mathcal{A}$. Notice that since these states realize the vectors which generate the extremal rays, each of them simultaneously saturates two of the three inequalities which specify the cone. This in fact must be the case, since the extremal rays lie precisely at the intersections of the planes corresponding to the saturation of the inequalities which specify the cone. The bipartite quantum entropy cone and its extremal rays are shown in Fig.~\ref{fig:bipartite_cone}.

\begin{figure}[]
\centering
\includegraphics[width=0.5\textwidth]{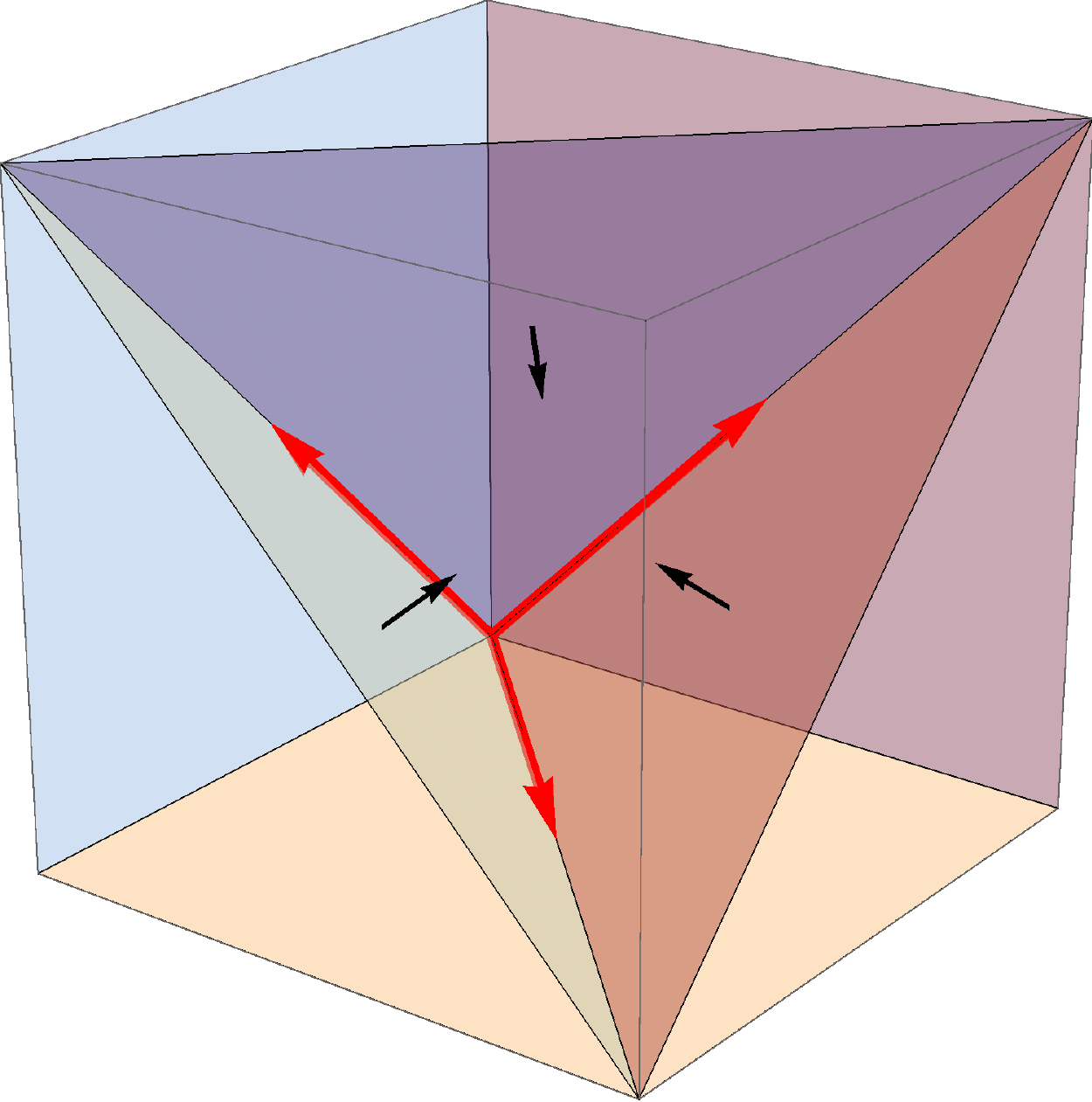}
\put(8,60) {\makebox(0,0) {\small{$S_\mathcal{A}$}}}
\put(-220,50) {\makebox(0,0) {\small{$S_\mathcal{B}$}}}
\put(-115,228) {\makebox(0,0) {\small{$S_\mathcal{AB}$}}}
\caption{The bipartite quantum entropy cone (delimited by the diagonal planes, with short black arrows indicating the direction prescribed by SA and AL) and its extremal rays (long red arrows) embedded in entropy space $\mathbb{R}_+^3$.}
\label{fig:bipartite_cone}
\end{figure}

It is important to stress the difference between the set defined as the collection of entropy vectors realized by all possible bipartite quantum states, and the entropy cone specified by the inequalities, which is its topological closure. Although, as we showed, it is straightforward to construct quantum states that realize the vectors which generate all the extremal rays of the bipartite entropy cone, it is not true that an arbitrary conical combination (viz., a linear combination with non-negative coefficients) of these vectors can be \textit{exactly} realized by a quantum state. In particular it is important to notice that the \textit{holographic entropy cone} \citep{Bao:2015bfa} was defined as a collection of finite\footnote{ 
While the authors of \citep{Bao:2015bfa} were interested in holographic field theories where entanglement is plagued by UV divergences, finiteness was achieved by considering states in the tensor product of a set of holographic field theories. Geometrically these states correspond to multi-boundary wormhole geometries, and by restricting the allowed subsystems to be entire boundaries, one has finite entanglement (per unit spatial volume).} entropy vectors, and not as a region of entropy space bounded by a set of inequalities.\footnote{ 
The complications of the quantum mechanical case do not arise in the holographic context, where the collection of entropy vectors automatically coincides with its topological closure. Specifically, if one can construct geometries that realize the generators of the extremal rays, it is guaranteed that any other vector within the cone can also be realized by some geometry, see \citep{Bao:2015bfa} for more details.} The latter perspective will instead characterize our approach.

\subsubsection{Case 2: Quantum field theory}
\label{subsec:overview2b}

To explain why it is preferable to delineate regions in entropy space defined by inequalities, it will be useful to first extend the previous construction to a quantum field theory. On a Cauchy slice of the background spacetime the field theory lives on, consider a configuration $\c$ of two subsystems $\mathcal{A}$ and $\mathcal{B}$. We can construct the entropy vector associated to the corresponding density matrix\footnote{
Of course,  the reduced density matrix depends both on the configuration as well as on the total state.  However,  in the interest of avoiding unnecessary clutter of notation, and to indicate what will be the more crucial aspect in what follows, we will explicitly write only the configuration dependence, leaving the state dependence implied.
}
as in the quantum mechanical case. However, since in quantum field theory the von Neumann entropy is generically infinite, the interpretation of this vector is unclear. One possibility is to fix a regulator $\epsilon$ and consider the entropy vector ${\bf S}_\epsilon(\c)$, with all entropies finite by construction.
 However, the values of the various entropies now have no intrinsic physical meaning, since they depend on the regulator.\footnote{
In fact, the regulator need not be a constant value over all space (especially in conformal field theories where there is no intrinsic meaning to a scale), so ${\bf S}_\epsilon(\c)$ is determined not just by a parameter $\epsilon$ but by the function $\epsilon(\vec{x})$.
} In particular, by locally varying the regulator, one can obtain an infinite collection of entropy vectors ${\bf S}_\epsilon(\c)$ which will in general not be proportional to each other. Therefore in quantum field theory one is forced to associate a configuration of subsystems to an infinite collection of finite entropy vectors, rather than to a single one, as was the case for finite dimensional Hilbert spaces. Furthermore, this collection of finite entropy vectors will generally span the whole entropy space, thereby preventing us from identifying a particular location associated to the configuration $\c$, unlike the quantum mechanical case.

However, in some particular circumstances, the unregulated entropies satisfy some cutoff-independent relation. This is the case when the individual divergences cancel in a universal way, which only becomes apparent as we remove the cutoff.
Consider for example
 the \textit{mutual information}
\begin{equation}
{\bf I}_2(\mathcal{A}:\mathcal{B})\equiv S_\mathcal{A}+S_\mathcal{B}-S_{\mathcal{AB}} \,,
\label{eq:midef}
\end{equation}
and for simplicity take a pair of intervals $\mathcal{A}$ and $\mathcal{B}$ of fixed sizes $\ell_\mathcal{A}$ and $\ell_\mathcal{B}$ on a time slice of a (1+1)-dimensional CFT. At finite cutoff $\epsilon$ the vectors  ${\bf S}_\epsilon(\c)$ will span the full entropy space $\mathbb{R}_+^3$. Let us however examine what happens as we take the cut-off $\epsilon \to 0$.  While each term in Eq.~\eqref{eq:midef} diverges, these divergences cancel so as to render ${\bf I}_2(\mathcal{A}:\mathcal{B})$ not only finite (for separation $x>0$ between the two intervals), but cut-off independent. In particular, this finite value ${\bf I}_2({\bf S}(\c))$ has physical significance since it is independent from the regulator scheme.
This means that although the (unregulated) entropy vector ${\bf S}(\c)$ is divergent, we can think of it as being \textit{localized} on a hyperplane defined by the following relation
\begin{equation}
S_\mathcal{A}+S_\mathcal{B}-S_{\mathcal{AB}}={\bf I}_2({\bf S}(\c))\,,
\end{equation}
where we now think of the entropies $\{S_\mathcal{A},S_\mathcal{B},S_\mathcal{AB}\}$ as variables in entropy space. 

As we modify the configuration $\c$, the value of ${\bf I}_2({\bf S}(\c))$ will change and the vector ${\bf S}(\c)$ will be localized on different hyperplanes. In particular, one may wonder if there exists a particular choice of configuration such that this hyperplane corresponds to one of the facets of the quantum entropy cone, specifically 
\begin{equation}
S_\mathcal{A}+S_\mathcal{B}-S_{\mathcal{AB}}=0\,.
\end{equation}
However, in general this is not possible in quantum field theory. In fact, if we increase the separation $x$ between the two intervals, the mutual information decays\footnote{ For conformally invariant theories in general dimensions, the mutual information falls off as a power-law  and exponent set by the minimum sum of scaling dimensions of two operators whose operator product contains the vacuum \cite{Cardy:2013nua}. It is natural to expect that for gapped systems one would see an exponential decay.} but it never vanishes exactly because it is lower bounded by correlation functions of operators supported on the two intervals \cite{Wolf:2007aa}. The situation is vastly improved for geometric states in holographic theories, to which we turn next.

\subsubsection{Case 3: Holographic field theories}
\label{subsec:overview2b}

Thus far our discussion has not used any information about the existence of gravitational duals of the field theory. Here we have an additional parameter at hand to dial, viz., the central charge set by $N$. We will now recall the special features that occur when $N\to \infty$ and motivate therefrom a set of quantities that will be explicitly regulator-independent.  While the logic for choosing the information quantities we focus on is predicated on holography,  as apparent from the above discussion, the quantities themselves will be cut-off independent and therefore should be of interest in quantum field theories more generally. 

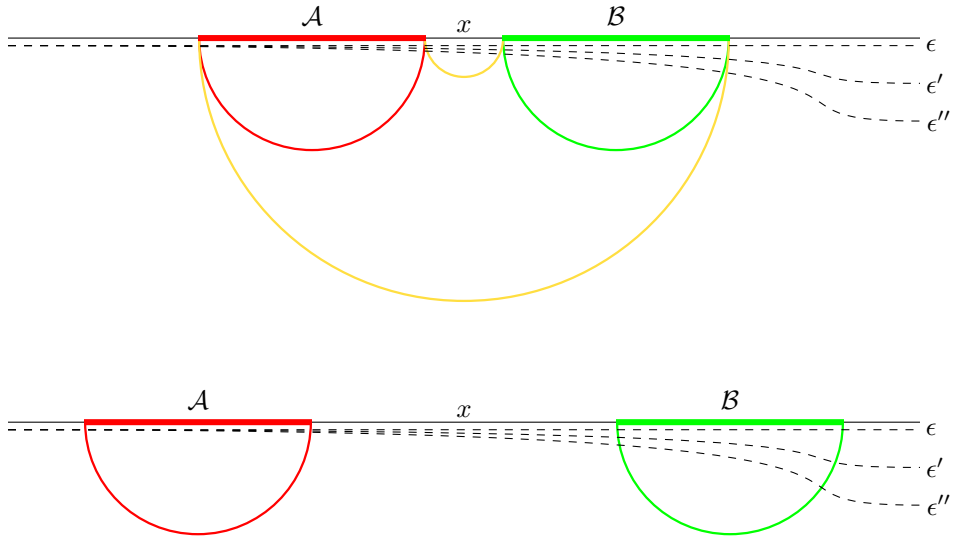
\begin{figure}[t!]
\centering
\begin{subfigure}{1\textwidth}
\centering
\begin{tikzpicture}
\draw (-6,0) -- (6,0);
\draw[line width=0.3mm, red] (-3.485,0) arc (-180:0:1.485);
\draw[line width=0.3mm, green] (0.515,0) arc (-180:0:1.485);
\draw[line width=0.3mm, Goldenrod] (-0.515,0) arc (-180:0:0.515);
\draw[line width=0.3mm, Goldenrod] (-3.485,0) arc (-180:0:3.485);
\draw [line width=0.8mm, red] (-3.5,0) -- (-0.5,0); 
\draw [line width=0.8mm, green] (0.5,0) -- (3.5,0); 
\draw[dashed] (-6,-0.1) -- (6,-0.1);
\draw[dashed] (-6,-0.1) .. controls (8,-0.1) and (3,-0.6) ..  (6,-0.6);
\draw[dashed] (-6,-0.1) .. controls (8,-0.1) and (3,-1.1) .. (6,-1.1);
\node at (-2,0.3) {\footnotesize{$\mathcal{A}$}};
\node at (2,0.3) {\footnotesize{$\mathcal{B}$}};
\node at (0,0.15) {\footnotesize{$x$}};
\node at (6.15,-0.1) {\footnotesize{$\epsilon$}};
\node at (6.2,-0.6) {\footnotesize{$\epsilon'$}};
\node at (6.25,-1.1) {\footnotesize{$\epsilon''$}};
\end{tikzpicture}
\end{subfigure}

\vspace{1cm}

\begin{subfigure}{1\textwidth}
\centering
\begin{tikzpicture}
\draw (-6,0) -- (6,0);
\draw[line width=0.3mm, red] (-4.985,0) arc (-180:0:1.485);
\draw[line width=0.3mm, green] (2.015,0) arc (-180:0:1.485);
\draw [line width=0.8mm, red] (-5,0) -- (-2,0); 
\draw [line width=0.8mm, green] (2,0) -- (5,0); 
\draw[dashed] (-6,-0.1) -- (6,-0.1);
\draw[dashed] (-6,-0.1) .. controls (8,-0.1) and (3,-0.6) ..  (6,-0.6);
\draw[dashed] (-6,-0.1) .. controls (8,-0.1) and (3,-1.1) .. (6,-1.1);
\node at (-3.5,0.3) {\footnotesize{$\mathcal{A}$}};
\node at (3.5,0.3) {\footnotesize{$\mathcal{B}$}};
\node at (0,0.15) {\footnotesize{$x$}};
\node at (6.15,-0.1) {\footnotesize{$\epsilon$}};
\node at (6.2,-0.6) {\footnotesize{$\epsilon'$}};
\node at (6.25,-1.1) {\footnotesize{$\epsilon''$}};
\end{tikzpicture}
\end{subfigure}
\caption{A configuration (top) which does not saturate SA. Since the entropies are computed by three different surfaces, the three entropy vectors ${\bf S}_\epsilon(\c)$, ${\bf S}_{\epsilon'}(\c)$ and ${\bf S}_{\epsilon''}(\c)$ can all be made independent from each other by an appropriate choice of regulators, and therefore span the whole $\mathbb{R}_+^3$. When the configuration saturates SA (bottom), the relation ${\bf I}_2(\mathcal{A}:\mathcal{B})=0$ holds for any choice of cut-off and the resulting entropy vectors only span a two-dimensional plane.}
\label{fig:cut-off_independence}
\end{figure}

In situations where the quantum field theory is holographic and the state under consideration is dual to a classical geometry, the Ryu-Takayanagi formula implies that at leading order in $N$, the mutual information  can vanish exactly and subadditivity is \textit{saturated} \cite{Headrick:2010zt}. This occurs when the bulk minimal surface whose area computes the entropy $S_\mathcal{AB}$ is the union of the surfaces which compute the entropy of $S_\mathcal{A}$ and $S_\mathcal{B}$ individually. Therefore it is clear that while the values of the entropies depend on the cut-off, the saturation of subadditivity is achieved independently from the choice of regulator, see Fig.~\ref{fig:cut-off_independence}. In the following, the crucial fact for us will be that in this particular case, the (infinite) collection of entropy vectors associated to the configuration only spans the plane associated to the saturation of subadditivity, and not the whole entropy space.

Similarly, it is possible to find configurations that saturate the inequalities corresponding to the two other facets of the bipartite quantum entropy cone, which are just the two permutations of the AL inequality. An example is shown in Fig.~\ref{fig:cut-off_independence2}, again for a (1+1)-dimensional CFT. 

Finally, following the quantum mechanical construction described above, we can also identify in field theory the configurations that realize the extremal rays of the cone. It is sufficient to consider the state $\ket{0}_{\text{CFT}_1}\otimes\ket{0}_{\text{CFT}_2}$, where $\ket{0}$ is the vacuum in the two CFTs, and consider an arbitrary bipartition of one of the two factors. As in the quantum mechanical case, one gets all the extremal rays by different choices of labels for the subsystems. Furthermore, notice that since in this case there is just a single bulk surface which computes the entropies, the collection of finite entropy vectors obtained by different choices of the regulator now only spans a one dimensional subspace, i.e., the extremal ray.

The main lesson we wish to draw is that while entropy vectors are in general ambiguous in quantum field theory, and a generic configuration of subsystems is associated to an infinite collection of them, there exists specific entropy relations which holographically  hold exactly (at leading order in $N$), independently from the choice of a regulator. The most explicit manifestation of this fact is that the collection of regulated entropy vectors only span a lower dimensional subspace, instead of the whole entropy space. As will be more clear later, this is possible only because of the particular structure of the information quantities that we considered. Our strategy in what follows will be to turn this argument around, and use regulator independence as a constraint in searching for new multipartite information theoretic quantities which are natural from a holographic perspective (and thereby potentially more generally).

\begin{figure}[t!]
\centering
\begin{tikzpicture}
\draw (-6,0) -- (6,0);
\draw[line width=0.3mm, red] (-2.985,0) arc (-180:0:2.985);
\draw[line width=0.3mm, green] (-0.486,0) arc (-180:0:0.486);
\draw[line width=0.3mm, red] (-0.5165,0) arc (-180:0:0.5165);
\draw [line width=0.8mm, red] (-3,-0) -- (-0.5,0); 
\draw [line width=0.8mm, red] (0.5,0) -- (3,0); 
\draw [line width=0.8mm, green] (-0.5,0) -- (0.5,0);
\draw[dashed] (-6,-0.1) -- (6,-0.1);
\draw[dashed] (-6,-0.1) .. controls (4,-0.1) and (0.6,-0.6) ..  (6,-0.6);
\draw[dashed] (-6,-0.1) .. controls (4,-0.1) and (0.6,-1.1) .. (6,-1.1);
\node at (-1.75,0.3) {\footnotesize{$\mathcal{A}$}};
\node at (1.75,0.3) {\footnotesize{$\mathcal{A}$}};
\node at (0,0.3) {\footnotesize{$\mathcal{B}$}};
\node at (6.15,-0.1) {\footnotesize{$\epsilon$}};
\node at (6.2,-0.6) {\footnotesize{$\epsilon'$}};
\node at (6.25,-1.1) {\footnotesize{$\epsilon''$}};
\end{tikzpicture}
\caption{A choice of subsystems which saturates the AL inequality 
(namely $S_\mathcal{B}+S_\mathcal{AB}= S_{\mathcal{A}}$) 
independently from the regulator.}
\label{fig:cut-off_independence2}
\end{figure}
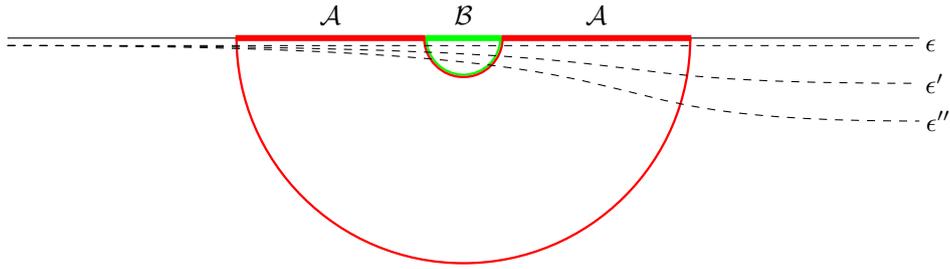

\subsection{Conditions for fundamental entropy relations}
\label{subsec:overview3}

The quantum mechanical definitions of entropy vectors and entropy space, introduced in the previous sections for two parties, naturally generalize to the multipartite setting, where the Hilbert space has $\sf{N}$ factors $\mathcal{H}_{\regA_1}\otimes \mathcal{H}_{\regA_2}\otimes ...\otimes\mathcal{H}_{\regA_{\sf N}}$. Entropy vectors are now defined as
\begin{equation}
{\bf S}(\rho_{\regA_1\regA_2...\regA_{\sf N}})=\{S_{\regA_1},S_{\regA_2},...,S_{\regA_{\sf N}},S_{\regA_1\regA_2},...,S_{\regA_1\regA_2...\regA_{\sf N}}\}\in\mathbb{R}_+^{\sf D} \ ,
\end{equation}
so the entropy space is $\mathbb{R}_+^{\sf D}$, with ${\sf D}=2^{\sf N}-1$. To indicate the subsystems of interest, in the rest of this paper we will use the notation $\{\regA_1,\regA_2,...,\regA_{\sf N}\}$ when ${\sf N}$ is unspecified, and switch to $\{\mathcal{A},\mathcal{B},\mathcal{C},...\}$ when we work with fixed small values of ${\sf N}$. 

For finite dimensional Hilbert spaces, one can again consider the collection of all vectors realized by all density matrices. It can be proved 
\citep{Pippenger:2003aa} that the topological closure of this set is a convex cone for any ${\sf N}$. Very little is known about this cone for arbitrary ${\sf N}$ \cite{Cadney:2011aa}. However, it was proven in \cite{Hayden:2016cfa} that the holographic entropy cone of \citep{Bao:2015bfa} is a proper subset of the quantum entropy cone for any ${\sf N}\geq 3$. In the following, by \textit{${\sf N}$-partite entropy cone} we will mean the region of entropy space bounded by all the ${\sf N}$-party inequalities 
(yet to be determined)
which are satisfied by entropies computed via the HRT formula.\footnote{ For further comments about the relation between the quantum and holographic entropy cone see also \cite{Marolf:2017shp,Rota:2017ubr}.}

We will be interested in \textit{information quantities} $\bQ$ of the general form
\begin{equation}
\bQ({\bf S})=\sum_{\si=1}^{\sf D}\; \qcf{\si} \, S_{\si},\qquad \qcf{\si} \in\mathbb{R} \ ,
\label{eq:info_quantity}
\end{equation}
for different values of ${\sf N}$, where the summation index $\si$ invokes all combinations of subsystems $\regA_\ell$ (see \S\ref{sec:formalization} for a precise definition).
It will again be convenient to have a geometric representation of these quantities. If we think of the components $S_\si$ of an entropy vector as variables, an \textit{entropy relation} of the form $\bQ({\bf S})=0$ represents a codimension-one hyperplane in entropy space, specified by the coefficients $\{\qcf{\si}\}$.
 An entropy vector ${\bf S}$ (being a finite vector ${\bf S}(\rho_{\regA_1\regA_2...\regA_{\sf N}})$ in quantum mechanics or a regulated vector ${\bf S}_\epsilon(\c_{\sf N})$ in quantum field theory) belongs to this hyperplane if it satisfies the equation $\bQ({\bf S})=0$. We will henceforth think of any information quantity $\bQ$ as being associated to the corresponding hyperplane.\footnote{ Following the discussion in \S\ref{subsec:overview2b}, one could more generally associate to an information quantity ${\bf Q}$, an entire family of parallel hyperplanes. However, the fact that such a quantity \textit{can} vanish, will be crucial in our construction and therefore motivates our choice. We will comment again on the more general identification of information quantities and families of (rather than single) hyperplanes in \S\ref{sec:discuss}.}

In the preceding discussion we have seen that for geometric states in holographic theories (at leading order in $N$), it is possible to make sense of this relation independently of the cut-off, at least for some specific quantities $\bQ$. This motivates our first definition of the information quantities of interest.
\begin{definition}
An entropic information quantity of the form \eqref{eq:info_quantity} will be said to be \emph{faithful} if there exists at least one geometric  state, and at least one (sufficiently generic\footnote{ 
\label{fn:generiticity}
We define a configuration $\c_{\sf N}$ to be considered sufficiently generic if the bulk extremal surface that computes the entropy of any subsystem varies continuously under continuous deformations of $\c_{\sf N}$, or equivalently if to each entangling surface there exists a unique (globally minimal) extremal surface (which in particular disallows configurations fine-tuned to phase transitions of minimal surfaces).  Moreover, we require that it has at least one connected component anchored to the boundary. The special configurations of \citep{Bao:2015bfa}, where all bulk extremal surfaces are compact, are therefore excluded.}) configuration of subsystems $\c_{\sf N}$, such that to leading order in $N$, 
\\
$\bQ({\bf S}_\epsilon(\c_{\sf N}))=0$ independently from the cut-off $\epsilon$.
\label{def:faithful}
\end{definition}

This definition is also motivated by a second, independent, argument. In the following we will be mostly interested in finding a list of information quantities which are good candidates for new holographic entropy inequalities. However, it is straightforward to generate infinitely many trivial inequalities which are necessarily satisfied. In fact, one can associate such a trivial inequality to any information quantity associated to a hyperplane that intersects the cone only at the origin; for a pictorial representation see $\bQ^{(1)}$ in Fig.~\ref{fig:redundant_inequalities}. 
Requiring that our information quantities be faithful then manifestly removes all such inequalities from our search.  

However, by itself, this requirement is still very weak, as we argue momentarily. We will refer to a combination of entropies as \textit{balanced} if for each of the ${\sf N}$ subsystems $\regA_\ell$ we have
\begin{equation}
 \sum_{\si\;\text{s.t.}\;\ell\in\si} \qcf{\si} =0 \,,
\end{equation}
where the sum is over all collections of subsystems which include $\regA_\ell$.  
In other words, the occurrence of each $\regA_\ell$ by itself in \eqref{eq:info_quantity}   (ignoring all the others) would cancel out. 
 According to this definition, it then follows that any 
\textit{balanced} $\bQ$ is \textit{faithful}. As an explicit example of a configuration $\c_{\sf N}$ which implements a balanced information quantity, consider ${\sf N}$ intervals of the same length on a time slice of a geometric state  in a (1+1)-dimensional holographic CFT, where all intervals are sufficiently separated form each other, such that the mutual information between any of them and the union of all the others vanish.

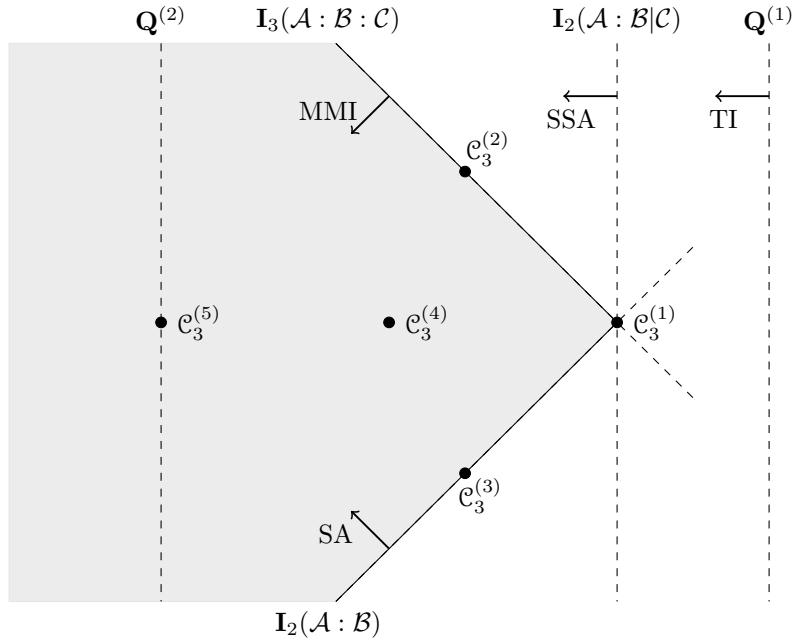
\begin{figure}[t!]
\centering
\begin{tikzpicture}
\filldraw[fill=gray!15,draw=gray!20] (-3.7,3.7) -- (0,0) -- (-3.7,-3.7) -- (-8,-3.7) -- (-8,3.7) -- cycle;
\draw[dashed] (-3.7,3.7) -- (1,-1);
\draw[dashed] (-3.7,-3.7) -- (1,1);
\draw[dashed] (0,3.7) -- (0,-3.7);
\draw[dashed] (2,3.7) -- (2,-3.7);
\draw[dashed] (-6,3.7) -- (-6,-3.7);
\draw (-3.7,3.7) -- (0,0);
\draw (-3.7,-3.7) -- (0,0);
\draw[->, line width=0.25mm] (-3,3) -- (-3.5,2.5);
\draw[->, line width=0.25mm] (-3,-3) -- (-3.5,-2.5);
\draw[->, line width=0.25mm] (0,3) -- (-0.71,3);
\draw[->, line width=0.25mm] (2,3) -- (1.29,3);
\node at (-3.8,4) {\footnotesize{${\bf I}_3(\mathcal{A}:\mathcal{B}:\mathcal{C})$}};
\node at (-3.8,-4) {\footnotesize{${\bf I}_2(\mathcal{A}:\mathcal{B})$}};
\node at (0,4) {\footnotesize{${\bf I}_2(\mathcal{A}:\mathcal{B}|\mathcal{C})$}};
\node at (2,4) {\footnotesize{$\bQ^{(1)}$}};
\node at (-6,4) {\footnotesize{$\bQ^{(2)}$}};
\filldraw [black] (0,0) circle (2pt);
\filldraw [black] (-2,2) circle (2pt);
\filldraw [black] (-2,-2) circle (2pt);
\filldraw [black] (-3,0) circle (2pt);
\filldraw [black] (-6,0) circle (2pt);
\node at (0.5,0) {\footnotesize{$\c_3^{(1)}$}};
\node at (-1.7,2.3) {\footnotesize{$\c_3^{(2)}$}};
\node at (-1.8,-2.3) {\footnotesize{$\c_3^{(3)}$}};
\node at (-2.5,0) {\footnotesize{$\c_3^{(4)}$}};
\node at (-5.5,0) {\footnotesize{$\c_3^{(5)}$}};
\node at (-3.8,2.8) {\footnotesize{MMI}};
\node at (-3.7,-2.8) {\footnotesize{SA}};
\node at (-0.6,2.7) {\footnotesize{SSA}};
\node at (1.4,2.7) {\footnotesize{TI}};
\end{tikzpicture}
\caption{A schematic representation of the transverse section of the three parties entropy cone in $\mathbb{R}_+^7$. The interior  of the cone (shaded) is bounded by the hyperplanes (represented by solid lines) associated to ${\bf I}_2(\mathcal{A}:\mathcal{B})$ and ${\bf I}_3(\mathcal{A}:\mathcal{B}:\mathcal{C})$ (the arrows show the corresponding inequalities, respectively SA and MMI). $\bQ^{(1)}$ is an information quantity which is not faithful (i.e., it does not satisfy Definition~\ref{def:faithful}) and therefore corresponds to a trivial inequality (TI). 
SSA is redundant and can only be saturated by a configuration ($\c^{(1)}_3$) that simultaneously saturates SA and MMI. $\c^{(2)}_3$ and $\c^{(3)}_3$ are configurations that only individually saturate MMI and SA, but none of the other inequalities corresponding to the facets of the cone, and therefore generate ${\bf I}_2(\mathcal{A}:\mathcal{B})$ and ${\bf I}_3(\mathcal{A}:\mathcal{B}:\mathcal{C})$ respectively. $\c^{(4)}_3$ does not saturate any inequality and its entropy vectors span the whole cone. $\bQ^{(2)}$ is a hypothetical fundamental information quantity, generated by the configuration $\c^{(5)}_3$, which does not correspond to a new inequality.
}
\label{fig:redundant_inequalities}
\end{figure}

To introduce the second and more stringent condition on the information quantities of interest, it is useful to look in more detail at the particular case of three subsystems. In this case, as for bipartite systems, the quantum entropy cone is again polyhedral. Some of the inequalities that specify this cone are inherited from the bipartite case (see \S\ref{sec:three_parties} for more details), while among the new ones are the possible permutations of \textit{strong subadditivity} (SSA) and \textit{weak monotonicity} (WM)
\begin{align}
\text{SSA:}\qquad &S_\mathcal{AC}+S_\mathcal{BC}\geq S_{\mathcal{C}}+S_{\mathcal{ABC}}  \\
\text{WM:}\qquad &S_\mathcal{AC}+S_\mathcal{BC}\geq S_{\mathcal{A}}+S_{\mathcal{B}} 
\end{align}
As for SA and AL, WM and SSA are equivalent to each other under the purification symmetry.

Furthermore, for holographic states, to leading order in $N$, an additional inequality, proven in \cite{Hayden:2011ag},\footnote{ While the proof of \cite{Hayden:2011ag} was limited to the time-reversal symmetric states, the extension to dynamical setting was established in  \cite{Wall:2012uf}.} is the \textit{monogamy of mutual information} (MMI) mentioned in the introduction,
\begin{equation}
S_\mathcal{AB}+S_\mathcal{BC}+S_\mathcal{AC}\geq S_\mathcal{A}+S_\mathcal{B}+S_\mathcal{C}+S_\mathcal{ABC} \,.
\label{eq:mmidef}
\end{equation}
These two inequalities, namely SSA and MMI, are associated to two information quantities known as the \textit{conditional mutual information} and the \textit{tripartite information}, respectively, 
\begin{align}
{\bf I}_2(\mathcal{A}:\mathcal{B}|\mathcal{C}) &\equiv S_\mathcal{AC}+S_\mathcal{BC}-S_{\mathcal{C}}-S_{\mathcal{ABC}}\\
{\bf I}_3(\mathcal{A}:\mathcal{B}:\mathcal{C}) &\equiv S_{\mathcal{A}}+S_{\mathcal{B}}+S_{\mathcal{C}}-S_\mathcal{AB}-S_\mathcal{AC}-S_\mathcal{BC}+S_\mathcal{ABC}
\end{align}
In terms of these quantities, SSA  can be rephrased as the statement that the conditional mutual information is always non-negative, and similarly MMI says that the tripartite information is non-positive.\footnote{ The notation for conditional mutual information is chosen to emphasize the similarity to the tripartite information and other generalizations which we will encounter in the course of our discussion (although one might argue that $-{\bf I}_3$ is a more natural object than ${\bf I}_3$, and more analogous to ${\bf I}_2$).
The subscripts in ${\bf I}_2,{\bf I}_3$ should be understood as being part of the `name' of a particular information quantity and should not be conflated with the total number of parties ${\sf N}$  (in particular, the arguments of  ${\bf I}_2$ and ${\bf I}_3$ can consist of composite subsystems; we will further comment on the relation between ${\sf N}$ and the number of parties appearing in a specific quantity $\bQ$ in \S\ref{sec:multipartite_information} and \S\ref{sec:discuss}).}

An important fact, already noticed in \citep{Bao:2015bfa}, is that SSA does not correspond to one of the facets of the holographic cone since it is a redundant inequality. A \textit{redundant} inequality is, by definition, one which is implied by other more fundamental inequalities since it can be obtained as a conical combination of them. For  SSA, the generating inequalities are MMI and an appropriate instance of SA, as illustrated in Fig.~\ref{fig:redundant_inequalities}. 
In particular, the redundancy of SSA implies that it can be saturated only by configurations, like $\c^{(1)}_3$ in Fig.~\ref{fig:redundant_inequalities}, which \textit{simultaneously} saturate both MMI and a particular instance of SA. These configurations are characterized by the fact that the corresponding entropy vectors, obtained as before by varying the regulator, only span a codimension-two subspace of entropy space which is the intersection of the two hyperplanes associated to the tripartite and conditional mutual information. This observation motivates our second definition:
\begin{definition}
A faithful information quantity $\bQ$ will be said to be \emph{primitive} if there exists at least one geometric state and one (sufficiently generic\footnote{ See footnote \ref{fn:generiticity}.}) configuration $\c_{\sf N}$ such that 
\begin{itemize}
\item $\bQ({\bf S}_\epsilon(\c_{\sf N}))=0$ independently from the cut-off $\epsilon$, and 
\item for any other information quantity $\,\bQ' \neq k\, \bQ$, with $k\in\mathbb{R}$, the equation $\bQ'({\bf S}_\epsilon(\c_{\sf N}))=0$ cannot hold generically, for an arbitrary choice of cut-off $\epsilon$.
\end{itemize}
\label{def:fund}
\end{definition}
\noindent
We will say that the configuration $\c_{\sf N}$ that satisfies these requirements \emph{generates} the primitive quantity $\bQ$.  

We are now in a position to state the full set of conditions we wish to impose to aid in our search for new information quantities. Namely, we require that for any number of parties, \emph{the information quantities of relevance are precisely the primitive ones.}

Our ultimate goal is to find the set of all primitive information quantities, for any value of ${\sf N}$, and study its properties. Although in this paper we will not answer this hard problem in full generality, we will explain in \S\ref{sec:formalization} how this can be done, at least in principle. In \S\ref{sec:three_parties} we will show that for three parties the primitive information quantities are precisely those that correspond to the facets of the holographic entropy cone. In particular, we will show that ${\bf I}_2(\mathcal{A}:\mathcal{B})$ and ${\bf I}_3(\mathcal{A}:\mathcal{B}:\mathcal{C})$ are primitive according to the previous definition by explicitly constructing the generating configurations ($\c^{(2)}_3$ and $\c^{(3)}_3$ in Fig.~\ref{fig:redundant_inequalities}). It is important to notice that a primitive information quantity does not necessarily correspond to a true holographic inequality, since it can be associated to a hyperplane that `cuts through' the cone (like $\bQ^{(2)}$ in Fig.~\ref{fig:redundant_inequalities}). Although the results of \S\ref{sec:three_parties} will show that for three parties this is not possible, we will see in \S\ref{sec:multipartite_information} that this can happen if ${\sf N}\geq 4$, and we will derive an infinite family of fundamental quantities which generalize ${\bf I}_3(\mathcal{A}:\mathcal{B}:\mathcal{C})$ to ${\sf N}\geq 3$.

\section{Formalization of the construction}
\label{sec:formalization}

We will now develop the formalism that allows us to derive the primitive information quantities defined in \S\ref{sec:overview}. First we explain in \S\ref{subsec:rephrasingRT} how the definitions of faithfulness and primitivity can be more efficiently reformulated by abstracting away from the issue of cut-off dependence. We will then see how the problem of finding the primitive quantities can be reformulated in terms of \textit{combinatorics of connected components of extremal surfaces}, requiring us to perform a scan over all possible geometric states and choices of configurations. Following that, in \S\ref{subsec:redundancy}, we will explain why such a scan is overly redundant, and how an appropriate `gauge fixing' can drastically simplify the problem. Finally, in \S\ref{subsec:organizing}, we will introduce a classification of configurations into families characterized by certain topological properties, which will turn out to be convenient for organizing the scan into various steps, at an increasing level of complexity.

\subsection{Proto-entropy and cut-off independence}
\label{subsec:rephrasingRT}

%
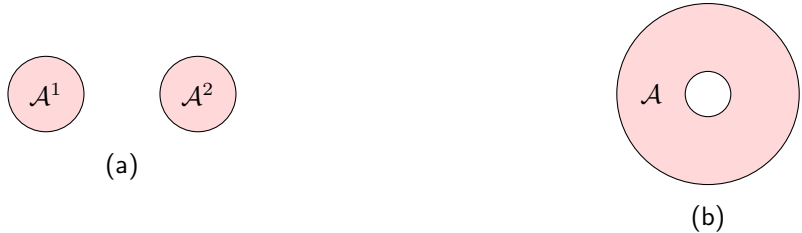
\begin{figure}
\centering
\begin{subfigure}{0.49\textwidth}
\centering
\begin{tikzpicture}
\draw[fill=red!15] (-1,0) circle (0.5cm);
\draw[fill=red!15] (1,0) circle (0.5cm);
\node at (-1,0) {\footnotesize{$\regA^1$}};
\node at (1,0) {\footnotesize{$\regA^2$}};
\end{tikzpicture}
\caption{}
\label{fig:domains_a}
\end{subfigure}
\hfill
\begin{subfigure}{0.49\textwidth}
\centering
\begin{tikzpicture}
\draw[fill=red!15] (0,0) circle (1.2cm);
\draw[fill=white!15] (0,0) circle (0.3cm);
\node at (-0.75,0) {\footnotesize{$\regA$}};
\end{tikzpicture}
\caption{}
\label{fig:domains_b}
\end{subfigure}
\caption{Examples of subsystems whose entropy is computed by a disconnected extremal surface $\extr{\regA}$. In (a) the two disks are sufficiently separated to be uncorrelated. In (b) the entangling surface $\partial\regA$ is disconnected and the ``inner'' circle is taken to be sufficiently small.}
\label{fig:domains}
\end{figure}

In the previous section, to motivate the definition of faithful and primitive information quantities, we have, for simplicity, considered examples where the bulk geometry was static, so that the entropies were computed by minimal surfaces via the RT formula. However, an important feature of our construction is that it does \textit{not} prefer in any way static geometries over dynamical ones. It will be equally valid for time-dependent states where the HRT prescription must be used to compute the holographic entanglement entropy.

Given this, we can consider a general set-up where the bulk is an asymptotically AdS manifold $\bulk$, of arbitrary dimension, with ${\sf M}$ disjoint causally disconnected  boundaries 
$\partial\bulk=\bigcup_{m=1}^{{\sf M}}\partial\bulk_m$. 
The bulk dynamics is dual to the time evolution of the tensor product $\text{CFT}^{\otimes {\sf M}}$ of multiple copies of a holographic CFT living on $\partial\bulk$. The state of the field theories on a Cauchy slice\footnote{
To generalize the notion of Cauchy slice to multiple disconnected (boundary) spacetime components $\partial\bulk_i$, we simply take a collection of Cauchy slices (one for each component), such that initial data on the full collection determines the evolution throughout the entire $\partial\bulk$.
} $\Sigma$ of $\partial\bulk$ is a pure state $\ket{\psi_\Sigma}$. 

In previous examples, we have moreover considered a given subsystem, say $\regA$, to be delineated by a spatial region (e.g., a single interval in the 1+1 dimensional CFT).  We now generalize this notion to allow the subsystem to consist of multiple regions.  To this end, on $\Sigma$  consider a subsystem $\regA=\bigcup_i\regA^i$ defined as the union of an arbitrary number of disjoint\footnote{
We use  the standard definition  of disjoint to disallow any intersection, including those of higher co-dimension, i.e., $\regA^i\cap\regA^j=\emptyset\quad\forall\;i,j.$} \textit{regions} $\regA^i$ distributed across the various boundaries. A  \textit{region} $\regA^i$ (denoted by an \textit{upper index} to distinguish it from subsystem identification) is defined as a connected subset of $\Sigma$, which is naturally associated with a bulk spacetime codimension-$2$ region.\footnote{ When we refer to boundary surfaces as having a particular bulk-codimension we are viewing the boundary as part of the bulk spacetime (at least topologically).} The state of the field theory on the subsystem $\regA$ is described by a reduced density matrix 
\begin{equation}
\rho_\regA=\Tr_{\regA^c}\ket{\psi_\Sigma}\bra{\psi_\Sigma}
\end{equation}
where the complement $\regA^c$ of $\regA$ is taken on $\Sigma$. To compute the entropy of $\regA$ via HRT, we proceed in two steps.  First we use the area functional to determine the bulk extremal surface $\extr{\regA}$ homologous to $\regA$ (and therefore anchored to the \textit{entangling surface} $\partial\regA=\bigcup_j \partial\regA^j$.)\footnote{
Note that the number of regions $\regA^i$, the number of entangling surfaces $\partial\regA^j$, and the number of connected components of the corresponding extremal surface $\extr{\regA}$ can all be distinct.  Furthermore, these numbers are likewise completely unrelated to the number of partitions ${\sf N}$ and the number of spacetime boundaries ${\sf M}$.} Second, we evaluate this area functional to determine the entropy
\begin{equation}
S_\epsilon(\rho_\regA)=\frac{\text{Area}_\epsilon(\extr{\regA})}{4\,G_N}\,.
\label{eq:HRT_original}
\end{equation}
Since the area of $\extr{\regA}$ is infinite, to obtain a finite value one has to introduce a cut-off surface which truncates the geometry $\bulk$. This corresponds to introducing a regulator $\epsilon$ in the field theory and we can think of \eqref{eq:HRT_original} as associating to $\regA$ a real function of $\epsilon$, $S_\epsilon(\rho_\regA)$, as described in \S\ref{sec:overview}. 

In general the bulk extremal surface found via the HRT prescription is not necessarily connected. A simple example is when the subsystem $\regA$ is a union of multiple disjoint regions and the mutual information between some of the regions vanishes (see Fig.~\ref{fig:domains_a} and the bottom panel of Fig.~\ref{fig:cut-off_independence}). It is however also possible for $\extr{\regA}$ to be disconnected even when $\regA$ is a single connected region; this happens for example if the entangling surface is disconnected (see Fig.~\ref{fig:domains_b}). 

In what follows it will be crucial to keep track of the connectivity of $\extr{\regA}$. We will therefore write $\extr{\regA}=\bigcup_\mu \esf^\mu$, where all the $\esf^\mu$ are connected codimension-2 bulk surfaces. We can then rewrite the HRT formula as
\begin{equation}
S_\epsilon(\rho_\regA)=\frac{1}{4G_N}\sum_\mu\,\text{Area}_\epsilon(\esf^\mu)
\; \eqdef\; 
\mathfrak{Area}_\epsilon\left[\sum_\mu\, \esf^\mu\right]\,.
\label{eq:HRT_new}
\end{equation}
In the above formula the sum $\sum_\mu\, \esf^\mu$ on the RHS is now a \textit{formal linear combination} of connected bulk extremal surfaces and we have defined a new operator $\mathfrak{Area}_\epsilon$ which acts linearly on this formal linear combination of surfaces and reduces to the usual area functional when it acts on a connected surface.\footnote{ We have absorbed the normalization factor $\tfrac{1}{4\,G_N}$ into the definition of $\mathfrak{Area}_\epsilon$ for convenience.} It is  important to note here that we are thinking of the area operator as a geometric object that takes a smooth codimension-2 bulk surface as input and gives back a number, its area, as output. In particular it is purely classical in the bulk and such is conceptually different from other notions of  area operators discussed in the holographic context cf., \cite{Papadodimas:2015jra,Almheiri:2016blp}.
 
Since we can think of the entropy $S_\epsilon(\rho_\regA)$ as obtained from a set of surfaces $\esf^\mu$ via the area operator, 
it is  convenient to introduce a new map $\breve{S}$, which one can think of as a sort of ``proto-entropy'', that associates to the subsystem $\regA$ the formal linear combination $\sum_\mu \,\esf^\mu$ which appears in \eqref{eq:HRT_new}:
\begin{equation}
\breve{S}(\rho_\regA) : \regA \; \mapsto \; \sum_\mu \esf^\mu\,.
\label{eq:protoS}
\end{equation}	
With this definition we can then write the holographic entanglement entropy in terms of the action of the area operator  $\mathfrak{Area}_\epsilon$ acting on this proto-entropy functional, viz.,
\begin{equation}
S_\epsilon(\rho_\regA)=\mathfrak{Area}_\epsilon\left[\breve{S}(\rho_\regA)\right]\,.
\end{equation}
In practice, for a state $\ket{\psi_\Sigma}$ and choice of subsystem $\regA$, one can evaluate $\breve{S}(\rho_\regA)$ by following the usual HRT prescription, but stopping short of choosing a cut-off surface and evaluating the area. Therefore, to efficiently implement the cut-off independence required by the two  definitions introduced in 
\S\ref{sec:overview}, we should abstract away from the usual entropy and rephrase these definitions in terms of the new map $\breve{S}$.

On $\Sigma$, we now consider a collection $\c_{\sf N}$ of ${\sf N}$ subsystems, labeled by $\regA_\ell$:  
\begin{equation}
\c_{\sf N} = \big\{\regA_\ell = \bigcup_i \regA_\ell^i\big\} \,, \qquad \ell \in \{1,2,...,{\sf N}\} \eqdef  \left[{\sf N}\right] \,.
\label{eq:}
\end{equation}	
We do not impose any restriction on the choice of subsystems, although by convention, and without loss of generality, we will take them to be non-overlapping.\footnote{ For any pair of subsystems $\regA_{\ell_1}$ and $\regA_{\ell_2}$ we have $\regA_{\ell_1}\cap\regA_{\ell_2}\subseteq\partial\regA_{\ell_1}\cap\partial\regA_{\ell_2}$, i.e.,  we only allow the subsystems to intersect on a higher co-dimension subset contained within their boundaries.} We will refer to the  \textit{lower index} as the \textit{color label/index} and the complement of the union of all subsystems $\univ$ the purifier.\footnote{ In our terminology the purifier is uncolored; $\univ$ is \textit{not} a ``color''.} The entropy vector associated to the state $\ket{\psi_\Sigma}$ and the configuration $\c_{\sf N}$ is then defined as
\begin{equation}
{\bf S}_\epsilon(\c_{\sf N},\psi_\Sigma)=\{S_\epsilon(\rho_{\regA_{\si}}),\; \si \subseteq [{\sf N}]\;\text{and}\;\si \neq\emptyset\}\,,
\qquad \regA_{\si}=\bigcup_{\ell \in \si} \regA_{\ell}
\label{eq:CFT_entropy_vector}
\end{equation} 
where the $\sf{D}$ components of the vector ${\bf S}_\epsilon$ are labeled by the new index $\si$ which represents a collection of colors, specified by the corresponding subset of $[{\sf N}]$, as in  \eqref{eq:info_quantity}. More precisely, $\si$ is a non-empty element of the power set of $[{\sf N}]$, i.e.,
\begin{equation}
\si \in \{\raisebox{0.25em}{$_{\{1\}}$},\raisebox{0.25em}{$_{\{2\}}$}, \ldots, \raisebox{0.25em}{$_{\{1,2\}}$}, \raisebox{0.25em}{$_{\{1,3\}}$}, \ldots, \raisebox{0.25em}{$_{\{1,2,3,\ldots, {\sf N}\}}$}\}\,.
\end{equation}	
Altogether there are three sets of labels associated with the subsystems we consider: 
\begin{itemize}
\item A lower index $\ell$ that specifies a color; we will call a collection of regions with fixed $\ell$ a \textit{monochromatic subsystem}.
\item An upper index $i,j,k$ that specifies the connected components (regions) of a particular color.
\item An index $\si, \sj, \sk$ which refers to a collection of monochromatic subsystems;  since such a collection invokes multiple colors, we will call it a \textit{polychromatic subsystem}.
\end{itemize}
The configuration $\c_{\sf N}$ is an amalgamation of all such possibilities. Note that in addition to the labels of the subsystems, we also have the index $\mu$ which labels the connected components of bulk extremal surfaces evoked in the computation of ${\bf S}_\epsilon(\c_{\sf N},\psi_\Sigma)$. 

We want to introduce a generalization of the entropy vector  ${\bf S}_\epsilon$ using the abstract proto-entropy map $\breve{S}$ defined above. For each of the subsystems $\regA_{\si}$, we build the list  $\om_{\si} = \bigcup_{\mu[\si]} \, \omega^{\mu[\si]}$ of all the connected bulk surfaces $\omega^{\mu[\si]}$ which enter in the computation of the entropy $S_\epsilon(\rho_{\regA_{\si}})$. We are using a  shorthand $\mu[\si]$ to denote the set of bulk surfaces which are associated with a particular polychromatic subsystem $\mathcal{A}_{\si}$.  
The union of all the sets $\om_{\si}$, for all $\si$, is a finite set $\om(\c_{\sf N},\psi_\Sigma)$, completely determined by the state and the choice of configuration. We then use $\om(\c_{\sf N},\psi_\Sigma)$ as a basis for the construction of an abelian free group 
$\boldsymbol{\mathscr{E}}(\c_{\sf N},\psi_\Sigma)$, which is the space of formal integer linear combinations of the elements of $\om(\c_{\sf N},\psi_\Sigma)$. The map $\breve{S}$ then associates an element of $\boldsymbol{\mathscr{E}}(\c_{\sf N},\psi_\Sigma)$ to any subsystem $\regA_{\si}$ and we can introduce the abstract vector
\begin{equation}
{\breve{\bf S}}(\c_{\sf N},\psi_\Sigma)=\{\breve{S}(\rho_{\regA_{\si}}),\;\si \subseteq [{\sf N}]\;\text{and}\;\si \neq\emptyset\}\,,\qquad 
\regA_{\si}=\bigcup_{\ell \in \si} \regA_{\ell}
\end{equation}
which is simply related to \eqref{eq:CFT_entropy_vector} by
\begin{equation}
{\bf S}_\epsilon(\c_{\sf N},\psi_\Sigma)=\big\{ \mathfrak{Area}_\epsilon\left[\breve{S}(\rho_{\regA_{\si}}) \right],\;\si \subseteq [{\sf N}]\;\text{and}\;\si \neq\emptyset\big\}
\eqdef
\mathfrak{Area}_\epsilon\left[{\breve{\bf S}}(\c_{\sf N},\psi_\Sigma)\right]\,.
\end{equation}

We are interested in information theoretic quantities $\bQ$ which are linear combinations of entropies, as in \eqref{eq:info_quantity}. If we replace the entropy vector ${\bf S}_\epsilon(\c_{\sf N},\psi_\Sigma)$ with the abstract form ${\breve{\bf S}}(\c_{\sf N},\psi_\Sigma)$, an expression like \eqref{eq:info_quantity} is an element of $\boldsymbol{\mathscr{E}}(\c_{\sf N},\psi_\Sigma)$ provided  the each coefficient $\qcf{\si}$ of the entropy $S_{\si}$ is an integer. 
We therefore define an abstract entropic quantity
\begin{equation}
\breve{\bQ}({\breve{\bf S}})=\sum_\si \qcf{\si}\, \breve{S}_{\si}\,,\qquad  \qcf{\si}\in\mathbb{Z} \,.
\label{eq:info_quantity_abstract}
\end{equation}
We can now think of the formal expression $\breve{\bQ}({\breve{\bf S}})=0_{\boldsymbol{\mathscr{E}}}$, where $0_{\boldsymbol{\mathscr{E}}}$ is the identity element in $\boldsymbol{\mathscr{E}}$, as an abstract version of an entropy relation. For a chosen pair $(\c_{\sf N},\psi_\Sigma)$ of state and configuration, we then have the important implication
\begin{equation}
\breve{\bQ}({\breve{\bf S}}(\c_{\sf N},\psi_\Sigma))=0_{\boldsymbol{\mathscr{E}}}\quad \Longrightarrow\quad \bQ({\bf S}_\epsilon(\c_{\sf N},\psi_\Sigma))=0
\label{eq:abstract_relation}
\end{equation}
independent from the choice of any UV regulator 
$\epsilon$.

Using this formalism we can now rephrase the definitions introduced in \S\ref{sec:overview} in a manifestly cut-off independent manner. For chosen $(\c_{\sf N},\psi_\Sigma)$, the evaluation of an abstract information quantity $\breve{\bQ}$ on ${\breve{\bf S}}(\c_{\sf N},\psi_\Sigma)$ takes the form
\begin{equation}
\begin{split}
\breve{\bQ}({\breve{\bf S}}(\c_{\sf N},\psi_\Sigma)) &= 
	\sum_\si \,  \qcf{\si}\, \breve{S}_{\si} (\c_{\sf N},\psi_\Sigma) =	
	 \sum_\si \qcf{\si}  \left(\sum_{\mu[\si]} \, \omega^{\mu[\si]}\right) \\
&	  \equiv   \sum_\si \qcf{\si}  \left(\sum_{\mu} \, M_{\si \mu} \, \omega^{\mu}\right)   	 
 =
	 \sum_{\mu} \left(\sum_\si\,  M_{\si \mu} \, \qcf{\si} \right) \omega^{\mu}\,.
	 \end{split}
\label{eq:surface_decomposition}
\end{equation}
The above formula plays a central role in our construction and deserves a careful explanation. The index $\mu[\si]$, as noted above, runs over the elements of the set $\om_\si$, i.e., all connected components of the extremal surface which computes the proto-entropy $\breve{S}_{\si} (\c_{\sf N},\psi_\Sigma)$. We want to extend this sum to all elements of $\om$ so that we can swap the order of the summation; we can implement this by introducing a $(0,1)$-matrix $M_{\si \mu}$ which for every polychromatic subsystem $\si$ takes into account which surfaces in $\om$ enter in the computation. The index $\mu$ in the last expression now runs over all elements of $\om$. Since all the surfaces $\omega^\mu$ are different (or equivalently, as they are linearly independent in the abstract vector space we conjured), the requirement that $\breve{\bQ}$ is \textit{faithful} translates into a system of linear equations
\begin{equation}
\left\{\sum_\si\,  M_{\si \mu} \, \qcf{\si}=0, \; \; \forall \mu\right\}
\label{eq:constaints_definition}
\end{equation}
which we will call \textit{constraints}. For a pair $(\c_{\sf N},\psi_\Sigma)$ we will indicate the list of corresponding constraints as $\{\f(\c_{\sf N},\psi_\Sigma)\}$. 
With this abstraction, the faithfulness requirement described in \S\ref{subsec:overview3} can then be rephrased as follows:

\begin{definition}
In an ${\sf N}$-partite setting, an entropic information quantity $\breve{\bQ}$ is \emph{faithful} if there exists at least one pair $(\c_{\sf N},\psi_\Sigma)$ such that the coefficients $\{ \qcf{\si}\}$ are a solution\footnote{ Obviously, we ignore the trivial solution $\bQ\equiv{\bf 0}$.} to the constraints $\{\f(\c_{\sf N},\psi_\Sigma)\}$.
\label{con:faith}
\end{definition}

It is clear that if an information quantity $\breve{\bQ}$ satisfies Definition~\ref{con:faith} it also satisfies Definition~\ref{def:faithful}, this is guaranteed by the implication \eqref{eq:abstract_relation}. To see that the opposite implication is also true, suppose that an information quantity $\breve{\bQ}$ does not satisfy Definition~\ref{con:faith}. This means that for \textit{any} pair $(\c_{\sf N},\psi_\Sigma)$ of a state and configuration, $\breve{\bQ}(\breve{{\bf S}}(\c_{\sf N},\psi_\Sigma))\neq 0_{\boldsymbol{\mathscr{E}}}$, i.e., it is a formal linear combination of \textit{some} surfaces. As explained in \S\ref{sec:overview}, this means that $\bQ({\bf S}_\epsilon(\c_{\sf N},\psi_\Sigma))$ is necessarily cut-off dependent. Specifically, even if this quantity could still vanish, it would vanish only for specific choices of the regulator.

We now would like to recast in this abstract language the notion of primitive information quantities (Definition~\ref{def:fund} in \S\ref{sec:overview}) which will play a central role in our analysis. Suppose that for a given faithful quantity $\breve{\bQ}$ we could find a pair $(\c_{\sf N},\psi_\Sigma)$ of a state and configuration such that the space of solutions to the constraints $\{\f(\c_{\sf N},\psi_\Sigma)\}$ has dimension greater than one, and includes $\breve{\bQ}$. This means that such space contains infinitely many other distinct information quantities $\breve{\bQ}'\neq k\breve{\bQ}$ which solve the same set of constraints. This in turn implies that for this particular pair $(\c_{\sf N},\psi_\Sigma)$ there are different faithful quantities that vanish independently from the cut-off, violating the second requirement of Definition~\ref{def:fund}. The primitivity requirement can therefore be rephrased as follows:

\begin{definition}
In an ${\sf N}$-partite setting, an entropic information quantity $\breve{\bQ}$ is \emph{primitive} if there exists at least one pair $(\c_{\sf N},\psi_\Sigma)$ such that the coefficients $\{Q_\si\}$ are the only solution (up to a constant factor) to the system of constraints $\{\f(\c_{\sf N},\psi_\Sigma)\}$.
\label{con:fun}
\end{definition}

Given an information quantity $\breve{\bQ}$, one could in principle scan over all possible pairs $(\c_{\sf N},\psi_\Sigma)$ of states and configurations, to determine if such a quantity is faithful, and eventually also primitive, according to the above definitions. However, this is not what we want to do. The whole purpose of constructing the present framework is instead to \textit{find} the primitive information quantities pertaining to geometric states. To this end, we will proceed in the opposite direction. 

For a fixed choice of a pair $(\c_{\sf N},\psi_\Sigma)$, we will think of the coefficients $\{ \qcf{\si}\}$ as variables and solve the set of constraints $\{\f(\c_{\sf N},\psi_\Sigma)\}$. Any solution will correspond to a faithful quantity, making again evident the weakness of such property. On the other hand, when the constraints $\{\f(\c_{\sf N},\psi_\Sigma)\}$ for a chosen pair $(\c_{\sf N},\psi_\Sigma)$ have a one parameter family of solutions, they will generate a primitive quantity $\breve{\bQ}$. Therefore, to find all primitive information quantities for any given number of parties ${\sf N}$, we will have to scan over all possible pairs of states and configurations to find all possible combinations that satisfy the above requirements. We will explain how to organize this scan in the next section. Although this problem seems a-priori overwhelmingly complex, we will see below  that there is a huge amount of redundancy and that efficiently removing such redundancy allows for a vast simplification.

Since our construction crucially depends on the usage of the proto-entropy defined above, as opposed to the usual entropy, from now on we will always implicitly assume this abstraction, and to simplify the notation we will write ${\bf S}$ and $\bQ$, instead of $\breve{{\bf S}}$ and $\breve{\bQ}$.

\subsection{Gauge-fixing for geometric states and configurations}
\label{subsec:redundancy}

Now that we have the basic framework in place, it is useful to first analyze how it can aid us in our search for primitive information quantities. A-priori we would want to make sure that the procedure is not overly redundant and identify the aspects that allow it to transcend some of the limitations of the previous explorations (such as those of \cite{Bao:2015bfa}). We now give a brief account of various features, though the discussion here will perhaps be more illuminating at a second reading, after that of \S\ref{sec:three_parties}, where we exemplify the procedure by deriving the $3$-party information quantities.  

Let us first see how much redundancy is built into the formalism. Consider a pair $(\c_{\sf N},\psi_\Sigma)$, comprising of a ${\sf N}$-party configuration and a state of the full system on $\Sigma$, which together generate a primitive information quantity $\bQ$ via the set of constraints $\{\f(\c_{\sf N},\psi_\Sigma)\}$.
Leaving the state $\psi_\Sigma$ fixed, we can deform the regions which compose the subsystems in $\c_{\sf N}$. This will entail a change in the geometry of the bulk extremal surfaces, which enter into the derivation of the constraints. However, as long as the change in the bulk surfaces is smooth there will be no effect on the constraints. On the other hand, the nature of the extremal surfaces would change under a phase transition, for instance where a connected and disconnected extremal surface exchange dominance connected  extremal surface exchanges dominance with a set of disconnected ones (within the same homology class). Similarly, keeping the configuration $\c_{\sf N}$ fixed, we can change the state $\psi_\Sigma$ to modify the bulk geometry. 
By our genericity assumption (see footnote \ref{fn:generiticity}),
small deformations will not affect the extremal surfaces overmuch, but we can certainly again  obtain a qualitative change in the extremal surfaces as the state becomes sufficiently different. In both cases, though, we would need to change the  connectedness of the extremal surfaces before we see a realignment of the constraints. Thus, a-priori, we have a large degree of redundancy in how the fundamental constraints  are manifested in the scan over pairs $(\c_{\sf N},\psi_\Sigma)$.  

However, this large redundancy  within the formalism can be converted into a virtue, once we identify the essential features that delineate a particular set of constraints over others. The essence of the previous paragraph is that the precise nature of the extremal surfaces is immaterial; all one cares about is how the different components $\regA_{\ell}^{i}$ making up the configuration are represented in the bulk via the surfaces $\omega^\mu$. 
A moment's thought will convince the reader that what we are describing here amounts to saying that the structure of constraints associated to a pair $(\c_{\sf N},\psi_\Sigma)$ only depends on the pattern of mutual information between the various regions $\regA_{\ell}^i$ which compose the various subsystems $\regA_{\ell}$.\footnote{ While it is easy to understand the construction in terms of the  mutual information, we will see later that the actual implementation is done in a slightly different manner in our algorithm for the search.} More specifically, what we care about is whether the mutual information between different parts of the configuration, say $\regA_{\ell_1}^{i_1}$ and $\regA_{\ell_2}^{i_2}$, is vanishing or non-vanishing. As was the case for the actual areas, the precise value of the mutual information is immaterial to our construction.

This feature allows us to truncate the redundancy by focusing on equivalence classes of pairs  $(\c_{\sf N},\psi_\Sigma)$ characterized by the constraints they produce (more on this below). We now make a set of (a-priori naive) observations,  which will allow for a vast simplification: 
\begin{itemize} 
\item Since we only care about the pattern of (vanishing vs.\ non-vanishing) mutual information between the regions $\regA_{\ell}^i$, we do not have to undertake a scan over all geometries. Given that our relations ultimately are tied to the divergence structure of individual entanglement entropies, and this is the same in all states, it in fact suffices that we focus on  the vacuum state of the theory! 
\item Since the freedom to deform the regions allows us to realize the requisite patterns of mutual information even if we limit ourselves to work in the vacuum, we need not even consider more general bulk geometries involving multiple boundary components and a tensor product of CFTs. Multi-boundary wormhole geometries are still nevertheless useful to construct the extremal rays, which was partly the reason why they were used extensively in the holographic entropy cone analysis of \cite{Bao:2015bfa} (see \S\ref{sec:discuss} for further comments on this point).
\item None of our arguments single out a particular dimension, so we can for convenience of visualization focus on the case of $(2+1)$-dimensional field theories where spatial regions $\regA_{\ell}^i$ are just two-surfaces embedded in ${\mathbb R}^2$ (and correspondingly the individual entangling surfaces which compose $\partial \regA_{\ell}^i$ are closed curves in ${\mathbb R}^2$). While passing to higher dimensions will of course allow for more complicated topology for $\regA_{\ell}^i$, this is again not relevant for our program since it  only adds to the aforementioned redundancy. Conversely, the situation in $(1+1)$-dimensional field theories is a bit too non-generic to extract useful lessons; it is not a-priori guaranteed that our technology can be applied in that setting effectively (see \S\ref{sec:discuss} for further comments on this case).
\end{itemize}

The fact that, for the purpose of finding the primitive information quantities, we can limit ourselves to work in the vacuum of a single CFT, will be perhaps more evident a-posteriori, by looking more carefully at the details of the derivation. Nonetheless, we can already provide an heuristic explanation for why this should be the case. As we mentioned above, the essential point is that since the actual value of the mutual information is immaterial, we can achieve any pattern of (vanishing vs.\ non-vanishing) mutual information between various regions already in the vacuum state of a 2+1 CFT on $\Sigma = {\mathbb R}^2$. The detailed structure of the configurations that generate the primitive information quantities is in general quite complicated, and as we said it is only 
in retrospect that one can prove that all the necessary patterns of correlations can be realized in this restricted setting.
For now we present two particular examples, which should however be sufficiently suggestive. 

The first example is a ${\sf N}$-party configuration made of ${\sf N}$ disjoint regions (one per color), each of which is topologically a disk. By an appropriate deformation of the individual regions it is possible to guarantee that the mutual information is non-vanishing for any pair of monochromatic subsystems, i.e.,   
\begin{equation}
{\bf I}_2(\regA_{\ell_1}:\regA_{\ell_2})\neq 0\quad\forall\;\ell_1,\ell_2\,,
\end{equation}
which in turn implies that the same relation also holds for any pair of polychromatic subsystems (by monotonicity of mutual information). Note that this particular pattern of correlations can not be achieved in the $(1+1)$-dimensional case if each color is represented by a single interval.

In the second example we again consider a ${\sf N}$-party configuration made of ${\sf N}$ disjoint disks, but now we hold their geometry fixed and only allow to change their size and location. Even under such restriction it is possible, still working only in the vacuum state, to achieve the following pattern of mutual information. For any two collections of subsystems $\regA_\si$ and $\regA_\sk$ such that $\regA_\si\cap\regA_\sk=\emptyset$ we have
\begin{equation}
\begin{cases}
{\bf I}_2(\regA_\si:\regA_\sk)\neq 0\quad \text{if}\;\regA_\si\cup\regA_\sk=\bigcup_{\ell=1}^{\sf N}\regA_\ell\\
{\bf I}_2(\regA_\si:\regA_\sk)= 0\quad \text{otherwise}
\end{cases}
\end{equation}
To see that this is the case, suppose for simplicity that ${\sf N}=3$ and that the triplet of disks is arranged on the vertices of an equilateral triangle. We can vary the distance between the vertices such that they are near enough to ensure a three-legged `octopus'\footnote{ Since we are considering topologically disk shaped regions in ${\mathbb R}^2$ the bulk surfaces are either (i) `domes' over a single region, or an (ii) `arch' straddling two regions, or more generally, (iii) an `octopus' homologous to multiple disks (cephalopod is more linguistically appropriate, but we will stick with octopus for sake of imagery).}  surface for $\regA_1\cup\regA_2\cup\regA_3$, but simultaneously far enough to disallow the `arch' like surfaces over any pair of disks.\footnote{
Geometrically, the fact that an octopus is possible without any arches (whereas two arches involving all three disks guarantee an octopus) follows from nesting of minimal surfaces \cite{Hayden:2011ag,Headrick:2013zda} (or more generally entanglement wedges \cite{Headrick:2014cta,Wall:2012uf}). Hence a pair of arches guarantees surface which lies outside both, which is the octopus, and so cannot for example be composed of the individual domes.
}   This particular structure of correlations will play an important role in the following, {especially in the derivation of the main theorem of \S\ref{sec:multipartite_information}. 

Let us take stock and summarize the above discussion in a manner that will enable us outline the overall strategy we wish to pursue.
The redundancy inherent in the scan over the choice of state and configurations can be phrased in terms of an equivalence relation: 
\begin{equation}
(\c_{\sf N},\psi_\Sigma) \simeq (\c'_{\sf N},\psi'_\Sigma) \;\; \Longleftrightarrow \;\; \{\f(\c_{\sf N},\psi_\Sigma)\} 
= \{\f(\c'_{\sf N},\psi'_\Sigma)\}\,.
\label{eq:equivalence}
\end{equation}	
This redundancy can be viewed as a form of `gauge invariance' and our gauge fixing procedure involves
\begin{itemize}
\item Restricting $\psi_\Sigma$ to be vacuum state of a CFT$_3$ on ${\mathbb R}^{2,1}$. 
\item Scanning over all possible configurations $\c_{\sf N}$, with an arbitrary number of regions with arbitrary topology and geometry.
\end{itemize}
To simplify the notation, since we have restricted to the vacuum of the theory, from now on we will always drop the state dependence from the constraints and only write $\{\f(\c_{\sf N})\}$.

As we explained above, even if we restrict to the vacuum state, when we slightly deform the regions such that there is no phase transition for the bulk surfaces, the constraints will not change. This means that even after our choice of gauge fixing there is still a residual redundancy. As above, this can be phrased in terms of an equivalence relation: 
\begin{equation}
\c_{\sf N} \simeq \c'_{\sf N} \;\; \Longleftrightarrow \;\; \{\f(\c_{\sf N}\} = \{\f(\c'_{\sf N})\}\,.
\label{eq:equivalence2}
\end{equation}	
Furthermore, it is clear from the above Definition~\ref{con:faith} and Definition~\ref{con:fun} that the actual form of the constraints is immaterial: all that really matters is the space of solutions. Therefore, by defining two (possibly different) sets of constraints $\{\f(\c_{\sf N}\}$ and $\{\f(\c'_{\sf N})\}$ to be equivalent if they have the same space of solutions, the above equivalence relation for configurations can be relaxed to the following
\begin{equation}
\c_{\sf N} \simeq \c'_{\sf N} \;\; \Longleftrightarrow \;\; \{\f(\c_{\sf N}\} \simeq \{\f(\c'_{\sf N})\}\,.
\label{eq:equivalence3}
\end{equation}	
As we proceed, it will become clear that this equivalence relation extends far beyond small continuous deformations of the configurations. Namely, there are configurations which are equivalent even if the topology of the bulk extremal surfaces, as well as of the configurations themselves, is very different. 

To summarize, we have reduced the problem of finding the primitive information quantities for ${\sf N}$ parties to the problem of classifying all the equivalence classes of configurations under the relation \eqref{eq:equivalence3}, and identifying among them all those which are associated to a set of constraints which has a one-dimensional space of solutions. However this is still a complicated problem, we will explain the next section how we plan to address it in the rest of the paper and future work \cite{Hubeny:2018aa}.

\subsection{The search strategy}
\label{subsec:organizing}

To classify the equivalence classes of configurations, we will find it convenient to organize the possible configurations into various families according to some topological properties. The main distinction will be between two scenarios:
\begin{itemize}
\item A \textit{disjoint scenario}, where all the regions are disjoint, i.e., 
\begin{equation}
\regA_{\ell_1}^{i_1}\cap\regA_{\ell_2}^{i_2}=\emptyset\qquad\forall \ell_1,\ell_2,i_1,i_2
\end{equation}
and the mutual information between any pair of subsystems is finite.
\item An \textit{adjoining scenario}, where regions of different colors can share portions of their boundaries although they never overlap, i.e., 
\begin{equation}
\regA_{\ell_1}^{i_1}\cap\regA_{\ell_2}^{i_2}\subseteq\partial\regA_{\ell_1}^{i_1}\cap\partial\regA_{\ell_2}^{i_2}\qquad\forall \ell_1,\ell_2,i_1,i_2
\end{equation}
and the mutual information is divergent for some pair of subsystems. 
\end{itemize}
 
As we will exemplify in \S\ref{sec:three_parties}, the nature of the constraints is more transparent in the disjoint scenario. However, as the number of parties grows, it is still far from obvious how to obtain the full classification. To tackle the problem, it will be convenient to further characterize the configurations according to an additional property that we will call \textit{enveloping}. Since we are working on $\mathbb{R}^2$, and all the regions composing the various subsystems are compact, the complement $\univ$ of any configuration $\c_{\sf N}$ (the purifier) is a union of a finite number of compact regions and a remaining part which is non-compact and extends to infinity. We will refer to this latter component of the purifier as the \textit{universe}. We will then say that the region $\regA_{\ell_1}^{i_1}$ 
\textit{is enveloping} (or envelops)
the region $\regA_{\ell_2}^{i_2}$ if for every pair of points $P,P'$, respectively in the universe and the region $\regA_{\ell_2}^{i_2}$, any connected path from $P$ to $P'$ has to cross the region $\regA_{\ell_1}^{i_1}$.\footnote{ This notion of enveloping can be generalized to the case where one has \textit{multiple enveloping} (for example the enveloped region $\regA_{\ell_2}^{i_2}$ is itself enveloping a third region $\regA_{\ell_3}^{i_3}$).} In the special case where none of regions is enveloping any other region, we will show in \S\ref{sec:multipartite_information} how it is possible to derive the full spectrum of primitive information quantities for any number of parties.\footnote{ More precisely, to simplify the proof, in \S\ref{sec:multipartite_information} we will make a slightly stronger assumption, namely that each of the regions is simply connected.} 
In \S\ref{sec:discuss} we will comment on the generalization to the case where enveloping is allowed, which we will explore in future work \cite{Hubeny:2018aa}.

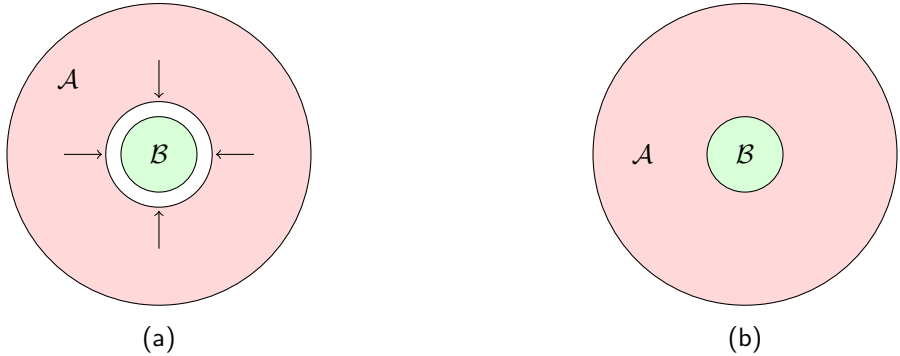
\begin{figure}[tb]
\centering
\begin{subfigure}{0.49\textwidth}
\centering
\begin{tikzpicture}
\draw[fill=red!15] (0,0) circle (2cm);
\draw[fill=white!15] (0,0) circle (0.7cm);
\draw[fill=green!15] (0,0) circle (0.5cm);
\draw[->] (0,1.25) -- (0,0.75);
\draw[->] (0,-1.25) -- (0,-0.75);
\draw[->] (-1.25,0) -- (-0.75,0);
\draw[->] (1.25,0) -- (0.75,0);
\node at (-1.2,1) {\footnotesize{$\mathcal{A}$}};
\node at (0,0) {\footnotesize{$\mathcal{B}$}};
\end{tikzpicture}
\caption{}
\label{}
\end{subfigure}
\hfill
\begin{subfigure}{0.49\textwidth}
\centering
\begin{tikzpicture}
\draw[fill=red!15] (0,0) circle (2cm);
\draw[fill=green!15] (0,0) circle (0.5cm);
\node at (-1.35,0) {\footnotesize{$\mathcal{A}$}};
\node at (0,0) {\footnotesize{$\mathcal{B}$}};
\end{tikzpicture}
\caption{}
\label{}
\end{subfigure}
\caption{The configuration (b) can be thought to be obtained from (a) by a deformation of $\mathcal{A}$ (indicated by the arrows).}
\label{fig:limit}
\end{figure}

If, for some value of ${\sf N}$, we can derive the full list of primitive quantities in the disjoint scenario, the hope is that one can then generalize the construction to the adjoining scenario. In this latter case, since the mutual information can be divergent, one would like to understand the configurations as a limiting case, where some regions become adjacent under a continuum deformation 
 (see Fig.~\ref{fig:limit}).\footnote{ The issues we encounter are similar to the discussion in \cite{Casini:2015woa}, where such a regulating scheme was employed to carefully tackle the proof of the F-theorem in three dimensions.} Even more complicated situation, where there are multiple intersections of entangling surfaces (see Fig.~\ref{fig:AL_alternatives2} of \S\ref{sec:three_parties} for an example) might require further consideration, however it is not a-priori guaranteed that these degenerate cases can in fact generate new primitive information quantities.\footnote{ On the other hand, using intuition from bit threads and multi-commodity flows \cite{Cui:2018dyq}, one might suspect that these multi-color junctions do implement new entanglement structures.  We thank Matt Headrick for sharing this perspective.}

\section{Holographic information quantities for three parties}
\label{sec:three_parties}

In this section we will use the formalism introduced in \S\ref{sec:formalization} to derive the primitive information quantities for the case of three parties. We will start by briefly reviewing the structure of the ${\sf N}=3$ holographic entropy cone, completing the discussion initiated in \S\ref{sec:overview}. Then we will show how all the information quantities associated to the facets of the cone, and in particular the tripartite information, can be generated by an appropriate configuration. Finally, we will argue that this list exhausts all the possibilities, and there exists no primitive information quantity which would \textit{not} correspond to one of the facets of the cone. 

The ${\sf N}=3$ holographic entropy cone is determined by MMI together with some instances of the bipartite inequalities (uplifted to the context of three subsystems). More specifically, for three subsystems $\mathcal{A},\mathcal{B},\mathcal{C}$ one can consider two different versions of SA (up to permutations of the labels),
\begin{equation}
S_{\mathcal{A}}+S_{\mathcal{B}}\geq S_{\mathcal{AB}}\qquad \text{and}\qquad S_{\mathcal{AB}}+S_{\mathcal{C}}\geq S_{\mathcal{ABC}} 
\,,
\label{eq:SA_instances}
\end{equation}
however only the former corresponds to a facet of the cone. The reason why the latter is redundant is that it can be obtained by summing MMI and the two other permutations of the former inequality. Note that, similarly to what we discussed in \S\ref{sec:overview} regarding the saturation of SSA, the fact that the second inequality in \eqref{eq:SA_instances} is redundant means that even if the corresponding information quantity is faithful, it cannot be primitive, since} the relation ${\bf I_2}(\mathcal{AB}:\mathcal{C})=0$ can only be satisfied if simultaneously ${\bf I_2}(\mathcal{A}:\mathcal{C})=0$ and ${\bf I_2}(\mathcal{B}:\mathcal{C})=0$ (as well as ${\bf I_3}(\mathcal{A}:\mathcal{B}:\mathcal{C})=0$).\footnote{ This is true also for arbitrary quantum systems, as a consequence of monotonicity of mutual information (which is SSA), here replaced by a stronger statement (MMI).}

In the case of the AL inequality, there are instead three formally different instances (again up to permutations of labels),
\begin{equation}
S_{\mathcal{A}}+S_{\mathcal{AB}}\geq S_{\mathcal{B}}\ ,\qquad S_{\mathcal{AB}}+S_{\mathcal{ABC}}\geq S_{\mathcal{C}}\ , \qquad\text{and}\qquad
S_{\mathcal{A}}+S_{\mathcal{ABC}}\geq S_{\mathcal{BC}}\,, 
\label{eq:AL_instances}
\end{equation}
and one can verify that only the last one is a facet inequality. This is consistent with the fact that it is the one which can be obtained from the first inequality of \eqref{eq:SA_instances} via the usual purification procedure.\footnote{ In applying this symmetry transformation we hold the total number of subsystems fixed. Specifically, to derive the last inequality of \eqref{eq:AL_instances} from the first one of \eqref{eq:SA_instances}, the purification of $\mathcal{AB}$ is $\mathcal{CO}$ and not $\mathcal{O}$ alone.}

The primitive information quantities which we want to derive are the ones which are associated to these facet inequalities. To sum up, they are the tripartite information $\bQ^{\text{MMI}} \equiv {\bf I_3}(\mathcal{A}:\mathcal{B}:\mathcal{C})$, the three permutations of the mutual information $\bQ^{\text{SA}} \equiv {\bf I_2}(\mathcal{A}:\mathcal{B})$, and the three permutations of 
\begin{equation}
\bQ^{\text{AL}} =S_\mathcal{A}+S_\mathcal{ABC}-S_\mathcal{BC} \,.
\end{equation}
Before we construct the configurations that generate these information quantities, let us first see how to attain the tripartite information ${\bf I_3}$ using our formalism. 

\begin{figure}[]
\centering
\begin{subfigure}{1\textwidth}
\centering
\begin{tikzpicture}
\draw[fill=red!15] (-5,0) circle (2cm);
\draw[fill=green!15] (-5,0) circle (0.5cm);
\draw[fill=green!15] (0,0) circle (2cm);
\draw[fill=blue!15] (0,0) circle (0.5cm);
\draw[fill=blue!15] (5,0) circle (2cm);
\draw[fill=red!15] (5,0) circle (0.5cm);
\draw[dotted,->] (-5,-1) -- (-5,-0.55);
\draw[dotted,->] (0,-1) -- (0,-0.55);
\draw[dotted,->] (5,-1) -- (5,-0.55);
\node at (-5.5,1) {\small{$\mathcal{A}_1$}};
\node at (-5,0) {\small{$\mathcal{B}_2$}};
\node at (-0.5,1) {\small{$\mathcal{B}_1$}};
\node at (0,0) {\small{$\mathcal{C}_2$}};
\node at (4.5,1) {\small{$\mathcal{C}_1$}};
\node at (5,0) {\small{$\mathcal{A}_2$}};
\node at (-5,-1.2) {\small{$b_2$}};
\node at (0,-1.2) {\small{$c_2$}};
\node at (5,-1.2) {\small{$a_2$}};
\node at (-3.4,1.6) {\small{$a_1$}};
\node at (1.6,1.6) {\small{$b_1$}};
\node at (6.6,1.6) {\small{$c_1$}};
\end{tikzpicture}
\end{subfigure}

\vspace{1cm}

\begin{subfigure}{1\textwidth}
\centering
\small
\begin{tabular}{|| c || c | c | c | c | c | c | c || }
\hline
\hline
{\shadeB } &  {\shadeB }  & {\shadeB  } &
	{ \shadeB } & { \shadeB } &
	{ \shadeB } & { \shadeB } & { \shadeB } \\
{\shadeB Surfaces} &  {\shadeB $S_{\mathcal{A}}$}      & {\shadeB $S_{\mathcal{B}}$ } &
	{ \shadeB $S_{\mathcal{C}}$} & { \shadeB $S_{\mathcal{AB}}$} &
	{ \shadeB $S_{\mathcal{AC}}$} & { \shadeB $S_{\mathcal{BC}}$} & { \shadeB $S_{\mathcal{ABC}}$}  \\
{\shadeB } &  {\shadeB }  & {\shadeB } &
	{ \shadeB } & { \shadeB } &
	{ \shadeB } & { \shadeB } &{ \shadeB } \\
\hline
\hline
{\shadeR $a_1$} & {\shadeR \checkmark} & {\shadeR } &
	{\shadeR } & {\shadeR \checkmark} & {\shadeR \checkmark} &  {\shadeR } &
	{\shadeR \checkmark}  \\
\hline 
{\shadeR $b_1$} & {\shadeR } & {\shadeR \checkmark} &
	{\shadeR } & {\shadeR \checkmark} & {\shadeR } &  {\shadeR \checkmark} &
	{\shadeR \checkmark}  \\
\hline 
 {\shadeR $c_1$} & {\shadeR } & {\shadeR } &
	{\shadeR \checkmark} & {\shadeR } & {\shadeR \checkmark} &  {\shadeR \checkmark} &
	{\shadeR \checkmark}  \\
\hline 
 {\shadeR $a_2$} & {\shadeR \checkmark} & {\shadeR } &
	{\shadeR \checkmark} & {\shadeR \checkmark} & {\shadeR } &  {\shadeR \checkmark} &
	{\shadeR }  \\
\hline
 {\shadeR $b_2$} & {\shadeR \checkmark} & {\shadeR \checkmark} &
	{\shadeR } & {\shadeR } & {\shadeR \checkmark} &  {\shadeR \checkmark} &
	{\shadeR }  \\
\hline 
 {\shadeR $c_2$} & {\shadeR } & {\shadeR \checkmark} &
	{\shadeR \checkmark} & {\shadeR \checkmark} & {\shadeR \checkmark} &  {\shadeR } &
	{\shadeR }  \\
\hline
\hline
\end{tabular}
\end{subfigure}
\caption{The simplest (minimal number of connected components of bulk extremal surfaces) configuration that generates the tripartite information ${\bf I_3}(\mathcal{A}:\mathcal{B}:\mathcal{C})$. 
Each column in the table is an entry of ${\bf S}(\c)$, while the rows correspond to the elements of $\om(\c)$. For each component $S_\mathscr{I}(\c)$ the check marks show which are the surfaces that enter in the linear combination.}
\label{fig:tripartite_information1}
\end{figure}

An example of a configuration $\c$ which generates the tripartite information is shown in Fig.~\ref{fig:tripartite_information1} for a $(2+1)$-dimensional CFT. Each subsystem in the configuration is the union of two regions,\footnote{ For small values of ${\sf N}$ we adopt a different notation for the subsystems, calling them $\mathcal{A},\mathcal{B},\mathcal{C},...$ (like in \S\ref{sec:overview}) instead of $\mathcal{A}_1,\mathcal{A}_2,\mathcal{A}_3,...$ (like in \S\ref{sec:formalization}). Consequently, the connected regions within the various subsystems are now labeled by a lower index.
} $\mathcal{A}=\mathcal{A}_1\mathcal{A}_2$, $\mathcal{B}=\mathcal{B}_1\mathcal{B}_2$ and $\mathcal{C}=\mathcal{C}_1\mathcal{C}_2$, and the labels $a_1,b_1,c_1,a_2,b_2,c_2$ indicate the  corresponding entangling surfaces.
We will restrict to configurations satisfying the following criteria:
\begin{itemize}
\item The distance between the disks $\mathcal{B}_2\mathcal{A}_1$, $\mathcal{C}_2\mathcal{B}_1$ and $\mathcal{A}_2\mathcal{C}_1$ are chosen such that they are all uncorrelated among each other. Specifically, we have ${\bf I_2}(\mathcal{B}_2\mathcal{A}_1:\mathcal{C}_2\mathcal{B}_1\mathcal{A}_2\mathcal{C}_1)=0$ and likewise for the other two cases.
\item  Furthermore, the disks $\mathcal{A}_2, \mathcal{B}_2, \mathcal{C}_2$ are taken to be sufficiently small such that they are uncorrelated with the purifier, i.e., ${\bf I_2}(\mathcal{A}_2:\univ\mathcal{C}_2\mathcal{B}_1\mathcal{B}_2\mathcal{A}_1)=0$ and likewise for  the other two cases. With this choice the surface which computes the entropy of each annular regions (for example $\mathcal{A}_1$) is the union of two surfaces, one homologous to the ``internal'' disk ($\mathcal{B}_2$), and the other to the union of the disk and the annulus ($\mathcal{B}_2\mathcal{A}_1$). 
\end{itemize}
In this particular case, each connected component of the bulk extremal surfaces $\omega^\mu$ is  specified by the single entangling surface on which it is anchored. With a little abuse of notation we will  give these bulk surfaces the same labels that we used for the entangling surfaces. 

Under these assumptions the set $\om(\c)$ is then built out of the six surfaces, viz., 
\begin{equation}
\om(\c)=\{a_1,b_1,c_1,a_2,b_2,c_2\} \,.
\end{equation}
 $\boldsymbol{\mathscr{E}}(\c)$ is then the set of formal integer linear combinations of these surfaces. We can now use the formalism introduced in the previous section and compute the entropy for each entry of the entropy vector ${\bf S}(\c)$ as a formal linear combination of the above surfaces. The results are displayed in the table of Fig.~\ref{fig:tripartite_information1}. Each column in the table is an entry of ${\bf S}(\c)$, while the rows correspond to the elements of $\om(\c)$. For each component $S_\mathscr{I}(\c)$ the check marks show which are the surfaces that enter in the linear combination. 

 The constraints \eqref{eq:constaints_definition} associated to $\c$ are then immediately readable from the rows of the table, explicitly they are
\begin{align}
\begin{cases}
\;\qcf{\mathcal{A}}+\qcf{\mathcal{AB}}+\qcf{\mathcal{AC}}+\qcf{\mathcal{ABC}}=0 \\
\;\qcf{\mathcal{B}}+\qcf{\mathcal{AB}}+\qcf{\mathcal{BC}}+\qcf{\mathcal{ABC}}=0 \\
\;\qcf{\mathcal{C}}+\qcf{\mathcal{AC}}+\qcf{\mathcal{BC}}+\qcf{\mathcal{ABC}}=0 \\
\;\qcf{\mathcal{A}}+\qcf{\mathcal{C}}+\qcf{\mathcal{AB}}+\qcf{\mathcal{BC}}=0 \\
\;\qcf{\mathcal{A}}+\qcf{\mathcal{B}}+\qcf{\mathcal{AC}}+\qcf{\mathcal{BC}}=0 \\
\;\qcf{\mathcal{B}}+\qcf{\mathcal{C}}+\qcf{\mathcal{AB}}+\qcf{\mathcal{AC}}=0 
\end{cases}
\end{align}
Plugging the one-parameter family of solutions to this system of equations back into the definition \eqref{eq:info_quantity} one gets 
$\bQ({\bf S})=\lambda \, {\bf I_3}(\mathcal{A}:\mathcal{B}:\mathcal{C})$, for some constant $\lambda$, which in entropy space is the hyperplane associated to the tripartite information.\footnote{ We stress that while this argument allows us to derive the tripartite information from purely holographic considerations, a-priori it does not have any implication for its sign definiteness. At this stage, to prove MMI, one still needs to rely on the common arguments of \cite{Hayden:2011ag}\cite{Wall:2012uf} (see \S\ref{sec:discuss} for further comments about the connection between primitive information quantities and holographic entropy inequalities).}

The example just described is particularly nice because the configuration $\c$ contains a minimal number of bulk surfaces. Moreover, the corresponding constraints are linearly independent. This is in fact the cleanest configuration that generates ${\bf I_3}$. However, to be able to systematize the search, we need to understand what is the origin of the constraints. Indeed, as we explained in \S\ref{sec:formalization}, it is precisely the possible structures of constraints that we need to classify, rather than the configurations themselves. 

As we mentioned in \S\ref{subsec:organizing}, for purposes of organizing the search in a way that can be generalized to more ($>3$) subsystems, it is useful to consider a restricted class of configurations where all the regions which compose the various subsystems are not adjacent to each other, i.e., they do not share any portion of their boundaries. This restriction is also preferable form a field theory perspective, since in this case the mutual information between all component subsystems is finite. The strategy will then be to first scan over this restricted class of configurations and only subsequently ask whether there is any new information quantity that can be generated by lifting this restriction.

Let us start by considering the simplest possible configuration, three disjoint disks $\mathcal{A}, \mathcal{B}, \mathcal{C}$ which are sufficiently separated from each other to be completely uncorrelated, i.e., ${\bf I_2}(\mathcal{A}:\mathcal{BC})=0$, ${\bf I_2}(\mathcal{B}:\mathcal{AC})=0$ and ${\bf I_2}(\mathcal{C}:\mathcal{AB})=0$. This trivial configuration cannot generate any primitive information quantity because the dimension of the space of solutions is too large. It is nevertheless useful to look at the structure of the corresponding constraints to build intuition for what follows. One  finds 
\begin{equation}
\begin{split}
\alpha:&\quad \qcf{\mathcal{A}}+\qcf{\mathcal{AB}}+\qcf{\mathcal{AC}}+\qcf{\mathcal{ABC}}=0\\
\beta:&\quad \qcf{\mathcal{B}}+\qcf{\mathcal{AB}}+\qcf{\mathcal{BC}}+\qcf{\mathcal{ABC}}=0 \\
\gamma:&\quad \qcf{\mathcal{C}}+\qcf{\mathcal{AC}}+\qcf{\mathcal{BC}}+\qcf{\mathcal{ABC}}=0
\end{split}
\label{eq:abconstraints} 
\end{equation}
Notice that the first constraint, which we call $\alpha$, is the sum of all the variables $ \qcf{\sj}$ where the index $\mathscr{I}$ contains the label $\mathcal{A}$, and similarly for $\beta$ and $\gamma$. 

If we move the disk $\mathcal{B}$ closer to $\mathcal{A}$, such that ${\bf I_2}(\mathcal{A}:\mathcal{B})\neq 0$ while we still have ${\bf I_2}(\mathcal{C}:\mathcal{A}\mathcal{B})=0$, the constraints change and we get
\begin{equation}
\begin{split}
\alpha\bar{\beta}:&\quad \qcf{\mathcal{A}}+\qcf{\mathcal{AC}}=0 \\
\beta\bar{\alpha}:&\quad \qcf{\mathcal{B}}+\qcf{\mathcal{BC}}=0 \\
\gamma:&\quad \qcf{\mathcal{C}}+\qcf{\mathcal{AC}}+\qcf{\mathcal{BC}}+\qcf{\mathcal{ABC}}=0 \\
\alpha\beta:&\quad \qcf{\mathcal{AB}}+\qcf{\mathcal{ABC}}=0
\end{split}
\end{equation}
where the bars now indicate which labels \textit{do not} appear in the sum. For example, $\alpha\bar{\beta}$ is the sum over all $ \qcf{\sj}$ where the index $\mathscr{I}$ contains the label $\mathcal{A}$ but not the label $\mathcal{B}$, while $\alpha\beta$ is the sum over all $ \qcf{\sj}$ where the index contains \textit{both} $\mathcal{A}$ and $\mathcal{B}$. Notice now that these constraints satisfy the simple relations
\begin{equation}
\begin{split}
\alpha\bar{\beta}+\alpha\beta=\alpha \\
\beta\bar{\alpha}+\alpha\beta=\beta
\end{split}
\end{equation}
Therefore if we replace the constraints $\alpha\bar{\beta}$ and $\beta\bar{\alpha}$ with $\alpha$ and $\beta$ while keeping $\alpha \beta$, obviously the solution is unchanged. We will say that the constraints $\{\alpha,\beta,\gamma,\alpha\beta\}^{\text{can}}$ are the \textit{canonical form} of the original constraints $\{\alpha\bar{\beta},\beta\bar{\alpha},\gamma,\alpha\beta\}$ derived from the configuration. The canonical form is characterized by the fact that there are ``no bars''.\footnote{ A precise definition will be given in \S\ref{sec:multipartite_information} for arbitrary ${\sf N}$.}

\begin{figure}[]
\centering
\begin{subfigure}{1\textwidth}
\centering
\begin{tikzpicture}
\draw[fill=red!15] (0,0) circle (1.2cm);
\draw[fill=blue!15] (-3,0) circle (1.2cm);
\draw[fill=green!15] (3,0) circle (1.2cm);
\draw[dashed] (-1.8,0) -- (-1.2,0);
\draw[dashed] (1.2,0) -- (1.8,0);
\node at (-3,0) {\small{$\mathcal{A}$}};
\node at (0,0) {\small{$\mathcal{B}$}};
\node at (3,0) {\small{$\mathcal{C}$}};
\node at (-4,1) {\small{$a$}};
\node at (-1,1) {\small{$b$}};
\node at (2,1) {\small{$c$}};
\end{tikzpicture}
\end{subfigure}

\vspace{1cm}

\begin{subfigure}{1\textwidth}
\centering
\small
\begin{tabular}{|| c || c | c | c | c | c | c | c || c ||}
\hline
\hline
{\shadeB } &  {\shadeB }  & {\shadeB  } &
	{ \shadeB } & { \shadeB } &
	{ \shadeB } & { \shadeB } & { \shadeB } & { \shadeB } \\
{\shadeB Surfaces} &  {\shadeB $S_{\mathcal{A}}$}      & {\shadeB $S_{\mathcal{B}}$ } &
	{ \shadeB $S_{\mathcal{C}}$} & { \shadeB $S_{\mathcal{AB}}$} &
	{ \shadeB $S_{\mathcal{AC}}$} & { \shadeB $S_{\mathcal{BC}}$} & { \shadeB $S_{\mathcal{ABC}}$} & { \shadeB Relations} \\
{\shadeB } &  {\shadeB }  & {\shadeB } &
	{ \shadeB } & { \shadeB } &
	{ \shadeB } & { \shadeB } &{ \shadeB } & { \shadeB } \\
\hline
\hline
{\shadeR $a$} & {\shadeR \checkmark} & {\shadeR } &
	{\shadeR } & {\shadeR } & {\shadeR \checkmark} &  {\shadeR } &
	{\shadeR } &  {\shadeR $\alpha\bar{\beta}$} \\
\hline 
{\shadeR $b$} & {\shadeR } & {\shadeR \checkmark} &
	{\shadeR } & {\shadeR } & {\shadeR } &  {\shadeR } &
	{\shadeR } &  {\shadeR $\beta\bar{\alpha}\bar{\gamma}$} \\
\hline 
 {\shadeR $c$} & {\shadeR } & {\shadeR } &
	{\shadeR \checkmark} & {\shadeR } & {\shadeR \checkmark} &  {\shadeR } &
	{\shadeR } &  {\shadeR $\gamma\bar{\beta}$} \\
\hline 
 {\shadeR $ab$} & {\shadeR } & {\shadeR } &
	{\shadeR } & {\shadeR \checkmark} & {\shadeR } &  {\shadeR } &
	{\shadeR } &  {\shadeR $\alpha\beta\bar{\gamma}$} \\
\hline 
 {\shadeR $bc$} & {\shadeR } & {\shadeR } &
	{\shadeR } & {\shadeR } & {\shadeR } &  {\shadeR \checkmark} &
	{\shadeR } &  {\shadeR $\beta\gamma\bar{\alpha}$} \\
\hline 
 {\shadeR $abc$} & {\shadeR } & {\shadeR } &
	{\shadeR } & {\shadeR } & {\shadeR } &  {\shadeR } &
	{\shadeR \checkmark} &  {\shadeR $\alpha\beta\gamma$} \\
\hline
\hline
\end{tabular}
\end{subfigure}
\caption{The configuration that generates the mutual information ${\bf I_2}(\mathcal{A}:\mathcal{C})$. The dashed lines indicate the pattern of mutual information between the disks, ${\bf I_2}(\mathcal{A}:\mathcal{B})\neq 0$ and ${\bf I_2}(\mathcal{B}:\mathcal{C})\neq 0$ while ${\bf I_2}(\mathcal{A}:\mathcal{C})=0$.}
\label{fig:mutual_information}
\end{figure}

If we now also bring the disk $\mathcal{C}$ closer to $\mathcal{B}$ we get a configuration that generates the mutual information, see Fig.~\ref{fig:mutual_information}. The table shows the constraints as obtained directly from the configuration, without any manipulations. Using relations between the constraints like the ones above, one can check that
\begin{equation}
\{\alpha\bar{\beta},\beta\bar{\alpha}\bar{\gamma},\gamma\bar{\beta},\alpha\beta\bar{\gamma},\beta\gamma\bar{\alpha},\alpha\beta\gamma\}\equiv\{\alpha, \beta, \gamma,  \alpha \beta, \beta\gamma, \alpha\beta\gamma\}^{\text{can}}\,.
\end{equation}
\begin{figure}[]
\centering
\begin{subfigure}{1\textwidth}
\centering
\begin{tikzpicture}
\draw[fill=red!15] (-5,-1.5) circle (1.2cm);
\draw[fill=green!15] (-5,1.5) circle (1.2cm);
\draw[fill=green!15] (0,-1.5) circle (1.2cm);
\draw[fill=blue!15] (0,1.5) circle (1.2cm);
\draw[fill=blue!15] (5,-1.5) circle (1.2cm);
\draw[fill=red!15] (5,1.5) circle (1.2cm);
\draw[dashed] (-5,-0.3) -- (-5,0.3);
\draw[dashed] (0,-0.3) -- (0,0.3);
\draw[dashed] (5,-0.3) -- (5,0.3);
\node at (-5,-1.5) {\small{$\mathcal{A}_1$}};
\node at (-5,1.5) {\small{$\mathcal{B}_1$}};
\node at (0,-1.5) {\small{$\mathcal{B}_2$}};
\node at (0,1.5) {\small{$\mathcal{C}_1$}};
\node at (5,-1.5) {\small{$\mathcal{C}_2$}};
\node at (5,1.5) {\small{$\mathcal{A}_2$}};
\node at (-6.1,-0.4) {\small{$a_1$}};
\node at (-6.1,2.6) {\small{$b_1$}};
\node at (-1.1,2.6) {\small{$c_1$}};
\node at (-1.1,-0.4) {\small{$b_2$}};
\node at (3.9,2.6) {\small{$a_2$}};
\node at (3.9,-0.4) {\small{$c_2$}};
\end{tikzpicture}
\end{subfigure}

\vspace{1cm}

\begin{subfigure}{1\textwidth}
\centering
\small
\begin{tabular}{|| c || c | c | c | c | c | c | c || c ||}
\hline
\hline
{\shadeB } &  {\shadeB }  & {\shadeB  } &
	{ \shadeB } & { \shadeB } &
	{ \shadeB } & { \shadeB } & { \shadeB } & { \shadeB } \\
{\shadeB Surfaces} &  {\shadeB $S_{\mathcal{A}}$}      & {\shadeB $S_{\mathcal{B}}$ } &
	{ \shadeB $S_{\mathcal{C}}$} & { \shadeB $S_{\mathcal{AB}}$} &
	{ \shadeB $S_{\mathcal{AC}}$} & { \shadeB $S_{\mathcal{BC}}$} & { \shadeB $S_{\mathcal{ABC}}$} & { \shadeB Relations} \\
{\shadeB } &  {\shadeB }  & {\shadeB } &
	{ \shadeB } & { \shadeB } &
	{ \shadeB } & { \shadeB } &{ \shadeB } & { \shadeB } \\
\hline
\hline
{\shadeR $a_1$} & {\shadeR \checkmark} & {\shadeR } &
	{\shadeR } & {\shadeR } & {\shadeR \checkmark} &  {\shadeR } &
	{\shadeR } &  {\shadeR $\alpha\bar{\beta}$} \\
\hline 
{\shadeR $b_1$} & {\shadeR } & {\shadeR \checkmark} &
	{\shadeR } & {\shadeR } & {\shadeR } &  {\shadeR \checkmark} &
	{\shadeR } &  {\shadeR $\beta\bar{\alpha}$} \\
\hline 
 {\shadeR $a_1b_1$} & {\shadeR } & {\shadeR } &
	{\shadeR } & {\shadeR \checkmark} & {\shadeR } &  {\shadeR } &
	{\shadeR \checkmark} &  {\shadeR $\alpha\beta$} \\
\hline 
 {\shadeR $b_2$} & {\shadeR } & {\shadeR \checkmark} &
	{\shadeR } & {\shadeR \checkmark} & {\shadeR } &  {\shadeR } &
	{\shadeR } &  {\shadeR $\beta\bar{\gamma}$} \\
\hline 
 {\shadeR $c_1$} & {\shadeR } & {\shadeR } &
	{\shadeR \checkmark} & {\shadeR } & {\shadeR \checkmark} &  {\shadeR } &
	{\shadeR } &  {\shadeR $\gamma\bar{\beta}$} \\
\hline 
 {\shadeR $b_2c_1$} & {\shadeR } & {\shadeR } &
	{\shadeR } & {\shadeR } & {\shadeR } &  {\shadeR \checkmark} &
	{\shadeR \checkmark} &  {\shadeR $\beta\gamma$} \\
\hline	
	{\shadeR $c_2$} & {\shadeR } & {\shadeR } &
	{\shadeR \checkmark} & {\shadeR } & {\shadeR } &  {\shadeR \checkmark} &
	{\shadeR } &  {\shadeR $\gamma\bar{\alpha}$} \\
\hline 
 {\shadeR $a_2$} & {\shadeR \checkmark} & {\shadeR } &
	{\shadeR } & {\shadeR \checkmark} & {\shadeR } &  {\shadeR } &
	{\shadeR} &  {\shadeR $\alpha\bar{\gamma}$} \\
\hline 
 {\shadeR $a_2c_2$} & {\shadeR } & {\shadeR } &
	{\shadeR } & {\shadeR } & {\shadeR \checkmark} &  {\shadeR } &
	{\shadeR \checkmark} &  {\shadeR $\alpha\gamma$} \\
\hline
\hline
\end{tabular}
\end{subfigure}
\caption{Alternative derivation of ${\bf I_3}(\mathcal{A}:\mathcal{B}:\mathcal{C})$, the configuration now is composed of correlated pairs of disjoint regions.}
\label{fig:alternative_tripartite_info}
\end{figure}

If the configuration is composed of only three disks one can easily convince oneself that the mutual information is the only information quantity that can be generated. If there are fewer correlations than in the previous case, the system of constraints has infinitely many solutions. If $\mathcal{A}$ and $\mathcal{C}$ are also correlated there are too many linearly independent constraints and the solution trivializes. Finally, if we permute the pattern of correlations we simply get other permutations of the mutual information for three parties.

To allow for the possibility of new information quantities we need to consider more complicated subsystems. The obvious generalization is the case where each of the subsystems $\mathcal{A}$, $\mathcal{B}$, $\mathcal{C}$ is composed of an arbitrary number of disks. This in fact allows us  to construct an alternative configuration that gives the tripartite information, see Fig.~\ref{fig:alternative_tripartite_info}. Again, the constraints can be rewritten as
\begin{equation}
\{\alpha, \beta, \gamma, \alpha\beta, \alpha\gamma, \beta\gamma\}^{\text{can}}\,.
\end{equation}
Notice that so far all the information quantities have been obtained from elements of the following list
\begin{equation}
\{\alpha, \beta, \gamma, \alpha\beta, \alpha\gamma, \beta\gamma, \alpha\beta\gamma\}^{\text{can}}\,.
\end{equation}
This fact will play a central role in \S\ref{sec:multipartite_information}.

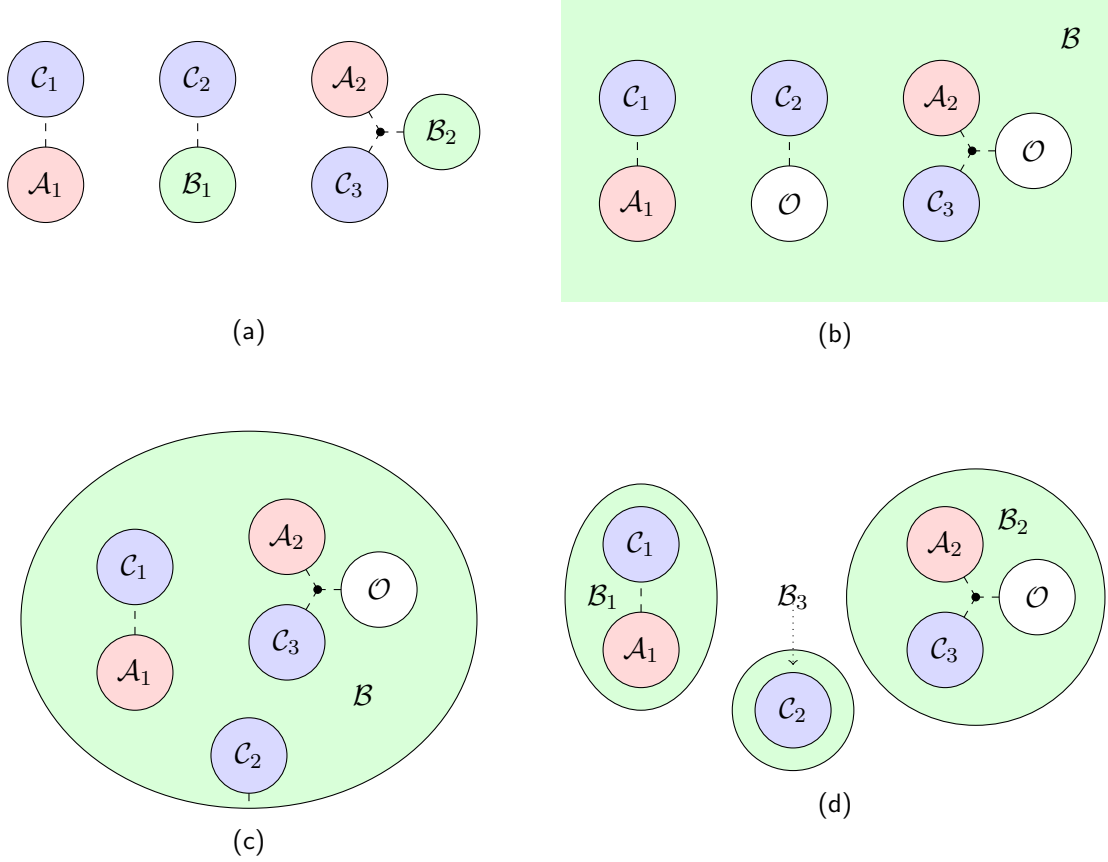
\begin{figure}[t]
\centering
\begin{subfigure}{0.49\textwidth}
\centering
\vspace{0.5cm}
\begin{tikzpicture}
\draw[fill=red!15] (-2,-0.7) circle (0.5cm);
\draw[fill=blue!15] (-2,0.7) circle (0.5cm);
\draw[fill=green!15] (0,-0.7) circle (0.5cm);
\draw[fill=blue!15] (0,0.7) circle (0.5cm);
\draw[fill=blue!15] (2,-0.7) circle (0.5cm);
\draw[fill=red!15] (2,0.7) circle (0.5cm);
\draw[fill=green!15] (3.212,0) circle (0.5cm);
\draw[dashed] (-2,-0.2) -- (-2,0.2);
\draw[dashed] (0,-0.2) -- (0,0.2);
\draw[dashed] (2.404,0) -- (2.25,-0.267);
\draw[dashed] (2.404,0) -- (2.712,0);
\draw[dashed] (2.404,0) -- (2.25,0.267);
\draw[fill=black] (2.404,0) circle (0.05cm);
\node at (-2,-0.7) {\small{$\mathcal{A}_1$}};
\node at (-2,0.7) {\small{$\mathcal{C}_1$}};
\node at (0,-0.7) {\small{$\mathcal{B}_1$}};
\node at (0,0.7) {\small{$\mathcal{C}_2$}};
\node at (2,-0.7) {\small{$\mathcal{C}_3$}};
\node at (2,0.7) {\small{$\mathcal{A}_2$}};
\node at (3.212,0) {\small{$\mathcal{B}_2$}};
\end{tikzpicture}
\vspace{1cm}
\caption{}
\end{subfigure}
\hfill
\begin{subfigure}{0.49\textwidth}
\centering
\begin{tikzpicture}
\fill[green!15] (-3,-2) -- (-3,2) --  (4.2,2) -- (4.2,-2) ;
\draw[fill=red!15] (-2,-0.7) circle (0.5cm);
\draw[fill=blue!15] (-2,0.7) circle (0.5cm);
\draw[fill=white] (0,-0.7) circle (0.5cm);
\draw[fill=blue!15] (0,0.7) circle (0.5cm);
\draw[fill=blue!15] (2,-0.7) circle (0.5cm);
\draw[fill=red!15] (2,0.7) circle (0.5cm);
\draw[fill=white] (3.212,0) circle (0.5cm);
\draw[dashed] (-2,-0.2) -- (-2,0.2);
\draw[dashed] (0,-0.2) -- (0,0.2);
\draw[dashed] (2.404,0) -- (2.25,-0.267);
\draw[dashed] (2.404,0) -- (2.712,0);
\draw[dashed] (2.404,0) -- (2.25,0.267);
\draw[fill=black] (2.404,0) circle (0.05cm);
\node at (3.7,1.5) {\small{$\mathcal{B}$}};
\node at (-2,-0.7) {\small{$\mathcal{A}_1$}};
\node at (-2,0.7) {\small{$\mathcal{C}_1$}};
\node at (0,-0.7) {\small{$\mathcal{O}$}};
\node at (0,0.7) {\small{$\mathcal{C}_2$}};
\node at (2,-0.7) {\small{$\mathcal{C}_3$}};
\node at (2,0.7) {\small{$\mathcal{A}_2$}};
\node at (3.212,0) {\small{$\mathcal{O}$}};
\end{tikzpicture}
\caption{}
\end{subfigure}

\vspace{1cm}

\begin{subfigure}{0.49\textwidth}
\centering
\begin{tikzpicture}
\draw[fill=green!15] (0,0) circle [x radius=3cm, y radius=2.5cm, rotate=0];
\draw[fill=red!15] (-1.5,-0.7) circle (0.5cm);
\draw[fill=blue!15] (-1.5,0.7) circle (0.5cm);
\draw[fill=blue!15] (0,-1.8) circle (0.5cm);
\draw[fill=blue!15] (0.5,-0.3) circle (0.5cm);
\draw[fill=red!15] (0.5,1.1) circle (0.5cm);
\draw[fill=white] (1.712,0.4) circle (0.5cm);
\draw[dashed] (-1.5,-0.2) -- (-1.5,0.2);
\draw[dashed] (0,-2.3) -- (0,-2.5);
\draw[dashed] (0.904,0.4) -- (0.75,0.133);
\draw[dashed] (0.904,0.4) -- (1.212,0.4);
\draw[dashed] (0.904,0.4) -- (0.75,0.667);
\draw[fill=black] (0.904,0.4) circle (0.05cm);
\node at (1.5,-1) {\small{$\mathcal{B}$}};
\node at (-1.5,-0.7) {\small{$\mathcal{A}_1$}};
\node at (-1.5,0.7) {\small{$\mathcal{C}_1$}};
\node at (0,-1.8) {\small{$\mathcal{C}_2$}};
\node at (0.5,-0.3) {\small{$\mathcal{C}_3$}};
\node at (0.5,1.1) {\small{$\mathcal{A}_2$}};
\node at (1.712,0.4) {\small{$\mathcal{O}$}};
\end{tikzpicture}
\caption{}
\end{subfigure}
\hfill
\begin{subfigure}{0.49\textwidth}
\centering
\begin{tikzpicture}
\draw[fill=green!15] (-2,0) circle [x radius=1cm, y radius=1.5cm, rotate=0];
\draw[fill=green!15] (2.404,0) circle (1.7cm);
\draw[fill=green!15] (0,-1.5) circle (0.8cm);
\draw[fill=red!15] (-2,-0.7) circle (0.5cm);
\draw[fill=blue!15] (-2,0.7) circle (0.5cm);
\draw[fill=blue!15] (0,-1.5) circle (0.5cm);
\draw[fill=blue!15] (2,-0.7) circle (0.5cm);
\draw[fill=red!15] (2,0.7) circle (0.5cm);
\draw[fill=white] (3.212,0) circle (0.5cm);
\draw[dashed] (-2,-0.2) -- (-2,0.2);
\draw[dashed] (2.404,0) -- (2.25,-0.267);
\draw[dashed] (2.404,0) -- (2.712,0);
\draw[dashed] (2.404,0) -- (2.25,0.267);
\draw[dotted,->] (0,0) -- (0,-0.9);
\draw[fill=black] (2.404,0) circle (0.05cm);
\node at (-2.5,0) {\small{$\mathcal{B}_1$}};
\node at (2.9,1) {\small{$\mathcal{B}_2$}};
\node at (0,0) {\small{$\mathcal{B}_3$}};
\node at (-2,-0.7) {\small{$\mathcal{A}_1$}};
\node at (-2,0.7) {\small{$\mathcal{C}_1$}};
\node at (0,-1.5) {\small{$\mathcal{C}_2$}};
\node at (2,-0.7) {\small{$\mathcal{C}_3$}};
\node at (2,0.7) {\small{$\mathcal{A}_2$}};
\node at (3.212,0) {\small{$\mathcal{O}$}};
\end{tikzpicture}
\caption{}
\end{subfigure}
\caption{Deriving Araki-Lieb from subadditivity via purification. Starting from a configuration (a) which generates subadditivity (we leave as an exercise for the reader to verify that this is the case, cf., the building blocks construction of \S\ref{sec:multipartite_information}), swap the subsystem $\mathcal{B}$ with the purifier $\mathcal{O}$ (b). We then choose one of the components of the purifier as a the new ``universe'' (c). The resulting configuration can be decomposed into simpler building blocks (d).}
\label{fig:AL_via_purification}
\end{figure}
\begin{figure}[tb]
\centering
\begin{subfigure}{0.49\textwidth}
\centering
\begin{tikzpicture}
\draw[fill=blue!15] (-2,0) circle (1.5cm);
\draw[fill=green!15] (-2,0) circle (1cm);
\draw[fill=red!15] (-2,0) circle (0.5cm);
\draw[fill=green!15] (2,0) circle (1.5cm);
\draw[fill=blue!15] (2,0) circle (1cm);
\draw[fill=red!15] (2,0) circle (0.5cm);
\node at (-2,0) {\footnotesize{$\mathcal{A}_1$}};
\node at (-2.7,0) {\footnotesize{$\mathcal{B}_1$}};
\node at (-3.2,0) {\footnotesize{$\mathcal{C}_1$}};
\node at (2,0) {\footnotesize{$\mathcal{A}_2$}};
\node at (1.3,0) {\footnotesize{$\mathcal{C}_2$}};
\node at (0.8,0) {\footnotesize{$\mathcal{B}_2$}};
\end{tikzpicture}
\caption{}
\end{subfigure}
\hfill
\begin{subfigure}{0.49\textwidth}
\centering
\begin{tikzpicture}
\draw[fill=green!15] (0,0) circle (1.5cm);
\fill[blue!15] (0,0.5) -- (0,1.5) --  (0,1.5) arc [start angle=90, end angle=-90, radius=1.5cm] -- (0,-0.5)
(0,-0.5) arc [start angle=-90, end angle=90, radius=0.5cm];
\draw (0,0) circle (1.5cm);
\draw[fill=red!15] (0,0) circle (0.5cm);
\draw (0,0.5) -- (0,1.5);
\draw (0,-0.5) -- (0,-1.5);
\node at (0,0) {\small{$\mathcal{A}$}};
\node at (-1,0) {\small{$\mathcal{B}$}};
\node at (1,0) {\small{$\mathcal{C}$}};
\end{tikzpicture}
\caption{}
\label{fig:AL_alternatives2}
\end{subfigure}
\caption{Two other (simpler) configurations that give Araki-Lieb. One requires multiple enveloping (a), while the other requires multiple intersection of entangling surfaces (b).}
\label{fig:AL_alternatives}
\end{figure}
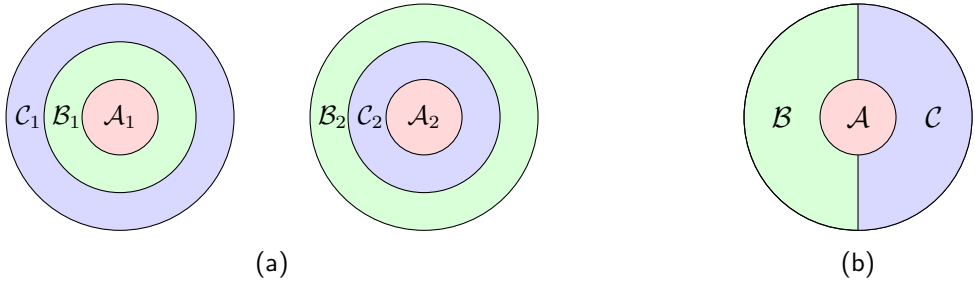

If the purifier is connected, the mutual and tripartite information are the only fundamental information quantities that can be generated for three parties. We will prove a more general version of this fact in \S\ref{sec:multipartite_information}, for arbitrary ${\sf N}$. Therefore to find new information quantities we have to consider even more complicated configurations. 

The last facet of the three parties holographic entropy cone that we have to generate is $\bQ^{\text{AL}}$. Since this is in fact equivalent to the mutual information, it should be possible to directly apply the purification symmetry to any configuration which generates the mutual information to derive a new configuration that generates $\bQ^{\text{AL}}$. This is in fact the case, and the procedure is shown in Fig.~\ref{fig:AL_via_purification}. There exist also simpler configurations which generate AL and some of them are shown in Fig.~\ref{fig:AL_alternatives}, however notice that they require either multiple enveloping or multiple intersections of entangling surfaces.  

We have shown how all the information quantities corresponding to the various facets of the ${\sf N}=3$ holographic entropy cone can be derived within our framework. In principle there could be additional primitive information quantities and they would correspond to hyperplanes that cut through the holographic entropy cone. We will defer a rigorous proof that this is not the case to a forthcoming work \cite{Hubeny:2018aa}, but intuitively we can already see why this is not possible by a simple counting argument. 
For ${\sf N}=3$ we have ${\sf D}=7$; therefore a configuration that generates a new fundamental quantity should be associated to six linearly independent constraints. However, we note that seven different constraints are necessary to avoid saturating the information  quantities already at hand. As this is an over-constrained system now, we do not expect to find further primitive information quantities.

\section{Holographic derivation of the multipartite information}
\label{sec:multipartite_information}

In this section we generalize the 3-party construction presented in \S\ref{sec:three_parties}, and we prove a general result about a particular family of primitive information quantities which are generated, for arbitrary ${\sf N}$, from a certain type of configurations. Although in this work we will not complete the scan over all possible configurations for arbitrary values of ${\sf N}$, the derivation presented here is the first step towards the solution of the more general problem \cite{Hubeny:2018aa}.

We will work in the vacuum of a single $(2+1)$-dimensional CFT on $\mathbb{R}^{2,1}$. 
As explained in \S\ref{sec:formalization} this is not a restriction, since changing the state (or the number of dimensions) does not generate any new information quantity. We consider ${\sf N}$ monochromatic subsystems $\regA_\ell$ ($\ell\in[{\sf N}]$), generally defined as unions of an arbitrary number of connected regions $\regA_\ell=\bigcup_i\regA_\ell^i$. A polychromatic subsystem $\regA_\si$ is the union of a collection of monochromatic ones (see \S\ref{sec:formalization} for more details on the definitions and the notation). It will often be convenient to keep track of the number of colors that characterize various objects; we will call it the \textit{degree} and we will denote it by an index $\kappa$. For example, a polychromatic index $\si$ of degree $\kappa$ will be denoted by $\si_\kappa$ and could characterize a $\kappa$-degree polychromatic subsystems $\regA_{\si_\kappa}$ or a $\kappa$-degree component of the entropy vector $S_{\si_\kappa}$.

Instead of considering completely general regions, in this section we will make two restrictions. First, we assume that each region has a single connected boundary (unlike, for instance, the example presented in Fig.~\ref{fig:restrictions_a}) 
\begin{equation}
\partial\regA_{\ell}^i\quad\text{is connected}\quad\forall\;\ell,i \,.
\label{eq:restriction2}
\end{equation}
Note that in $(2+1)$-dimensional field theory \eqref{eq:restriction2} is equivalent to imposing that each region is simply connected.\footnote{ However this is special to a $(2+1)$-dimensional setting and such an assumption would be unnecessarily restrictive in higher dimensions.  For example in $3+1$ dimensions a region with the topology of a solid torus would be allowed because even if not simply connected it has a connected boundary.} The second assumption we make is that all regions are completely disjoint, so that the boundaries of any two regions do not have any common point, i.e.,
\begin{equation}
\regA_{\ell_1}^{i_1}\cap\regA_{\ell_2}^{i_2}=\emptyset, \qquad \forall\;\ell_1,\ell_2,i_1,i_2\,.
\label{eq:restriction1}
\end{equation}
As we already mentioned in \S\ref{sec:formalization}, working in this disjoint scenario is a natural choice from a quantum field theory perspective, since it implies that the mutual information between any pair of subsystems is finite. The set of all possible configurations, for fixed ${\sf N}$, which fulfill the above requirements will be denoted by $\mathfrak{C}_{\sf N}$. 
In the present case of $(2+1)$-dimensional CFT, all configurations in $\mathfrak{C}_{\sf N}$ then have the topology of collections of disjoint disks.

It will be useful to consider a subset of ${\sf n}$ (out of ${\sf N}$) colors, which we can implement as follows. For a given value of ${\sf N}$, consider a permutation $\sigma$ (an element of the symmetric group of ${\sf N}$) defined as
\begin{equation}
\sigma: [{\sf N}]\longrightarrow [{\sf N}],\qquad\ell\longmapsto\sigma(\ell)\,.
\end{equation}
For any integer ${\sf n}$,  with $2\leq {\sf n}\leq {\sf N}$, we define a $(\sigma,{\sf n})$-\textit{reduction} of the entropy vector to be the entropy vector for ${\sf n}$-color subsystem obtained by  keeping the first ${\sf n}$ entries of a $\sigma$-permutation of the ${\sf N}$ colors.\footnote{ One can simply think of this as the collection of colors $\{\sigma(1), \sigma(2), \cdots, \sigma({\sf n})\}$.}  
Specifically, we consider the set $[{\sf n}]\subseteq[{\sf N}]$ and the restriction of $\sigma$ to this set 
\begin{equation}
\sigma |_{[{\sf n}]}: [{\sf n}]\longrightarrow [{\sf N}],\qquad\ell\longmapsto\sigma(\ell)\,.
\label{eq:reduction}
\end{equation}
The new monochromatic indices run over the image of $[{\sf n}]$ under the permutation, i.e., $\sigma(\ell)\in\sigma([{\sf n}])$. The polychromatic indices are collection of the new monochromatic ones as usual
\begin{equation}
\si^{(\sigma,{\sf n})}\subseteq \sigma([{\sf n}])\quad \text{and}\quad \si^{(\sigma,{\sf n})}\neq\emptyset\,.
\end{equation}
If a polychromatic index has degree $\kappa$, we will write $\si_\kappa^{(\sigma,{\sf n})}$ and we can think of the entropies $S_{\si^{(\sigma,{\sf n})}}$ as being the components of a ``reduced entropy vector''. Under these assumptions the following theorem holds
\begin{theorem}
{\bf (``${\bf I}_{\sf n}$-Theorem")}
For a given ${\sf N}$, the set of all the primitive information quantities generated by all the configurations in $\mathfrak{C}_{\sf N}$ is
\begin{equation}
\{{\bf I}_{\sf n}^{(\sigma)},\;\forall\sigma,\;2\leq {\sf n}\leq {\sf N}\}
\end{equation}
where ${\bf I}_{\sf n}^{(\sigma)}$ is the ${\sf n}$-partite information
\begin{equation}
{\bf I}_{\sf n}^{(\sigma)}(\regA_{\sigma(\ell_1)}:\regA_{\sigma(\ell_2)}:...:\regA_{\sigma(\ell_{\sf n})})=\sum_{\kappa=1}^{\sf n}\sum_{\si_\kappa^{(\sigma,{\sf n})}}(-1)^{\kappa+1}S_{\si_\kappa^{(\sigma,{\sf n})}}\,.
\end{equation}
\label{thm:In}
\end{theorem}
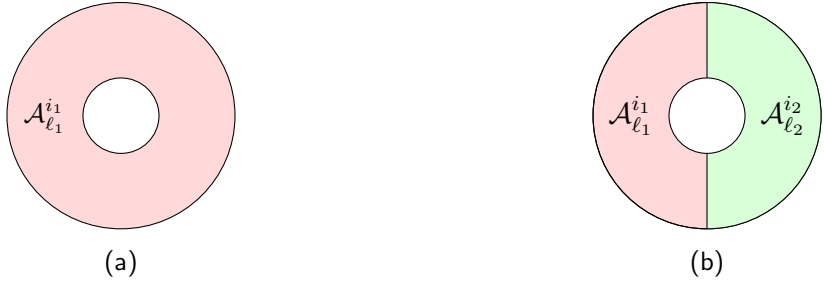
\begin{figure}[tb]
\centering
\begin{subfigure}{0.49\textwidth}
\centering
\begin{tikzpicture}
\draw[fill=red!15] (0,0) circle (1.5cm);
\draw[fill=white!15] (0,0) circle (0.5cm);
\node at (-1,0) {\footnotesize{$\regA_{\ell_1}^{i_1}$}};
\end{tikzpicture}
\caption{}
\label{fig:restrictions_a}
\end{subfigure}
\hfill
\begin{subfigure}{0.49\textwidth}
\centering
\begin{tikzpicture}
\draw[fill=red!15] (0,0) circle (1.5cm);
\fill[green!15] (0,0.5) -- (0,1.5) --  (0,1.5) arc [start angle=90, end angle=-90, radius=1.5cm] -- (0,-0.5)
(0,-0.5) arc [start angle=-90, end angle=90, radius=0.5cm];
\draw (0,0) circle (1.5cm);
\draw[fill=white!15] (0,0) circle (0.5cm);
\draw (0,0.5) -- (0,1.5);
\draw (0,-0.5) -- (0,-1.5);
\node at (-1,0) {\small{$\regA_{\ell_1}^{i_1}$}};
\node at (1,0) {\small{$\regA_{\ell_2}^{i_2}$}};
\end{tikzpicture}
\caption{}
\label{fig:restrictions_b}
\end{subfigure}
\caption{Examples of regions that violate the assumptions \eqref{eq:restriction2} and \eqref{eq:restriction1}. The region in (a) is disallowed because it has a disconnected boundary. The configuration in (b) is disallowed because the two regions share part of their boundaries. Note that the domain $\regA_{\ell_1}^{i_1}\cup\regA_{\ell_2}^{i_2}$ in (b) is connected but it has a disconnected boundary.}
\label{fig:restrictions}
\end{figure}

To prove the ${\bf I}_{\sf n}$-Theorem, we have to perform the scan over all possible configurations in $\mathfrak{C}_{\sf N}$. This will be done in the rest of this section through various steps. As explained in \S\ref{sec:formalization}, we can reorganize the scan over all possible configurations into a scan over equivalence classes. The first step then is to classify all equivalence classes in $\mathfrak{C}_{\sf N}$; this is achieved through Lemma~\ref{lemma:canonical_form} and Lemma~\ref{lemma:building_blocks} below. Next we have to identify which classes have a corresponding set of constraints with a one-parameter family of solutions (Lemma~\ref{lemma:permutations}). Finally, we solve the corresponding systems of equations to find the primitive information quantities (Lemma~\ref{lemma:solution}).

The restrictions \eqref{eq:restriction2}-\eqref{eq:restriction1} that we are imposing on the topology of the configurations in $\mathfrak{C}_{\sf N}$ have important implications for the structure of the extremal surfaces computing the various entropies; this is the content of Lemma~\ref{lemma:topology} below. To formulate the Lemma, it will be convenient to introduce a new terminology. We will call a \textit{domain} an arbitrary collection of regions, of arbitrary colors. By definition, a domain is a subset of a polychromatic subsystem and will be denoted by $\regA_\si^{\{i\}}\subseteq\regA_\si$, where the lower index is the usual polychromatic index indicating the color of the various regions of the domain, and the upper ``collective'' index $\{i\}$ is just a label for the collection of the various regions belonging to the domain.
\begin{lemma}
For any configuration $\c_{\sf N}\in\mathfrak{C}_{\sf N}$, each surface $\omega\in\om(\c_{\sf N})$ computes the entropy of a domain $\regA_\si^{\{i\}}$. Furthermore, any region within this domain ($\regA_\si^j\in\regA_\si^{\{i\}}$) has non-vanishing mutual information with the union of all other regions in the same domain, i.e., 
\begin{equation}
{\bf I}_2(\regA_\si^j:\regA_\si^{\{i\neq j\}})\neq 0,\qquad\forall\regA_\si^j\in\regA_\si^{\{i\}}\,.
\label{eq:topology}
\end{equation}
\label{lemma:topology}
\end{lemma}
\begin{proof} 
Consider a configuration $\c_{\sf N}\in\mathfrak{C}_{\sf N}$ and one of the monochromatic subsystems $\regA_\ell=\bigcup_i\regA_\ell^i$ in $\c_{\sf N}$. In general the extremal surface $\mathcal{E}_{\regA_\ell}$ is not connected and we decompose it as $\mathcal{E}_{\regA_\ell}=\sum_{\mu[\ell]}\omega^{\mu[\ell]}$, where all the $\omega^{\mu[\ell]}$ are connected. Since we are working in the vacuum, each surface $\omega^{\mu[\ell]}$ is anchored to the boundary.
 
In particular, an arbitrary surface $\omega^{\mu[\ell]}$ is anchored to a collection of entangling surfaces. Since we are assuming that each region $\regA_\ell^i\in\regA_\ell$ has a connected boundary \eqref{eq:restriction2}, by the homology constraint $\omega^{\mu[\ell]}$ has to be homologous to a domain $\regA_\ell^{\{i\}}$. Furthermore, the minimization involved in determining $\mathcal{E}_{\regA_\ell}$ implies that $\omega^{\mu[\ell]}$ is the minimal surface within the class of surfaces homologous to this domain and therefore computes its entropy.
 
To extend the argument to polychromatic subsystems we also need to invoke the other assumption \eqref{eq:restriction1}. As a consequence of \eqref{eq:restriction2}-\eqref{eq:restriction1}, any connected component of a domain is a region (i.e.,\ a connected component of a monochromatic subsystem), unlike the example indicated in Fig.~\ref{fig:restrictions_b}. Therefore we can run the same argument as before, completing the proof of the first part of the lemma.

To prove the second part, consider an arbitrary surface $\omega\in\om(\c_{\sf N})$. Again we denote by $\regA_\si^{\{i\}}$ the domain for which $\omega$ computes the entropy and we consider an arbitrary region $\regA_\si^j$ within this domain. The mutual information ${\bf I}_2(\regA_\si^j:\regA_\si^{\{i\neq j\}})$ vanishes if and only if\footnote{ See footnote \ref{fn:generiticity}.}
the surface computing the entropy of $\regA_\si^{\{i\}}$ splits into two surfaces computing the entropies of $\regA_\si^j$ and $\regA_\si^{\{i\neq j\}}$, which is not the case by assumption, because $\omega$ is connected.
\end{proof}

Two comments regarding the proof are in order when we consider non-vacuum states. First, the statement that every $\omega\in\om(\c_{\sf N})$ computes an entropy would not necessarily hold in certain excited states. Second, not all surfaces contributing to entropies need be anchored to the boundary.  For example, in a thermal state corresponding to an eternal Schwarzschild-AdS black hole, if the configuration includes regions in the `plateaux' regime \cite{Hubeny:2013gta}, then the corresponding extremal surface computing the entropy of such large regions consists of an extremal surface homologous to the complement along with a compact extremal surface wrapping the black hole horizon (bifurcation surface).  The latter by itself then does not compute the entropy of any of the given regions, unless we also include the purifier. However, these nuances are immaterial for our considerations, and these additional surfaces cannot change the result of this section, see \S\ref{subsec:redundancy} for further explanations.

As explained in \S\ref{sec:formalization}, two configurations $\c,\c'\in\mathfrak{C}_{\sf N}$ are equivalent if the corresponding sets of constraints $\{\f(\c)\}$ and $\{\f(\c')\}$ have the same subspace of solutions. Therefore, for the purpose of organizing all the configurations of $\mathfrak{C}_{\sf N}$ into equivalence classes, it is convenient to introduce a \textit{canonical form} for the constraints. The set $\mathfrak{F}^\text{can}$ of canonical form constraints, for fixed value of ${\sf N}$, is a set of ${\sf D=2^{\sf N}-1}$ linearly independent equations defined as follows
\begin{equation}
\mathfrak{F}^\text{can}=\{\f_{\si}^{\text{can}},\forall\,\si\subseteq[{\sf N}]\;\text{and}\;\si\neq\emptyset\},\qquad\f_{\si}^{\text{can}}:\;\sum_{\mathscr{K}\supseteq\mathscr{I}}Q_{\mathscr{K}}=0\,.
\label{eq:canonical_form_constaints}
\end{equation}
In other words, the set of coefficients $Q_{\mathscr{K}}$ which add up to zero are given by all the instances of ${\mathscr{K}}$ containing a given $\si$. Notice that this is the straightforward ${\sf N}$-party generalization of the canonical constraints $\{\alpha,\beta,\cdots,\alpha\beta\gamma\}^{\text{can}}$ introduced in \S\ref{sec:three_parties} for the $3$-party case.
We say that a given set of constraints $\mathfrak{F}$ is of the canonical form if it belongs to this set, $\mathfrak{F} \subseteq \mathfrak{F}^\text{can}$.
Using Lemma~\ref{lemma:topology} and the notion of canonical form constraints, we can now prove the main result which allows us to classify all the equivalence classes in $\mathfrak{C}_{\sf N}$.
\begin{lemma}
For any configuration $\c_{\sf N}\in\mathfrak{C}_{\sf N}$, all the corresponding constraints $\{\mathscr{F}(\c_{\sf N})\}$ can be converted to the canonical form \eqref{eq:canonical_form_constaints}.
\label{lemma:canonical_form}
\end{lemma}
\begin{proof}
For a configuration $\c_{\sf N}$, consider the set $\om=\om(\c_{\sf N})$ and a surface $\widehat{\omega}\in\om$, where the ``hat'' is stressing the fact that we are fixing a choice of one particular surface. We call $\widehat{\si}$ the set of colors of the domain $\regA_{\widehat{\si}}^{\{i\}}$ of which $\widehat{\omega}$ computes the entropy (Lemma~\ref{lemma:topology}). To find the constraint associated to $\widehat{\omega}$ we should examine various entries $S_\si$ of the entropy vector and determine for which of them $\widehat{\omega}$ enters in the combination of surfaces computing the entropy. For any surface $\widehat{\omega}\in\om$ we have the following possibilities depending on $\si$:
\begin{enumerate}[(i)]
\item When $\si\not\supseteq\widehat{\si}$, $S_\si$ does not include $\widehat{\omega}$.
\item $S_{\widehat{\si}}$ includes $\widehat{\omega}$.
\item When $\si\supset\widehat{\si}$, $S_\si$ may or may not include $\widehat{\omega}$ depending on the details of the configuration.
\end{enumerate}
(i) follows because in this case the domain $\regA_{\widehat{\si}}^{\{i\}}$ is not a subset of $\regA_\si$. (ii) follows from the definition of $\om$. Therefore the structure of the constraint depends entirely on (iii) and we have to look at the various possibilities for all the surfaces in $\om$. 

It will be convenient to organize the investigation of the structure of the constraints according to the number of colors associated to a surface. 
Mimicking the earlier terminology  for the number of colors of a polychromatic subsystem, we will say that a surface $\omega_\kappa\in\om$ has \textit{degree} $\kappa$ if the corresponding domain has $\kappa$ different colors. We decompose $\om$ as 
\begin{equation}
\om=\bigcup_\kappa\om_\kappa
\end{equation}
where $\om_\kappa$ is the set containing all surfaces $\omega_\kappa$ of degree $\kappa$. Note that $\om_1$ cannot be empty, while in general $\om_\kappa$ can be empty for any $\kappa\geq 2$.\footnote{ For example, if each $\regA_\ell$ is made of a single region and ${\bf I}_2(\regA_\ell:\c_{\sf N}\setminus\regA_\ell)=0\;\forall \ell$, then $\om_\kappa=\emptyset$ for all $\kappa\geq 2$.}

We denote by $\om_{\kappa^*}$ the set of surfaces of \textit{highest degree}, i.e., the set corresponding to the largest value of $\kappa$ such that $\om_{\kappa^*}\neq\emptyset$ and $\om_{\kappa}=\emptyset$ for all $\kappa^*<\kappa\leq {\sf N}$. Consider a surface $\widehat{\omega}_{\kappa^*}\in\om_{\kappa^*}$.  By Lemma~\ref{lemma:topology}, the surface $\widehat{\omega}_{\kappa^*}$ computes the entropy of a domain $\regA_{\widehat{\si}_{\kappa^*}}^{\{i\}}$. Each region within such domain has non-vanishing mutual information with the union of all the other regions in the same domain, Eq.~\eqref{eq:topology}. Since the mutual information cannot decrease by including additional subsystems, the surface $\widehat{\omega}_{\kappa^*}$ can `disappear' from entropies of higher degree only if it is replaced by a new surface which connects additional regions. However this would contradict the assumption that $\widehat{\omega}_{\kappa^*}$ has the highest degree. 

Therefore the surface $\widehat{\omega}_{\kappa^*}$ appears precisely in all the entropies $S_\si$ with $\si\supseteq\widehat{\si}_{\kappa^*}$ and the corresponding constraint is by definition in its canonical form. By the same argument, all surfaces in $\om_{\kappa^*}$ are associated to constraints which are automatically in their canonical form.

Consider now the set $\om_{\kappa'}$, where $\kappa'$ is the largest value of $\kappa$ such that $\kappa'<\kappa^*$ and $\om_{\kappa'}\neq\emptyset$, and choose a surface $\widehat{\omega}_{\kappa'}\in\om_{\kappa'}$. When $\si\supset\widehat{\si}_{\kappa'}$, we have the following possibilities:
\begin{enumerate}[(a)]
\item When $\mathscr{I}_\kappa\supset\widehat{\mathscr{I}}_{\kappa'}$ with $\kappa'<\kappa<\kappa^*$, $S_{\si_\kappa}$ includes $\widehat{\omega}_{\kappa'}$.
\item When $\mathscr{I}_\kappa\supset\widehat{\si}_{\kappa'}$ with $\kappa\geq\kappa^*$, there are two alternatives, depending on whether there exists a surface $\widehat{\omega}_{\kappa^*}\in\om_{\kappa^*}$ such that $\regA_{\widehat{\si}_{\kappa'}}^{\{i\}}\subset\regA_{\widehat{\si}_{\kappa^*}}^{\{j\}}$,
\begin{enumerate}[(b1)]
\item if such surface does not exist, then all the $S_{\si_\kappa}$ include $\widehat{\omega}_{\kappa'}$, for all $\mathscr{I}_\kappa\supset\widehat{\si}_{\kappa'}$ and $\kappa\geq\kappa^*$,
\item if such surface exists, it is \textit{unique}, and for all $\kappa\geq\kappa^*$,
\begin{itemize}
\item when $\si_\kappa\not\supseteq\widehat{\si}_{\kappa^*}$, $S_{\si_\kappa}$ includes $\widehat{\omega}_{\kappa'}$.
\item when $\si_\kappa\supseteq\widehat{\si}_{\kappa^*}$, $S_{\si_\kappa}$ does not include $\widehat{\omega}_{\kappa'}$.
\end{itemize} 
\end{enumerate}
\end{enumerate}
(a) and (b1) follow using the same argument that we used for the case $\kappa^*$, and the fact that all $\om_\kappa$ are empty (with $\kappa'<\kappa<\kappa^*$). In this case the constraint associated to $\widehat{\omega}_{\kappa'}$ is again automatically in canonical form. In the case (b2), the existence of this particular surface $\widehat{\omega}_{\kappa^*}$ implies (again from Lemma~\ref{lemma:topology})
\begin{equation}
{\bf I}_2\left(\regA_{\widehat{\si}_{\kappa^*}}^{\{j\}}\setminus\regA_{\widehat{\si}_{\kappa'}}^{\{i\}}:\regA_{\widehat{\si}_{\kappa'}}^{\{i\}}\right)\neq 0\,.
\end{equation}
To prove the uniqueness of the surface $\widehat{\omega}_{\kappa^*}$, suppose that there exist two different surfaces $\widehat{\omega}_{\kappa^*}^{(1)}$ and $\widehat{\omega}_{\kappa^*}^{(2)}$ such that 
\begin{equation}
\regA_{\widehat{\si}_{\kappa'}}^{\{i\}}\subset\regA_{\widehat{\si}_{\kappa^*}^{(1)}}^{\{j\}}\quad\text{and}\quad\regA_{\widehat{\si}_{\kappa'}}^{\{i\}}\subset\regA_{\widehat{\si}_{\kappa^*}^{(2)}}^{\{k\}}\,.
\end{equation}
Then Lemma~\ref{lemma:topology} and monotonicity of mutual information\footnote{ Specifically we use the fact that ${\bf I}_2(X:Y)\neq 0 \ \Rightarrow \ {\bf I}_2(X:YZ)\neq 0,\;\forall X,Y,Z$.} would imply the existence of a surface which connects all the regions in the three domains. The existence of such surface would either contradict the existence of the two previous surfaces or the hypothesis that $\kappa^*$ is the highest degree in $\om$. 

Now, if the unique surface $\widehat{\omega}_{\kappa^*}$ exists, then again by the same argument that we used for the surfaces of highest degree, $\widehat{\omega}_{\kappa'}$ must appear in all $S_{\si_\kappa}$ with $\si_\kappa\not\supseteq\widehat{\si}_{\kappa^*}$ (as for the cases (a) and (b1)). The reason why $\widehat{\omega}_{\kappa'}$ cannot appear in any $S_{\si_\kappa}$ with $\si_\kappa\supseteq\widehat{\si}_{\kappa^*}$ is that the surface $\widehat{\omega}_{\kappa^*}$ by definition has smaller area than the union of $\widehat{\omega}_{\kappa'}$ and another minimal surface homologous to $\regA_{\widehat{\si}_{\kappa^*}}^{\{j\}}\setminus\regA_{\widehat{\si}_{\kappa'}}^{\{i\}}$. It follows then that in this case we have 
\begin{equation}
\f^{\text{can}}_{\widehat{\si}_{\kappa'}}=\f(\widehat{\omega}_{\kappa'})+\f^{\text{can}}_{\widehat{\si}_{\kappa^*}}\,,
\end{equation}
meaning that we can use the canonical constraint for a surface $\widehat{\omega}_{\kappa^*}$ to convert to canonical form the one for $\widehat{\omega}_{\kappa'}$. Proceeding in this fashion for all surfaces in $\om_{\kappa'}$ we can convert all the corresponding constraints to canonical form.

Finally, consider the set $\om_{\kappa''}$, where $\kappa''$ is the largest values of $\kappa$ such that $\kappa''<\kappa'$ and $\om_{\kappa''}\neq\emptyset$, and choose a surface $\widehat{\omega}_{\kappa''}\in\om_{\kappa''}$. The same logic that we used before, for $\kappa'$ and $\kappa^*$, can be applied also to this case. The important difference however is that now, if there is a surface $\widehat{\omega}_{\kappa'}\in\om_{\kappa'}$ such that $\regA_{\widehat{\si}_{\kappa''}}^{\{i\}}\subset\regA_{\widehat{\si}_{\kappa'}}^{\{j\}}$, it does not have to be unique any more. On the other hand, if there are two surfaces $\widehat{\omega}_{\kappa'}^{(1)}$ and $\widehat{\omega}_{\kappa'}^{(2)}$, the same logic that we used above to prove uniqueness can now be used to show that there must exist a higher degree surface $\widehat{\omega}_{\kappa^*}\in\om_{\kappa^*}$  such that the constraint for $\widehat{\omega}_{\kappa''}$ can be written as
\begin{equation}
\f^{\text{can}}_{\widehat{\si}_{\kappa''}}=\f(\widehat{\omega}_{\kappa''})+\f(\widehat{\omega}^{^(1)}_{\kappa'})+\f(\widehat{\omega}^{^(2)}_{\kappa'})+\f^{\text{can}}_{\widehat{\si}_{\kappa^*}}\,.
\end{equation}
The logic naturally generalizes to the case where there are more than two surfaces. By iterating this procedure for all non-empty sets $\om_k$, all the way down to $\om_1$, we can then convert into canonical form all the constraints for the entire configuration.
\end{proof}

The previous lemma shows that the constraints associated to any configuration are equivalent to a subset $\mathfrak{F}$ of the set $\mathfrak{F}^\text{can}$ of canonical form constraints, but it does not automatically imply that an \textit{arbitrary} subset can be realized by some configuration. To prove which combinations of canonical form constraints are consistent with the possible configurations in $\mathfrak{C}_{\sf N}$, it will be convenient to define a particular way of combining configurations to build new ones. 

Given two configurations $\c'_{\sf N}$ and $\c''_{\sf N}$ we want to construct a new configuration $\c'_{\sf N}\sqcup\c''_{\sf N}$, which we call the \textit{uncorrelated union}, defined as follows. We consider the two configurations in the same copy of the vacuum state, and we take them sufficiently far apart from each other,\footnote{ 
Had the CFT lived on ${\bf S}^2$ rather than ${\mathbb{R}}^2$, we would also potentially need to first shrink each configuration appropriately in order to fully decorrelate them; however, as mentioned in \S\ref{subsec:redundancy} and verified below, the case of ${\mathbb{R}}^2$ is sufficient for our considerations.} such that ${\bf I}_2(\c'_{\sf N}:\c''_{\sf N})=0$. The following property then holds:
\begin{lemma}
For a configuration $\c_{\sf N}=\c'_{\sf N}\sqcup\c''_{\sf N}$, the list of constraints $\{\f(\c_{\sf N})\}$ is the union of the two lists of constraints $\{\f(\c'_{\sf N})\}$ and $\{\f(\c''_{\sf N})\}$ for $\c'_{\sf N}$ and $\c''_{\sf N}$, respectively.
\label{lemma:uncorrelated_union}
\end{lemma}
\begin{proof}
Clearly translating $\c'_{\sf N}$ and $\c''_{\sf N}$ does not change the constraints $\{\f(\c'_{\sf N})\}$ and $\{\f(\c''_{\sf N})\}$. Furthermore, since ${\bf I}_2(\c'_{\sf N}:\c''_{\sf N})=0$, we have $S_\mathscr{I}(\c_{\sf N})=S_\mathscr{I}(\c'_{\sf N})+S_\mathscr{I}(\c''_{\sf N})$ for all $\mathscr{I}$. Since $\bQ$ is linear, \eqref{eq:surface_decomposition} now gives 
\begin{equation}
\bQ({\bf S}(\c_{\sf N}))=\sum_{\mu'} \left(\sum_\si\,  M'_{\si \mu'} \, \qcf{\si} \right) \omega^{\mu'}+\sum_{\mu''} \left(\sum_\si\,  M''_{\si \mu''} \, \qcf{\si} \right) \omega^{\mu''}\,,
\end{equation}
and since all the surfaces $\{\omega^{\mu'}\}$ and $\{\omega^{\mu''}\}$ are independent, the thesis follows.
\end{proof}

Using the uncorrelated union, we can now prove which combinations of canonical form constraints can be realized by some configuration.  
It is helpful to note at the outset that the constraints must always include those coming from the constraints associated with each individual color, i.e., those of the form 
$\alpha$, $\beta$, etc.\ in the notation introduced in Eq.~\eqref{eq:abconstraints}. The essential result we need is summarized by the following 
\begin{lemma}
For any subset $\mathfrak{F}\subseteq\mathfrak{F}^{\rm can}$, there exists a configuration $\c_{\sf N}\in\mathfrak{C}_{\sf N}$ such that $\{\f(\c_{\sf N})\}=\mathfrak{F}$
if and only if
\begin{equation}
\mathfrak{F} \supseteq \mathfrak{F}_{[{\sf N}]}\ , \qquad {\rm where} \ \ \mathfrak{F}_{[{\sf N}]}\eqdef\{\f^{\rm can}_\ell,\; \ell\!\in\!\left[{\sf N}\right]\}\,.
\label{eq:subsets}
\end{equation}
\label{lemma:building_blocks}
\end{lemma}
\vspace{-1cm}
\begin{proof}
We first prove the direct statement by explicitly constructing, for any $\mathfrak{F}\subseteq\mathfrak{F}^\text{can}$ satisfying \eqref{eq:subsets}, a configuration with an equivalent set of constraints. Consider ${\sf N}$ disks with their centers arranged on a circle of radius $R$ such that the distance along the circle between two nearest-neighbor disks is constant (see Fig.~\ref{fig:building_blocks_a}). We attribute a different color to each disk, such that we have $[{\sf N}]$ subsystems $\regA_\ell$ each of which is composed of a single region. Furthermore, we assume that $R$ is sufficiently large such that 
\begin{equation}
{\bf I}_2\left(\regA_{\ell_i}:\bigcup_{j\neq i}\regA_{\ell_j}\right)=0,\qquad \forall i\in[{\sf N}]\,.
\end{equation}
This particular configuration will be denoted by $\c_{\sf N}^{^\circ}$. 
It is straightforward to check that the set of constraints associated to this particular configuration is precisely $\mathfrak{F}_{[{\sf N}]}$. 

Consider now the set $\mathfrak{F}=\mathfrak{F}_{[{\sf N}]}\cup\{\f_{\widehat{\si}_n}^{\text{can}}\}$, where $\f_{\widehat{\si}_n}^{\text{can}}$ is a choice of an arbitrary additional single constraint. To construct a configuration $\c_{\sf N}$ which gives precisely the specified set of constraints, $\mathfrak{F}=\{\f(\c_{\sf N})\}$, we proceed as follows. Starting from $\c_{\sf N}^{^\circ}$ we move the disks corresponding to $\widehat{\si}_n$ inward while holding the other disks fixed (see Fig.~\ref{fig:building_blocks_a}), such that the centers of these disks are now located on a smaller circle of radius $r$. As $r$ decreases below some threshold, the mutual information between any subset of disks in $\widehat{\si}_n$ and the collection of the remaining ones will no longer vanish, while the analogous mutual information for any proper subset of such collections still remains zero.  In particular, there exists a minimal value $r^*$, which depends on the size of the disks, such that: 
\begin{itemize}
\item The entropy of the union of these $\widehat{\si}_n$ disks will be computed by a multi-legged octopus\footnote{ While irrelevant to our discussion here, it is interesting to note that these octopoid surfaces can in fact have non-trivial topology in the bulk, leading to `holey-octopi'. We thank Erik Tonni for checking some examples numerically.}
surface connecting all of them. 
\item The entropy of any subset of these disks is still computed by a union of domes homologous to the various disks.
\end{itemize}
Using the procedure of Lemma~\ref{lemma:canonical_form}, one can easily show that the constraints associated to this configuration are equivalent to $\mathfrak{F}$ as we wanted. We denote this configuration by $\c_{\sf N}^{^\circ}[\widehat{\si}_n]$  and refer to it as a \textit{building block}. 

Finally consider an arbitrary set $\mathfrak{F}$ satisfying \eqref{eq:subsets},
i.e., one which includes all the obligatory constraints $\mathfrak{F}_{[{\sf N}]}$.
 Using the uncorrelated union defined in Lemma~\ref{lemma:uncorrelated_union}, the desired configuration, such that $\{\f(\c_{\sf N})\}=\mathfrak{F}$, can be constructed from the building blocks as follows
\begin{equation}
\c_{\sf N}=\bigsqcup_{\widehat{\si}_n\;\text{s.t.}\;\f_{\widehat{\si}_n}^{\text{can}}\in\mathfrak{F}}\c_{\sf N}^{^\circ}[\widehat{\si}_n]\,.
\end{equation}
In other words, we can simply add constraints one by one, each implemented by its corresponding configuration $\c_{\sf N}^{^\circ}[\widehat{\si}_n]$ separated from all the others.

To prove the converse statement
(that if $\mathfrak{F}$ does not include the full set $\mathfrak{F}_{[{\sf N}]}$, then it cannot be generated by any configuration),
one only has to revisit the proof of Lemma~\ref{lemma:canonical_form}. The configuration $\c_{\sf N}^{^\circ}$ is the one which has the minimal number of entangling surfaces, and as we discussed above, gives $\mathfrak{F}_{[{\sf N}]}$. For any other configuration $\c_{\sf N}\in\mathfrak{C}_{\sf N}$, Lemma~\ref{lemma:canonical_form} showed that the constraints can be reduced to canonical form and it is clear from the proof (since $\om_1$ cannot be empty) that one always has $\mathfrak{F}_{[{\sf N}]}\subseteq\mathfrak{F}$. 
\end{proof}

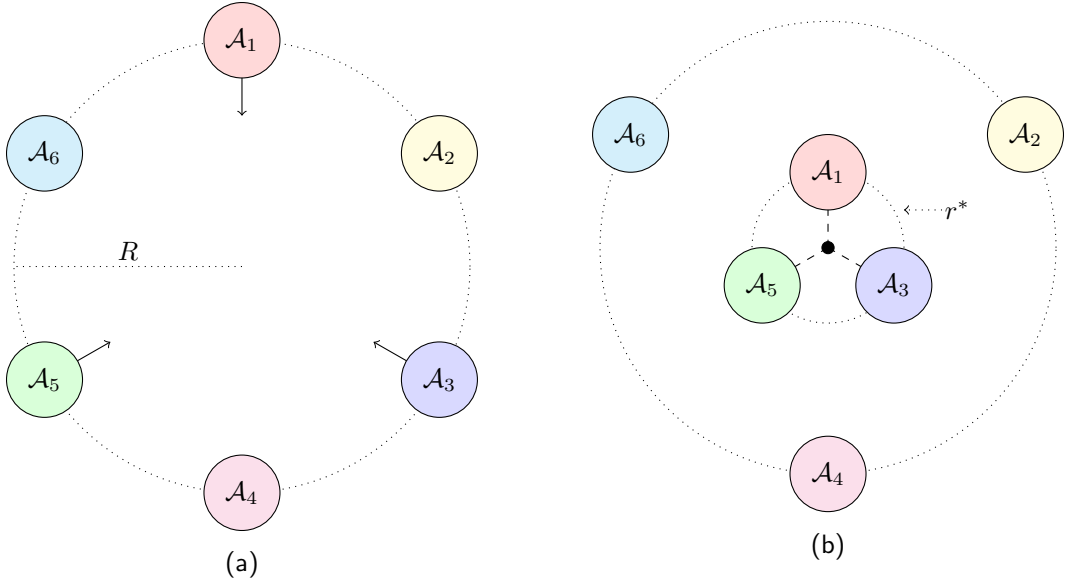
\begin{figure}[tb]
\centering
\begin{subfigure}{0.49\textwidth}
\centering
\begin{tikzpicture}
\draw[dotted] (0,0) circle (3cm);
\draw[dotted] (0,0) -- (-3,0);
\draw[->] (0,3) -- (0,2);
\draw[->] (-2.598,-1.5) -- (-1.732,-1);
\draw[->] (2.598,-1.5) -- (1.732,-1);
\draw[fill=red!15] (0,3) circle (0.5cm);
\draw[fill=green!15] (-2.598,-1.5) circle (0.5cm);
\draw[fill=blue!15] (2.598,-1.5) circle (0.5cm);
\draw[fill=yellow!15] (2.598,1.5) circle (0.5cm);
\draw[fill=cyan!15] (-2.598,1.5) circle (0.5cm);
\draw[fill=magenta!15] (0,-3) circle (0.5cm);
\node at (0,3) {\footnotesize{$\regA_1$}};
\node at (-2.598,-1.5) {\footnotesize{$\regA_5$}};
\node at (2.598,-1.5) {\footnotesize{$\regA_3$}};
\node at (2.598,1.5) {\footnotesize{$\regA_2$}};
\node at (-2.598,1.5) {\footnotesize{$\regA_6$}};
\node at (0,-3) {\footnotesize{$\regA_4$}};
\node at (-1.5,0.2) {\footnotesize{$R$}};
\end{tikzpicture}
\caption{}
\label{fig:building_blocks_a}
\end{subfigure}
\hfill
\begin{subfigure}{0.49\textwidth}
\centering
\begin{tikzpicture}
\draw[dotted] (0,0) circle (1cm);
\draw[dotted] (0,0) circle (3cm);
\draw[dotted,->] (1.5,0.5) -- (1,0.5);
\draw[fill=black] (0,0) circle (0.08cm);
\draw[dashed] (0,0) -- (0,1);
\draw[dashed] (0,0) -- (0.866,-0.5);
\draw[dashed] (0,0) -- (-0.866,-0.5);
\draw[fill=red!15] (0,1) circle (0.5cm);
\draw[fill=green!15] (-0.866,-0.5) circle (0.5cm);
\draw[fill=blue!15] (0.866,-0.5) circle (0.5cm);
\draw[fill=yellow!15] (2.598,1.5) circle (0.5cm);
\draw[fill=cyan!15] (-2.598,1.5) circle (0.5cm);
\draw[fill=magenta!15] (0,-3) circle (0.5cm);
\node at (0,1) {\footnotesize{$\regA_1$}};
\node at (-0.866,-0.5) {\footnotesize{$\regA_5$}};
\node at (0.866,-0.5) {\footnotesize{$\regA_3$}};
\node at (2.598,1.5) {\footnotesize{$\regA_2$}};
\node at (-2.598,1.5) {\footnotesize{$\regA_6$}};
\node at (0,-3) {\footnotesize{$\regA_4$}};
\node at (1.7,0.5) {\footnotesize{$r^*$}};
\end{tikzpicture}
\caption{}
\label{fig:building_blocks_b}
\end{subfigure}
\caption{Construction of a building block for ${\sf N}\!=\!6$. The special configuration $\c_6^{^\circ}$ is shown in (a). The arrows indicate the directions along which the disks $\regA_1,\regA_3,\regA_5$ are moved to construct the building block $\c_6^{^\circ}[\{1,3,5\}]$ shown in (b). The vertex and dashed lines in (b) represent the  octopoid surface which now connects the three disks.}
\label{fig:building_blocks}
\end{figure}

The combination of Lemma~\ref{lemma:canonical_form} and Lemma~\ref{lemma:building_blocks} imply that the equivalence classes of $\mathfrak{C}_{\sf N}$ are in one to one correspondence with the subsets $\mathfrak{F}\subseteq\mathfrak{F}^{\text{can}}$ which satisfy \eqref{eq:subsets}. As shown in the proof of Lemma~\ref{lemma:building_blocks}, a representative of each class can be constructed from the uncorrelated union of building blocks. Among the various classes we should now identify the ones that generate primitive information quantities. 
The following Lemma states that these correspond to the set $\mathfrak{F}^{\text{can}}$ with a single constraint removed.

\begin{lemma}
The equivalence classes of configurations in $\mathfrak{C}_{\sf N}$ which generate primitive information quantities are the ones which are associated to the following sets of constraints 
\begin{equation}
\mathfrak{F}^{\rm can}\setminus \f_{\widehat{\si}_{\sf n}}^{\rm can} ,\qquad {\rm for \ any} \quad \f_{\widehat{\si}_{\sf n}}^{\rm can}\notin\mathfrak{F}_{[{\sf N}]}
\label{eq:generating_equations}
\end{equation}	
where $1<{\sf n}\leq {\sf N}$.
\label{lemma:permutations}
\end{lemma}

\begin{proof}
There are $2^{\sf N}-1={\sf D}$ canonical form constraints and it is straightforward to check that they are linearly independent. Since we want to find the combinations of constraints with a one-dimensional subspace of solutions, we have to take all possible `consistent' collections of ${\sf D}-1$ constraints. To obtain such a collection we start from the full list $\mathfrak{F}^\text{can}$ and remove from it a single constraint $\f_{\widehat{\si}_{\sf n}}^{\text{can}}$. However, because of Lemma~\ref{lemma:building_blocks}, we have to require $\f_{\widehat{\si}_{\sf n}}^{\text{can}}\notin\mathfrak{F}_{[{\sf N}]}$ for consistency.
\end{proof}

Finally, to find the desired primitive information quantities we just need to solve these systems of equations. 

\begin{lemma}
The solution to the system of equations \eqref{eq:generating_equations}, with $\widehat{\si}_{\sf n}=\{\ell_1,\ell_2,...,\ell_{\sf n}\}$, is a subspace generated by ${\bf I}_{\sf n}(\regA_{\ell_1}:\regA_{\ell_2}:...:\regA_{\ell_{\sf n}})$.
\label{lemma:solution}
\end{lemma}

\begin{proof}
First consider the case ${\sf n}\!=\!{\sf N}$. 
Namely, we take the set of all canonical constraints up to the one involving all colors (so that our configuration $\c_{\sf N}$ does not admit $\omega$ corresponding to the ${\sf N}$-legged octopus).
The resulting system of equations contains, among all the others, the constraints $\f^{\,\text{can}}_{\mathscr{I}_{{\sf N}-1}}$. Explicitly these are 
\begin{equation}
\big\{Q_{\mathscr{I}_{{\sf N}-1}}+Q_{\mathscr{I}_{\sf N}}=0,\;\forall \mathscr{I}_{{\sf N}-1}\big\}\,.
\end{equation}
Setting  $Q_{\mathscr{I}_{\sf N}}\!\!=\!\lambda$ we get $\{Q_{\mathscr{I}_{{\sf N}-1}}\!\!=\!-\lambda,\;\forall \mathscr{I}_{{\sf N}-1}\}$. 
Next consider the constraints $\f^{\,\text{can}}_{\mathscr{I}_{{\sf N}-2}}$, and let us write them as
\begin{equation}
\big\{Q_{\mathscr{I}_{{\sf N}-2}}+\hspace{-0.7em}\sum_{\mathscr{I}_{{\sf N}-1}\supset\mathscr{I}_{{\sf N}-2}}\hspace{-0.7em} Q_{\mathscr{I}_{{\sf N}-1}}+Q_{\mathscr{I}_{\sf N}}=0,\;\forall \mathscr{I}_{{\sf N}-2}\big\}\,.
\end{equation}
Since there are $\binom{2}{1}\!=\!2$ terms in the sum, we get $\{Q_{\mathscr{I}_{{\sf N}-2}}\!\!=\!\lambda,\;\forall \mathscr{I}_{{\sf N}-2}\}$. Proceeding in this fashion, the $k$-th step is
\begin{equation}
\big\{Q_{\mathscr{I}_{{\sf N}-k}}+\hspace{-1em}\sum_{\mathscr{I}_{{\sf N}-k+1}\supset \mathscr{I}_{{\sf N}-k}}\hspace{-1em}Q_{\mathscr{I}_{{\sf N}-k+1}}+\hspace{-1.3em}\sum_{\mathscr{I}_{{\sf N}-k+2}\supset \mathscr{I}_{{\sf N}-k+1}}\hspace{-1.3em}Q_{\mathscr{I}_{{\sf N}-k+2}}+ \ldots+ Q_{\mathscr{I}_{{\sf N}-1}}+Q_{\mathscr{I}_{\sf N}}=0,\;\forall \mathscr{I}_{{\sf N}-k}\big\}\,,
\end{equation}
which reduces to
\begin{equation}
Q_{\mathscr{I}_{{\sf N}-k}}+\lambda\sum_{l=1}^k(-1)^{k-l}\binom{k}{l}=0
\end{equation}
and gives
\begin{equation}
Q_{\mathscr{I}_{{\sf N}-k}}=
\begin{cases}
      \lambda & \ k\; \text{even} \\
      -\lambda & \ k\; \text{odd}
    \end{cases}
\end{equation}
This procedure terminates when $k={\sf N}-1$, for which we get $Q_{\mathscr{I}_1}=\{-\lambda,\lambda\}$ depending on whether ${\sf N}$ is respectively even or odd. The resulting entropy relation is therefore ${\bf I}_{\sf N}$ (up to a possible factor of $-1$ which is irrelevant).

Consider now the general case of \eqref{eq:generating_equations}, with ${\sf n}<{\sf N}$. The constraint $\f_{\si_{\sf N}}^{\text{can}}$ implies $Q_{\si_{\sf N}} =0$, from which, using the constraints $\f_{\si_k}^{\text{can}}$ (with $k>{\sf n}$) and proceeding like above, one gets 
\begin{equation}
Q_{\si_k}=0,\quad \forall\si_k,\; \text{with}\;\,
{\sf n}<k\leq{\sf N}\,.
\end{equation}
Similarly one also gets
\begin{equation}
Q_{\si_{\sf n}}=0,\quad \forall\si_{\sf n}\neq\widehat{\si}_{\sf n}
\end{equation}
and
\begin{equation}
Q_{\si_k}=0,\quad \forall\si_k\not\subset\widehat{\si}_{\sf n},\;\,k<{\sf n}\,.
\end{equation}
These relations effectively implement a \textit{reduction} of the type discussed at the beginning of this section (see Eq.~\eqref{eq:reduction}), to a setting where there are effectively only ${\sf n}$ colors. Therefore one can run again the previous argument in this reduced setting, completing the proof.
\end{proof}

To summarize, we have proven that in an ${\sf N}$-party setting, any permutation of the ${\sf n}$-partite information is a primitive information quantity. Furthermore, the set of all quantities ${\bf I}^{(\sigma)}_{\sf n}$, for all permutations $\sigma$ and all values of $2\leq {\sf n}\leq {\sf N}$, is the full list of primitive quantities associated to the restricted class of configurations $\mathfrak{C}_{\sf N}$. To derive new primitive quantities, one therefore has to relax at least one of the two topological restrictions that we made at the beginning of this section.\footnote{ As we argued in \S\ref{sec:three_parties}, we expect that one needs at least ${\sf N}\geq 4$ to find new quantities.} Although this more complicated problem will be postponed to future work \citep{Hubeny:2018aa}, we will make further comments about how it can be approached in \S\ref{sec:discuss}. Before closing, we mention a simple result which goes in this direction and can be immediately derived from the construction presented in this section

\begin{corollary}
In a ${\sf N}$-party setting, any information quantity $\bQ$ derived from any primitive quantity ${\bf I}^{(\sigma)}_{\sf n}$ under the purification symmetry, is also primitive.
\label{cor:purifications}
\end{corollary}

\begin{proof}
For any primitive quantity ${\bf I}^{(\sigma)}_{\sf n}$ we have given an explicit construction, using the uncorrelated union of building blocks, of a generating configuration. As we exemplified in Fig.~\ref{fig:AL_via_purification}, for the derivation of the quantity $\bQ^{\text{AL}}$  associated to the AL inequality in the case of three parties, one can use the \textit{same} configuration to generate any other quantity obtained from ${\bf I}^{(\sigma)}_{\sf n}$ under the purification symmetry. To do that, one simply has to swap the role of the purifier $\univ$ and one of the monochromatic subsystems $\regA_\ell$.
\end{proof}

Note that the relabeling used to construct the configurations that generate the quantities obtained from  ${\bf I}^{(\sigma)}_{\sf n}$ under the purification symmetry (see proof above), necessarily imply that the new configurations (after relabeling) do not satisfy the topological restrictions that define $\mathfrak{C}_{\sf N}$. More precisely, since the purifier $\univ$ is necessarily adjoining to all the regions composing the configurations in $\mathfrak{C}_{\sf N}$, after relabeling some of the regions with different colors will be adjoining to each other. The result of Corollary~\ref{cor:purifications} is therefore consistent with the ${\bf I}_{\sf n}$-Theorem.

\section{Discussion}
\label{sec:discuss}

The main goal of the present work was to introduce a new framework for the derivation of information quantities which are natural from the perspective of geometric states in holographic theories. We have shown that in the case of three parties, the information quantities which emerge from this framework precisely correspond to the facets of the holographic entropy cone \cite{Bao:2015bfa}. Furthermore, we have proved a general theorem about a particular family of information quantities which can be derived, for an arbitrary number of parties ${\sf N}$, under some topological restrictions for the allowed choice of configurations in the field theory. This result should be understood as the first step of a broader program \cite{Hubeny:2018aa} which aims at a deeper understanding of the entanglement structure of geometric states.  For the present, we will  explain the relation of our work with that of \cite{Bao:2015bfa}  and comment on other outstanding questions that will be addressed in future publications.

\begin{figure}[tb]
\centering
\begin{tikzpicture}
\draw[fill=yellow!15] (0,0) circle (4cm);
\draw[fill=white!15] (0,2) circle (1.5cm);
\draw[fill=white!15] (1.732,-1) circle (1.5cm);
\draw[fill=white!15] (-1.732,-1) circle (1.5cm);
\draw[fill=red!15] (0,2) circle (1cm);
\draw[fill=green!15] (1.732,-1) circle (1cm);
\draw[fill=blue!15] (-1.732,-1) circle (1cm);
\draw[dashed] (0,3) -- (0,3.5);
\draw[dashed] (2.598,-1.5) -- (3.031,-1.75);
\draw[dashed] (-2.598,-1.5) -- (-3.031,-1.75);
\draw[dashed] (0,0) -- (0,1);{}
\draw[dashed] (0,0) -- (-0.866,-0.5);
\draw[dashed] (0,0) -- (0.866,-0.5);
\draw[fill=black] (0,0) circle (0.08cm);
\node at (0,2) {\small{$\mathcal{A}$}};
\node at (1.732,-1) {\small{$\mathcal{B}$}};
\node at (-1.732,-1) {\small{$\mathcal{C}$}};
\node at (-2.5,1.5) {\small{$\mathcal{D}$}};
\end{tikzpicture}
\caption{Example of a four-party configuration whose constraints do not reduce to Eq.~\eqref{eq:canonical_form_constaints}. The dashed lines indicate non-vanishing mutual information, namely ${\bf I}_2(\mathcal{A}:\mathcal{D})>0$ and  ${\bf I}_2(\mathcal{A}:\mathcal{B}\mathcal{C})>0$ while  ${\bf I}_2(\mathcal{A}:\mathcal{B})=0$, etc.}
\label{fig:4_parties}
\end{figure}
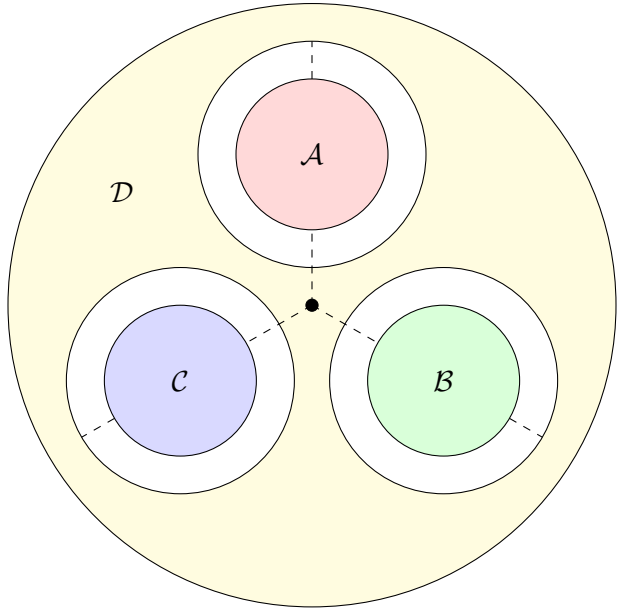

\paragraph{Beyond the ${\bf I}_{\sf n}$-Theorem:} The central result of our analysis is  Theorem~\ref{thm:In} which was proved under two restrictions on the allowed field theory configurations, namely that all regions are disjoint from each other (the disjoint scenario discussed in \S\ref{subsec:organizing}), and that each region has a connected boundary. While these assumptions were sufficient to prove the theorem, they are far from being necessary. In fact, one can show that the theorem holds much more generally. For example, we note without explicit proof, that the theorem continues to hold even when the regions have disconnected boundaries, as long as the regions are not \textit{enveloping} each other (see \S\ref{subsec:organizing} for the definition). The reason is that Lemma~\ref{lemma:canonical_form} continues to hold, and the constraints of any configuration, even in this more general family, can always be converted to the same canonical form Eq.~\eqref{eq:canonical_form_constaints}. Even more generally, our investigations show that there are plenty of configurations with enveloped regions whose constraints can nevertheless be converted to the canonical form Eq.~\eqref{eq:canonical_form_constaints}. Similarly, lifting our restrictions and allowing regions which share parts of their boundaries also turns out not to guarantee new information quantities (this is already clear from the first configuration that we  used to derive the tripartite information, see Fig.~\ref{fig:tripartite_information1}). 

To find new information quantities, one has to look for particular configurations whose corresponding sets of constraints cannot be reduced to the canonical form of Eq.~\eqref{eq:canonical_form_constaints}. Because of the general validity of  Theorem \ref{thm:In} beyond the restricted class of configurations we considered in its proof, we expect these configurations to be quite rare. For instance,  Fig.~\ref{fig:4_parties} provides a particular example involving four parties where we find a different canonical representative for the constraints. At this stage, it is unclear how the number of primitive information quantities scales with the number of parties ${\sf N}$. However,  the fine tuning required by configurations like the one in Fig.~\ref{fig:4_parties} seems to suggest that the number of possibilities could remain under control, at least for reasonably small values of ${\sf N}$. Furthermore, for any given ${\sf N}$, the primitive information quantities found by the ${\bf I}_{\sf n}$-Theorem are the ($\sigma,{\sf n}$)-reductions of the multipartite information (for all ${\sf n}$ with $2\leq {\sf n}<{\sf N}$). Let us define a \textit{genuine} ${\sf N}$-partite primitive information quantity $\bQ$ as one where all the ${\sf N}$ monochromatic subsystems appear in $\bQ$.\footnote{
More precisely, suppose that we are working in an ${\sf N}$-party setting. The quantities ${\bf I}^{(\sigma)}_{\sf n}$, with ${\sf n}<{\sf N}$, which result from the ${\bf I}_{\sf n}$-Theorem, clearly are non-genuine according to the above definition, since they only contain ${\sf n}$ monochromatic subsystems. But note that the other quantities obtained from them via Corollary~\ref{cor:purifications} can in general contain all the ${\sf N}$ subsystems and naively would seem to be genuine. However, this is an artifact of the definition of the purification symmetry, where we decided to hold fixed the total number of subsystems and never allow to join them (see also the comments about the various instance of SA and AL for three parties in \S\ref{sec:overview}). Hence, we define a  genuine information quantity as one which cannot be reduced to a simpler form (with a smaller number ${\sf n}$ of subsystems) by recombining the ${\sf N}$ appearing subsystems into new (joined) ones.} One can infer from the ${\bf I}_{\sf n}$-Theorem that (under the aforementioned restrictions) only a single new primitive information quantity emerges for any ${\sf N}$ (the others being instances of lower party quantities). This observation supports the intuition that the problem 
of finding primitive information quantities could be solved efficiently for any value of ${\sf N}$.  

We will continue the systematic search of new primitive information quantities, for four or more parties, in future publications. Our search strategy will broadly mimic the proof of Theorem~\ref{thm:In}, by examining new classes of configurations with increasing levels of topological complexity. More precisely, one first defines a set $\mathfrak{C}_{\sf N}$ of allowed configurations, by specifying some topological restriction. To classify the equivalence classes of configurations in $\mathfrak{C}_{\sf N}$, under the equivalence relation Eq.~\eqref{eq:equivalence3}, one has to introduce an appropriately generalized version of the canonical form for the constraints. The equivalence classes are then derived by explicitly constructing their representatives via the disjoint union of suitably defined building blocks. Finally, one considers the classes associated to a set of ${\sf D}-1$ linearly independent constraints to derive the primitive information quantities. 

As we discussed in \S\ref{subsec:organizing}, and exemplified in \S\ref{sec:three_parties} for the case of three parties, the pattern of constraints associated to various configurations is more transparent if one assumes that the regions do not share any portions of their boundaries. Indeed, the generalization of the ${\bf I}_{\sf n}$-Theorem to arbitrary configurations in this disjoint scenario will be the focus of \cite{Hubeny:2018aa}. We then hope to understand more complicated situations, where the regions do share portions of their boundaries, as particular limits of the aforementioned ones. Ultimately one would like to lift all topological restrictions in the definition of $\mathfrak{C}_{\sf N}$ and allow for completely general configurations.

\paragraph{Entropy hyperplane arrangement:} For given ${\sf N}$, consider the set of all primitive information quantities. Geometrically, we can visualize each quantity as being associated to a hyperplane in entropy space.\footnote{ More precisely, for a quantity ${\bf Q}$, we consider the particular hyperplane which is the space of solutions to the equation ${\bf Q}({\bf S})=0$, where the components $S_\si$ of ${\bf S}$ are now treated as variables.} Thus, the set of all of them is a \textit{hyperplane arrangement} that we will denote by $\arr_{\sf N}$. As we mentioned in the introduction, the main motivation for the program initiated in this work is the understanding of the structure of correlations which characterizes geometric states in holographic theories. The hope is that we can unearth its main features via the derivation of the primitive quantities which compose the arrangement, as well as of new holographic inequalities (more on this below). In particular, we expect that some distinctive properties of the entanglement structure of geometric states can become more evident via the study of the structure of the arrangement. 

First, notice that for any $\bQ\!\in\!\arr_{\sf N}$, all the other quantities obtained from $\bQ$ via the purification symmetry also belong to $\arr_{\sf N}$. This simply follows from the fact that all these quantities are generated by the same configuration by changing the labeling\footnote{ Corollary~\ref{cor:purifications} in \S\ref{sec:multipartite_information} is just a particular instance of this general fact.} (see Fig.~\ref{fig:AL_via_purification} for an example).  One is then naturally led to inquiring about the topology of the boundaries of the generating configurations, which in turn informs whether the mutual information between components is finite. 

Consider, for example,  the information quantity associated to AL (for both two and three parties), which is obtained from SA under the symmetry. It appears that there is no way to generate such quantities without considering regions that share some boundaries (and therefore having a divergent mutual information between components). 
In case of AL, this is necessitated by the fact that AL is not balanced, so that we need coincident entangling surfaces to effectively restore balance and cancel the divergence in the corresponding information quantity.
One might nevertheless wonder if a similar property also holds for other quantities. Specifically, we can ask if it is the case that for any primitive quantity $\bQ$, generated by a configuration where all mutual informations are finite (like SA above), all the other quantities obtained from $\bQ$ under the purification symmetry, can only be generated by configurations with divergent mutual information.

This turns out to be not necessarily the case. Indeed, for four parties, there is an instance of MMI which is of the form
\begin{equation}
S_\mathcal{ABD}+S_\mathcal{BC}+S_\mathcal{ACD}\geq S_\mathcal{AD}+S_\mathcal{B}+S_\mathcal{C}+S_\mathcal{ABCD} \,.
\label{eq:mmi_4}
\end{equation}
It is immediate to check that it can be derived from \eqref{eq:mmidef} using the purification symmetry. However, as we will show in detail in \cite{Hubeny:2018ab}, by allowing enveloping of regions, one can generate the quantity associated to \eqref{eq:mmi_4} with a pattern of finite mutual information. 

A second related question concerns the various instances of a given ${\sf N}$-partite quantity $\bQ$ appearing for a larger number of parties ${\sf N}' \!>\! {\sf N}$. It would be interesting to understand if there is some fundamental principle that determines which ${\sf N}'$-instances of $\bQ$ are primitive and which are not. For instance, as exemplified in \S\ref{sec:overview}, the instances of SA for ${\sf N}' >2$ can be non-primitive, as a consequence of MMI.  On the other hand, the previous example involving MMI in the form \eqref{eq:mmi_4} illustrates that some instances of ${\bf I}_3$ can be primitive (albeit not genuine) for ${\sf N}' >3$. In the $4$-party case, this is related to the fact that there are no new inequalities that could render \eqref{eq:mmi_4} redundant. In general, whether ${\sf N}'$-instances of a ${\sf N}$-partite quantity are primitive or not, is intimately related to the presence of new inequalities. Therefore, answering this question could be helpful in the search for new inequalities. Furthermore, having a general principle to establish which instances remain primitive at increasing number of parties could be useful to make more efficient the construction of  $\arr_{{\sf N}'}$, assuming knowledge of all $\arr_{\sf N}$ with ${\sf N}<{\sf N}'$.

Finally, a comment is in order about the relation between the structure of the arrangement and how a pair $(\c_{\sf N}, \psi_\Sigma)$ of a state and a configurations can be ``localized'' in entropy space. By construction, the primitive quantities are associated to configurations whose entropy vector can be `minimally localized',  in the sense that the regulated entropy vectors belong to a particular hyperplane (but not simultaneously to any other). Other configurations might be further localized. They could, for instance, satisfy more than a single relation, i.e., more than one primitive information quantity vanishes and the configuration belongs to a subspace of higher codimension (the intersection of various hyperplanes). It is interesting to inquire which configurations can be localized at particular locations on the arrangement. For sufficiently generic configurations (see \S\ref{sec:overview}), the codimension of the subspaces on which they are localized increases as the number of constraints decreases. A relevant case is the ${\sf N}$-party ``completely uncorrelated'' configuration, which can be realized, for example, by a set of ${\sf N}$ disks which are sufficiently separated from each other. As we discussed in \S\ref{sec:multipartite_information}, this particular configuration is associated to the set of (canonical) constraints $\mathfrak{F}_{[{\sf N}]}$. The solution to this set of constraints is a special ${\sf N}$-dimensional subspace of entropy space where all the hyperplanes associated to balanced information quantities intersect. This follows from the fact that this particular configuration belongs to any such hyperplane (see \S\ref{sec:overview}). 
This simple observation suggests that as we deform a configuration to reduce the number of constraints, the various regions become less and less correlated (see more on this point below), and the configurations are more and more localized in entropy space. However, this seems to indicate that only certain locations on the arrangement can be reached by this procedure. To find configurations which are localized on other specific locations on the arrangement, one might have to work with finite values of the entropy (and not with the proto-entropy).

\paragraph{Relation to the holographic entropy cone:} Clearly, not all primitive information quantities are associated to new holographic inequalities. It was already observed in \cite{Hayden:2011ag} that there are counterexamples to the sign-definiteness of the ${\sf N}$-partite information, at least for some values of ${\sf N}\geq 4$. In fact, since the holographic entropy cone \cite{Bao:2015bfa} for four parties is known, one can immediately check that the $4$-partite information ${\bf I}_4$ has opposite signs for some of the extremal rays of the cone. 

One of the main applications of our framework is the search for new holographic entropy inequalities. Indeed, this is what motivated the definition of primitive information quantities in the first place. Suppose that for a given ${\sf N}$ we have a list of primitive information quantities generated by the procedure outlined above. To find good candidates for new holographic inequalities we need to construct a \textit{sieve} which allows us to efficiently extract the candidates from the list. We will explain how this can be done in \cite{Hubeny:2018ab}.

It is clear that, strictly speaking, this procedure alone does not prove that the candidates are indeed true inequalities, in full generality, for any state (including dynamical ones) and choice of configuration. For static spacetimes, a direct proof of the inequalities via the standard `cutting and pasting' procedure of \cite{Headrick:2007km} quickly becomes unfeasible as ${\sf N}$ grows \cite{Bao:2015bfa}. The situation is even more dire in the dynamical case, where such technique cannot be employed. In fact, even if MMI was proven also in the dynamical case using the `maximin technique' of \cite{Wall:2012uf}, it was shown in \cite{Rota:2017ubr} that this method cannot be extended to the 5-party case.\footnote{ Furthermore, it was argued in \cite{Rota:2017ubr} that for more parties, the technique is even more unlikely to be useful for proving any inequality.} A more promising technique could be the one introduced in \cite{Hubeny:2018bri} to prove MMI using bit-threads \cite{Freedman:2016zud}. However this remains to be explored further.

On the other hand, one could also hope to be able to prove the inequalities via a more indirect argument. If one could prove that any holographic inequality is necessarily associated to a primitive information quantity, one could dispense with a direct proof technique and instead try to optimize the sieve. An indication that this might be possible already comes from \cite{Bao:2015bfa}, which showed that the holographic entropy cone is closed topologically (see \S\ref{sec:overview} above and the original paper for more details). In fact, this implies that given the geometries which realize the extremal rays of the cone, one can realize any other ray by an appropriate choice of tensor products and rescaling of the metric. In particular, one can realize a ray which is minimally localized on a facet of the cone (which corresponds to the saturation of one of the inequalities), i.e., localized only on the facet but not on a lower dimensional face. 
Heuristically, such a ray would then be associated to a configuration that generates the corresponding primitive quantity according to Definition~\ref{con:fun}. 

However, there is a potential subtlety, in that it is in principle possible that some information quantity can only be saturated  for finite values of the entropies (either obtained by fixing a cut-off or by choosing entire boundaries as subsystems) and not in the more abstract sense of our framework (proto-entropy). In other words, the existence of an entropy vector defined as in \cite{Bao:2015bfa} which saturates a facet inequality, does not necessarily guarantee that the corresponding information quantity is primitive in the sense of  Definition \ref{con:fun}. Additionally, this argument relies on the polyhedrality of the holographic entropy cone, but this was proven only for static spacetimes.\footnote{ As we explained in \S\ref{sec:formalization}, our construction makes no distinction between RT/HRT. In particular the full arrangement $\arr_{\sf N}$ is derived for \textit{all} geometric states, both static and dynamical and it is only a function of the number of parties ${\sf N}$. We find this to be indicative that also the holographic entropy cone could be the same, for static and dynamical geometries.} While we believe that some of these situations do not come to pass, it behooves us to explore these questions in greater detail as we develop our framework further.

The main purpose for considering the description of the holographic cone in terms of its generators, was that it allows us to directly check the completeness of the set of inequalities. Suppose that a set of inequalities has been proved for some ${\sf N}$, but one does not know how many more there are. One can consider the cone specified by the proven inequalities, extract the extremal rays and try to realize them by a certain geometry/configurations. This is in fact how \citep{Bao:2015bfa} proved that there are no new inequalities other than MMI for three and four parties. Of course, should one fail to realize such rays, one does not know if it is because other inequalities exist and have to be found, or if it just because of the complexity of the geometric problem.

Extremal rays however could also be interesting for another reason. As we explained in the preceding paragraphs, there are locations on the arrangement on which one could only hope to localize configurations by working with finite entropies. The extremal rays are likely to be important examples of such locations\footnote{ The bipartite case, being particularly simple, is an exception, see \S\ref{sec:overview}.} and it could be useful to further explore the corresponding entanglement structure. One option would be to introduce a regulator, but this is in general unphysical, as we explained in \S\ref{sec:overview}. The other option, adopted by \cite{Bao:2015bfa} is to realize the extremal rays with multi-boundary wormhole solutions. It is not fully clear if such solutions are indeed dual to field theory states (see also \cite{Marolf:2017shp,Marolf:2018ldl}), however it would be very interesting to understand if there are particular patterns of correlations which can only be realized by field theory states with a non-trivial bulk topology.\footnote{ Here we are imagining to work with finite values of the entropy, and not with the proto-entropy. Therefore, the actual value of the mutual information between regions within a configuration matters. In particular, there is no contradiction with the arguments of \S\ref{sec:formalization} and \S\ref{sec:multipartite_information}, where we only needed to know whether the mutual information was vanishing or not.}

\paragraph{The exceptional case of $(1+1)$-dimensional CFTs:} The astute reader will note that our gauge fixing procedure, illustrated in \S\ref{subsec:redundancy}, narrowed down our focus to scanning over boundary regions in the vacuum of a $(2+1)$-dimensional CFT, which holographically would be dual to an AdS$_4$ spacetime. Given that a vast amount of holographic entanglement entropy literature focuses on the simpler case of $(1+1)$-dimensional CFT and AdS$_3$ dynamics, should we have not further simplified to this case, one may naturally wonder. In this instance (and in fact with many other explorations in the subject), the $(1+1)$-dimensional case happens to be misleading owing to some over-simplifications, so conclusions drawn from here may not hold more generally.\footnote{ The issue is simple: the absence of non-trivial gravitational dynamics which is the reason for focusing on this case, also ends up being the bane of the analysis. Features of extremal surfaces that are generic to AdS$_3$ are non-generic in higher dimensions, potentially invalidating many conclusions drawn from the low-dimensional example.} 

More specifically, it is not fully clear, for example, if the technology that we employed in \S\ref{sec:multipartite_information} to prove the ${\bf I}_{\sf n}$-Theorem could be applied in a straightforward manner also in this more special setting. To illustrate the point, consider a configuration $\c_{\sf N}$ which would comprise of a collection of ${\sf N}$ intervals $\regA_\ell$, each one with a different color. The intervals can be arbitrarily ordered,  arbitrarily distanced from each other, and can have different length, but a key feature of the holographic scenario is that the mutual information vanishes quite easily when the the distance between the intervals increases. This feature, which is helpful in the analysis of bipartite systems (as the reader will notice we implicitly used this in \S\ref{sec:overview}), is a hassle in the case of multipartite systems. It is easy to convince oneself that, if one is only allowed to vary the size of the intervals and the distance between them, it is not possible to find configurations, like the one mentioned above, which realize arbitrary patterns of correlations. This should be contrasted with the higher dimensional case, where, as we discussed in \S\ref{subsec:redundancy}, one has instead the freedom to deform the shape of the regions (now disks), to realize any desired pattern of mutual information. Thus one is severely constrained in the class of configurations available to us in $(1+1)$-dimensions. Specifically, to try and circumvent this limitation, one is forced to consider more complicated configurations, where each monochromatic subsystem is composed of multiple intervals. This implies that in the configurations one should consider, intervals of different colors are inevitably enveloping each other (see \S\ref{subsec:organizing} for a definition), which can in turn obfuscate the pattern of constraints.

There is an equivalent way to see what the issue is from an information theoretic point of view. In $(1+1)$-dimensional CFT the limitation we just mentioned seems to imply that generically all the holographic inequalities collapse down to SSA. For instance, since it is hard to have three disjoint intervals with non-vanishing common mutual information, in the $3$-party case, we
could ensure that ${\bf I}_2(\mathcal{A}:\mathcal{B}) \neq 0$ and ${\bf I}_2(\mathcal{B}:\mathcal{C}) \neq 0$, but we easily end up having ${\bf I}_2(\mathcal{A}:\mathcal{C}) =0$ (all we need is to order the regions sequentially and have $\mathcal{A}$ and $\mathcal{C}$ be further away from each other). If this is case, it is straightforward to check that the tripartite information ${\bf I}_3(\mathcal{A}:\mathcal{B}:\mathcal{C})$ reduces to the conditional mutual information ${\bf I}_2(\mathcal{A}:\mathcal{C}|\mathcal{B})$. What this means is that in this particular situation MMI is already implied by SSA and does not contain any new information!  More generally, this logic seems to indicate that the same trivialization characterizes all the other holographic inequalities of \cite{Bao:2015bfa}. In other words, generically in $(1+1)$-dimensional holographic CFTs, the information quantities and inequalities constraining the holographic entropy cone are trivially implied by SSA for a single copy of the CFT.\footnote{ We thank Xi Dong for important discussions on this issue.} Showing that all the relations actually amount to no more than SSA involves a more detailed analysis which we will not undertake here (see \cite{Czech:2018aa} for progress in this direction). All we wish to illustrate here is that it is possible to be misled into thinking that one is deriving new relations owing to the somewhat degenerate situation in $(1+1)$-dimensions.

\paragraph{Interpretation of primitive information quantities:} Ultimately we would like to understand the implications of the properties of the arrangement and the entropy cone for the entanglement structure of geometric states. Although at this stage this is still not clear, it is worthwhile to extract some preliminary observations inspired by the derivation thus far.

For bipartite systems, the saturation of SA implies that the density matrix factorizes. Similarly, at least for finite dimensional Hilbert spaces, a particular form of factorization of the density matrix is also implied by the saturation of AL \cite{Zhang:2011aa}. Specifically, if the second inequality of Eq.~\eqref{eq:al} is saturated, there exists a bipartition of $\mathcal{A}$ into two subsystems $\mathcal{A}_1,\mathcal{A}_2$ such that
\begin{equation}
\rho_\mathcal{AB}=\ket{\psi}_{\mathcal{A}_1\mathcal{B}}\bra{\psi}_{\mathcal{A}_1\mathcal{B}}\otimes\rho_{\mathcal{A}_2}\,,
\label{eq:saturation_AL}
\end{equation}
where $\ket{\psi}_{\mathcal{A}_1\mathcal{B}}$ is a pure state. In the holographic context one can easily see why this must be the case. Consider the configuration of Fig.~\ref{fig:cut-off_independence2} and call $\univ$ the purifier of the bipartite system $\mathcal{AB}$. It is clear that one has $I_2(\mathcal{B}:\univ)=0$, which implies $\rho_\mathcal{BO}=\rho_\mathcal{B}\otimes\rho_\univ$. Since the state on $\mathcal{ABO}$ is pure, one immediately arrives at the factorization given in \eqref{eq:saturation_AL} (for other observations regarding the saturation of AL in the holographic context see  \cite{Hubeny:2013gta,Headrick:2013zda}). 

More interestingly, in the case of three parties, the configuration of Fig.~\ref{fig:alternative_tripartite_info} which we used to generate the tripartite information ${\bf I}_3$, also implies a factorization of the density matrix, now of the form
\begin{equation}
\rho_{\mathcal{ABC}}=\rho_{\mathcal{A}_1\mathcal{B}_1}\otimes\rho_{\mathcal{B}_2\mathcal{C}_1}\otimes\rho_{\mathcal{A}_2\mathcal{C}_2}\,.
\label{eq:saturation_MMI}
\end{equation}
Likewise, despite naively looking different, the other configuration which generates ${\bf I}_3$ (Fig.~\ref{fig:tripartite_information1}) is also associated to a density matrix which reduces to a similar form (up to permutation of the labels, and additional factors involving portions of single subsystems). To see this, one can simply use the structure \eqref{eq:saturation_AL} for each ``island'' in the configuration. More generally, the ${\bf I}_{\sf n}$-Theorem shows that for any ${\sf N}>3$ one can derive the tripartite information ${\bf I}_3^{(\sigma)}$, now a ($\sigma,{\sf n}$)-reduction, by removing from the full list of constraints  $\mathfrak{F}^{\text{can}}$ the one which is associated to an octopus connecting three monochromatic subsystems $\f^{\text{can}}_{\widehat{\si}_3}$ (see Lemma \ref{lemma:permutations}). The presence of surfaces of higher degree now implies that the density matrix does not have the structure \eqref{eq:saturation_MMI}. Nevertheless, the structure of the building blocks ensures that the factor containing ${\sf N}$-partite correlations still do not contain 3-partite correlations, i.e., all the marginals completely factorize. For example, in the $4$-party case, one gets (schematically)
\begin{equation}
\rho_{\mathcal{ABCD}}=\rho_{\mathcal{A}_1\mathcal{B}_1\mathcal{C}_1\mathcal{D}_1}\otimes\rho_{\text{bipartite}}\,.
\end{equation}
This seems to suggest that, holographically, the tripartite information is a measure of genuine tripartite correlation. 

The previous argument naturally generalizes to the ${\sf N}$-partite information, for arbitrary ${\sf N}$. Therefore, all the primitive information quantities found so far have a natural (from the holographic perspective) saturating density matrix with a tensor product structure. It is tempting to speculate that this is a general feature of all primitive information quantities. Indeed one might wonder if the above structure is also necessary, at least in the holographic context, for saturation. Similar properties would then be inherited by other more special locations of the arrangement. Developing this intuition further within our framework, and in particular making a connection with the conjecture of \cite{Cui:2018dyq} about the structure of geometric states, is a very interesting question that we leave for future investigations.\footnote{ However, we warn the reader that these statements should be understood as approximate, since a-priori one does not expect these relations to hold exactly if $1/N$ corrections are included in the evaluation of the various entropies.}

We conclude with a general comment that transcends the holographic context. While we have defined the arrangement $\arr_{\sf N}$ using the RT/HRT prescription, this object could be of interest in quantum field theory more broadly. In fact, while the primitive information quantities were identified using purely holographic arguments, as we discussed in \S\ref{subsec:overview2b}, their particular structure seems to suggest that, like the mutual information, they can be well defined measures of correlations. Specifically, they can be finite (at least when the subsystems do not share portions of their boundaries) and independent from the regulator scheme. If this were the case, one could use these quantities, from which one could now extract physically relevant information, to further localize configurations of subsystems in entropy space. One could see this procedure as a way to meaningfully characterize the multipartite correlation structure of field theory states, for arbitrary relativistic QFTs.

\section*{Acknowledgments}

It is a pleasure to thank Ning Bao, Xi Dong, Don Marolf, Bogdan Stoica and Sean J. Weinberg for useful conversations. We thank the Centro de Ciencias de Benasque Pedro Pascual, the Kavli Institute for Theoretical Physics in Santa Barbara, the Centro Atomico Bariloche, and the Galileo Galilei Institute in Florence for hospitality during various stages of this project. 	
V.~Hubeny and M.~Rangamani  would also like to acknowledge the hospitality of ICTS-TIFR, Bengaluru and the Yukawa Institute for Theoretical Physics at Kyoto University, for hospitality during the course of the workshops “20 years of AdS/CFT and beyond” and  YITP-T-18-04 ``New Frontiers in String Theory 2018", respectively. M.~Rota would also like to thank QMAP at University of California Davis, the University College London and Nordita in Stockholm, during the workshop ``Cosmology and Gravitational Physics with Lambda'', for hospitality while this work was in progress. 

V.~Hubeny and M.~Rangamani were supported by  U.S. Department of Energy grant DE-SC0009999 and by funds from the University of California. M.~Rota is supported by the Simons Foundation via the ``It from Qubit'' collaboration and by funds from the University of California.


\providecommand{\href}[2]{#2}\begingroup\raggedright\endgroup

\end{document}